\documentclass{lmcs}
\pdfoutput=1

\usepackage{lastpage}
\lmcsdoi{15}{1}{4}
\lmcsheading{}{\pageref{LastPage}}{}{}%
{Feb.~27,~2018}{Jun.~14,~2022}{}


\usepackage[utf8]{inputenc}
\usepackage{amsmath, amsthm, amssymb, mathtools, stmaryrd}
\usepackage{tikz}
\usepackage{xspace}
\usepackage{listings}
\usepackage{algorithm}
\usepackage{wrapfig}
\usepackage{color}
\usepackage{algorithmic}

\makeatletter
\newcommand{\@abbrev}[3]{
  \def\c@a@def##1{
      \if ##1.
        \relax
      \else
        \@ifdefinable{\@nameuse{#1##1}}{\@namedef{#1##1}{#2##1}}
        \expandafter\c@a@def
      \fi
    }
  \c@a@def #3.
}
\@abbrev{bb}{\mathbb}{ABCDEFGHIJKLMNOPQRSTUVWXYZ}
\@abbrev{bf}{\mathbf}{ABCDEFGHIJKLMNOPQRSTUVWXYZabcdefghijklmnopqrstuvwxyz}
\@abbrev{bit}{\boldsymbol}{
ABCDEFGHIJKLMNOPQRSTUVWXYZabcdefghijklmnopqrstuvwxyz}
\@abbrev{mc}{\mathcal}{ABCDEFGHIJKLMNOPQRSTUVWXYZ}
\@abbrev{mb}{\mathbb}{ABCDEFGHIJKLMNOPQRSTUVWXYZ}
\@abbrev{mf}{\mathfrak}{
ABCDEFGHIJKLMNOPQRSTUVWXYZabcdefghijklmnopqrstuvwxyz}
\@abbrev{rm}{\mathrm}{ABCDEFGHIJKLMNOPQRSTUVWXYZabcdefghijklmnopqrstuvwxyz}
\@abbrev{scr}{\mathscr}{
ABCDEFGHIJKLMNOPQRSTUVWXYZabcdefghijklmnopqrstuvwxyz}
\@abbrev{sf}{\mathsf}{ABCDEFGHIJKLMNOPQRSTUVWXYZabcdefghijklmnopqrstuvwxyz}
\@abbrev{b}{\bar}{abcdefghijklmnopqrstuvwxyz}
\makeatother

\newcommand{\NP}{\ensuremath{\textsc{NP}}}
\newcommand{\PTIME}{\ensuremath{\textsc{Ptime}}}

\newcommand{\FO}{\ensuremath{\textsc{FO}}}
\newcommand{\LFP}{\ensuremath{\textsc{LFP}}}
\newcommand{\FOnum}{\ensuremath{\textsc{FO}^+}}

\newcommand{\EFP}{\ensuremath{\textsc{EFP}}}
\newcommand{\FPC}{\ensuremath{\textsc{FPC}}}
\newcommand{\FOTC}{\ensuremath{\textsc{FO(TC)}}}

\newcommand{\Prop}{\ensuremath{\textsc{Prop}}}
\newcommand{\Res}{\ensuremath{\textsc{Res}}}
\newcommand{\kRes}{\ensuremath{k\textsc{-Res}}}
\newcommand{\HRes}{\ensuremath{\textsc{Horn-Res}}}
\newcommand{\kResx}[1]{\ensuremath{#1\textsc{-Res}}}
\newcommand{\PC}{\ensuremath{\textsc{PC}}}
\newcommand{\monPC}{\ensuremath{\textsc{mon-PC}}}
\newcommand{\monPCx}[1]{\ensuremath{\textsc{mon-PC}_{#1}}}
\newcommand{\PCx}[1]{\ensuremath{\textsc{PC}_{#1}}}

\newcommand{\classmonPCx}[1]{\ensuremath{\mathbf{\mathcal{K}_{\monPCx 
k}}}}

\newcommand{\Primes}{\ensuremath{\mbP}}

\newcommand{\DeriveMonPCx}[1]{\ensuremath{\text{MonPC}_{#1}}}

\newcommand{\DerivePCx}[1]{\ensuremath{\text{PC}_{#1}}}
\newcommand{\DerivePC}{\DerivePCx k}

\newcommand{\Q}[1]{\ensuremath{\mathcal{Q}_{#1}}}
\newcommand{\I}{\ensuremath{\mathcal{I}}}

\DeclareMathOperator{\lfp}{\mathbf{lfp}}

\DeclareMathOperator{\ifp}{\mathbf{ifp}}

\DeclareMathOperator{\dom}{\ensuremath{\text{dom}}}
\DeclareMathOperator{\Part}{\ensuremath{\text{Part}}}
\DeclareMathOperator{\Aut}{\ensuremath{\text{Aut}}}
\DeclareMathOperator{\im}{\ensuremath{\text{im}}}
\DeclareMathOperator{\kernel}{\ensuremath{\text{ker}}}
\DeclareMathOperator{\rank}{\ensuremath{\text{rk}}}
\DeclareMathOperator{\minimum}{\ensuremath{\text{min}}}
\DeclareMathOperator{\Hom}{\ensuremath{\text{Hom}}}

\renewcommand{\phi}{\varphi}

\newcommand{\look}{\begin{proof}}
\newcommand{\lueg}{\look}
\newcommand{\hx}{\end{proof}}

\newcommand{\ra}{\rightarrow}

\newcommand{\E}{\exists}
\newcommand{\A}{\forall}
\renewcommand{\phi}{\varphi}
\renewcommand{\theta}{\vartheta}

\renewcommand{\AA}{{\mathfrak A}}
\newcommand{\BB}{{\mathfrak B}}

\renewcommand{\epsilon}{\varepsilon}

\newcommand{\posLFP}{{\rm posLFP}}

\newcommand{\IFP}{{\rm IFP}}

\newcommand{\ML}{{\rm MultLin}}

\newcommand{\ptime}{\mbox{\sc Ptime}}

\newcommand{\Linf}{\ensuremath{\textup{L}^{\omega}_{\infty\omega}}}
\newcommand{\Linfx}[1]{\ensuremath{\textup{L}^{#1}_{\infty\omega}}}
\newcommand{\Cinf}{\ensuremath{\textup{C}^{\omega}_{\infty\omega}}}
\newcommand{\Cinfx}[1]{\ensuremath{\textup{C}^{#1}_{\infty\omega}}}

\newcommand{\N}{{\mathbb N}}

\newcommand{\Ee}{{\mathcal E}}

\newcommand{\Gg}{{\mathcal G}}

\newcommand{\Kk}{{\mathcal K}}
\newcommand{\Ll}{{\mathcal L}}

\newcommand{\Oo}{{\mathcal O}}

\newcommand{\Qq}{{\mathcal Q}}

\newcommand{\Tt}{{\mathcal T}}

\newcommand{\ns}{\ensuremath{\text{s}}}
\newcommand{\nws}{\ensuremath{\text{ws}}}
\newcommand{\NT}{\ensuremath{\text{NonTerm}}}

\newcommand{\PCdegree}{\ensuremath{\textsc{PC-Degree}}}
\newcommand{\PCsize}{\ensuremath{\textsc{PC-Size}}}

\newcommand{\lintdom}[1]{\ensuremath{#1_{\delta}}}
\newcommand{\lintcong}[1]{\ensuremath{#1_{\approx}}}

\newcommand{\IsoForm}{\ensuremath{\text{ISO}}}

\newcommand{\Str}{\textup{Str}}
\DeclareMathOperator{\characteristic}{char}

\newcommand{\field}{\mathbb}
\newcommand{\Inv}{\text{Inv}}

\newcommand{\CFI}{Cai-Fürer-Immerman\xspace}
\newcommand{\CFIgraph}[3]{\text{\sf CFI}\,[#1; #2; #3]}
\newcommand{\CFIclass}[2]{\text{\sf CFI}\,[#1; #2]}

\DeclareMathOperator{\Norm}{Norm}

\newcommand{\Ord}{\leq\!\!}
\newcommand{\taucfi}{\tau_{\text{\sf CFI}}}

\newcommand{\tup}[1]{{\vec #1}}
\newcommand{\Blocked}{{\textsc{Bl}}}

\newcommand{\cyc}[2]{\ensuremath{({#1},{#2})\text{-cyclic}}\xspace}
\newcommand{\cyclp}{\cyc{\ell}{p}}

\newcommand{\lcocy}{$\ell$-cocyclic\xspace}

\newcommand{\Typ}[1]{\ensuremath{\text{Tp}^{#1}}}
\newcommand{\Typk}[2]{\ensuremath{\text{Tp}^{#1}_{#2}}}

\usepackage[textwidth=2.5cm]{todonotes}

\begin{document}

\title{A Finite-Model-Theoretic View\texorpdfstring{\\}{ }on Propositional Proof Complexity}
\titlecomment{{\lsuper*}This article is an extended version  of the 
conference paper~\cite{GrPaPa17}.}

\author[E.~Gr\"adel]{Erich Gr\"adel\rsuper{a}}	
\address{RWTH Aachen University, Germany}
\email{\texttt{graedel@logic.rwth-aachen.de}}
\email{\texttt{grohe@informatik.rwth-aachen.de}}
\email{\texttt{benedikt.pago@rwth-aachen.de}}  
\author[M.~Grohe]{Martin Grohe\rsuper{a}}
\author[B.~Pago]{Benedikt Pago\rsuper{a}}	
\author[W.~Pakusa]{Wied Pakusa\rsuper{b}}
\address{University of Oxford, England}
\email{\texttt{pakusa@logic.rwth-aachen.de}}
\thanks{The fourth author was supported by a DFG grant (PA 2962/1-1).}




 
\begin{abstract} 
We establish new, and surprisingly tight, connections 
between propositional proof complexity and finite model 
theory.
Specifically, we show that the power of several propositional proof 
systems, 
such as Horn resolution, bounded-width resolution, and the polynomial 
calculus of bounded degree, can be characterised in a precise sense by 
variants of fixed-point logics that are 
of fundamental importance in descriptive 
complexity theory.
Our main results are that \emph{Horn resolution} has the 
same expressive power as \emph{least fixed-point logic}, that 
\emph{bounded-width resolution} captures \emph{existential least 
fixed-point logic}, and that
the \emph{monomial calculus with bounded degree over the rationals} 
solves precisely 
the problems definable in \emph{fixed-point logic with 
counting}. We also study the \emph{bounded-degree polynomial calculus}. Over the rationals, it captures fixed-point logic with counting if we restrict the bit-complexity of the coefficients. For unrestricted coefficients, we can only say that the bounded-degree polynomial calculus is at most as powerful as \emph{bounded variable infinitary counting logic}, but a precise logical characterisation of its power remains an open problem.\\
These connections between logics and proof systems allow us to establish 
finite-model-theoretic 
tools for proving lower bounds for the polynomial calculus over 
the rationals and also over finite fields.\\
\\
This is a corrected version of the paper (https://arxiv.org/pdf/1802.09377.pdf) published originally on January 23, 2019.
\end{abstract}

\maketitle

\section{Introduction}\label{S:one}
The question whether there exists an efficient proof system 
by means of which the validity of \emph{arbitrary propositional formulas} 
can be verified via \emph{proofs of polynomial 
size} is equivalent to the closure of NP under complementation. 
Since Cook and Reckhow \cite{CookRec79} made the notion of an efficient 
propositional proof system precise, 
a huge body of research on the power of various propositional proof 
system has been established. In particular, we now have super-polynomial 
lower 
bounds on the proof complexity for quite strong proof systems, see 
\cite{BeamePit01,Segerlind07} for surveys on propositional proof 
complexity.

In this paper we study \emph{polynomial-time variants} of propositional 
proof systems,
which admit efficient proof search, resulting in proofs of polynomial size,
such as restricted variants of resolution and the polynomial calculus. To be precise, one of these variants, the bounded-degree polynomial calculus over the rationals, is not known to admit polynomial-time proof search because the proofs may involve very large coefficients. Thankfully, as it turns out, this issue does not prevent a meaningful connection to finite model theory.

Recall that the resolution proof system 
$\Res$ takes as input a propositional 
formula $\phi$ in conjunctive normal form (CNF), and it
refutes the satisfiability of $\phi$ if there is a derivation of the empty 
clause from $\phi$. 
It is well-known that shortest resolution proofs can be of exponential 
size, so in general, we provably cannot search for resolution proofs in 
polynomial time.
However, there are interesting restrictions  of $\Res$, such as 
$\HRes$ (resolution restricted to Horn clauses) and bounded-width
resolution
$\kRes$ (resolution restricted to clauses of size 
$\leq k$) that do admit efficient proof search, that is the existence of 
refutations can be verified in polynomial time.
Of course, unless P = NP, any proof system that admits efficient proof 
search is necessarily incomplete for full propositional logic.  
Nevertheless we can still prove interesting statements
in such systems, and usually have completeness for relevant fragments of 
propositional logic,
such as Horn-logic or 2-CNF.
We can now try to solve algorithmic problems by reducing them to
provability (or refutability) in some specific polynomial-time proof 
system,
which, if it works successfully for all inputs, would give us a  
polynomial-time algorithm for the problem.
Our goal is to understand how powerful this approach can be, depending on
the specific proof system that we use.

Let us illustrate this by two concrete problems.
First we consider \emph{graph isomorphism},
a problem which is not known to be solvable in 
polynomial time although there is strong evidence that it is not
$\NP$-complete.
Given two graphs $G=(V,E)$ and $H = (W, F)$ we ask
whether there is a bijection $\pi\colon V \to W$ such that $\pi(E) = F$.
Of course, this can easily be encoded as the satisfiability 
problem of a propositional CNF-formula.
First, for each pair of vertices $v \in V$ and $w \in W$ we introduce a 
variable $X_{vw}$ with the intended meaning that $X_{vw} = 1$ if $\pi(v) = 
w$.
We add clauses $\bigvee_{w \in W} X_{vw}$ for every $v \in V$ and 
$\bigvee_{v \in V} X_{vw}$ for every $w \in W$  to ensure that 
every $v \in V$ has an image and every $w \in W$ has a preimage.
Additionally we add for all $v_1, v_2 \in V$ and $w_1, w_2 \in W$ a 
clause $\neg (X_{v_1w_1} \wedge X_{v_2w_2})$ in case that $\{ v_1 \mapsto 
w_1, v_2 \mapsto w_2\}$ is not a partial isomorphism.
The resulting CNF-formula, denoted by 
$\text{Iso}(G,H)$, is satisfiable if, and only if, the two graphs $G$ and 
$H$ are isomorphic.
Following our reasoning from above, we can now use an efficient 
variant of resolution, or of a stronger proof system, and try to refute 
the satisfiability of the 
formula $\text{Iso}(G,H)$. If this is possible, then $G$ and $H$ are not 
isomorphic. 
Unfortunately, if we do not find a proof, then we are stuck, because it 
might still be the case that $G$ and $H$ are not isomorphic, but our 
 proof system is just not strong enough to show this.
Hence, we get an efficient, sound, but not necessarily complete 
graph isomorphism test. 
The question of how successful this approach is  when based on resolution 
was
studied by Toran in~\cite{Toran13}. Unfortunately, he proved that shortest 
resolution proofs for graph non-isomorphism can be of exponential size 
(even for graphs with colour class size four).
More recently, Grohe and Berkholz showed that also in the stronger system 
polynomial calculus (\textsc{PC}) one cannot obtain small proofs for graph 
non-isomorphism~\cite{BerkholzGro15,BerkholzGro17} in the general case.

Our second example is directed graph reachability: Given a directed graph
$G = (V,E)$ with two distinguished vertices $s, t \in V$,
we want to know whether there is a path from $s$ to $t$ in $G$. Again, it 
is easy to encode this as a satisfiability problem in 
propositional logic, by taking the conjunction of all 
implication clauses $X_v \to X_w$, for all edges $(v,w)\in E$,
together with the two clauses $1 \to X_s$ and $X_t 
\to 0$. Clearly the resulting formula $\text{NonReach}(G,s,t)$ is 
unsatisfiable if, and only if, $t$ is reachable from $s$ in $G$.
However, in clear contrast to the formulas $\text{Iso}(G,H)$ from above, 
we can easily prove unsatisfiability for the formulas 
$\text{NonReach}(G,s,t)$ in efficient variants of resolution
such as $\HRes$ and $\kRes$ for $k\geq 2$.

Our two examples demonstrate the following: while certain problems, such 
as directed graph reachability, allow for small and efficient resolution 
proofs, other problems, such as the graph isomorphism problem, provably 
require proofs of super-polynomial size even in quite strong proof 
systems. 
This leads to the main question that we want to address in this paper: is 
there a \emph{classification} for those problems which can be solved in 
natural restricted versions of propositional proof systems such 
as 
$\HRes$, $\kRes$  and $\PCx k$ (the degree-$k$ restriction of the 
polynomial calculus)?
It came as a surprise to us that there is, indeed, a very clear and 
tight classification of the power of all of these proof systems in terms 
of definability in important fixed-point logics and infinitary logics which are well-studied in 
the area of descriptive complexity theory.

Before we can state our results in detail, we have to explain what 
we mean by saying that a problem, such as directed graph reachability, can 
be solved by a propositional proof system $\Prop$.
As usual, each decision problem can be identified with a membership 
problem 
``$\mfA \in \mcK$?'' for some class of structures $\mcK$.
For instance, the graph reachability problem from above is identified with 
the class $\mcK_{\text{Reach}} = \{ (V, E, s, t) : \text{there is a path 
from } s \text{ to } t \text{ in } G=(V,E)\}$.
Then we naturally want to say that a problem $\mcK$ can be solved by the 
proof system $\Prop$ if we can find a reduction function $f$ which 
maps structures $\mfA$ to inputs $f(\mfA)$ for $\Prop$ 
such that $\mfA \in \mcK$ if, and only if, $\Prop$ can prove that 
$f(\mfA)$ is not satisfiable.
It is clear that we only want to allow \emph{simple} 
reduction functions $f$, because otherwise the computation of the encoding 
could already contain part of the work to solve the problem. 
Coming from the area of finite model theory the obvious and natural
formalisation for ``$f$ being simple'' is to say that $f$ is 
definable in \emph{first-order logic} (\FO).
We introduce the precise technical definition of such reductions, which 
is the notion of  a \emph{first-order interpretation}, in 
Section~\ref{sec:prel}.
Note that for the two examples we discussed above the encoding 
functions are clearly $\FO$-definable.

Having established this definition it turns out that our 
classification problem is really about understanding the expressive power 
of the \emph{Lindström extensions} of first-order logic by 
\emph{generalised quantifiers} for propositional proof systems $\Prop$. 
We denote these logics by $\FO(\Prop)$. The basic idea of the logic 
$\FO(\Prop)$ is to extend first-order logic by new quantifiers 
$\mcQ_{\Prop}$ which are capable of simulating  $\Prop$.
In other words, we just incorporate into first-order logic the power to 
simulate $\Prop$ in an explicit way, that is the logics 
$\FO(\Prop)$ are a formalisation of the concept 
of 
oracle Turing-machines with access to $\Prop$ in the world of first-order 
logic (the oracle calls to the proof system $\Prop$ correspond to 
applications of the new generalised quantifiers). Again, the precise 
technical definitions of the Lindström extensions $\FO(\Prop)$ can be 
found in Section~\ref{sec:prel}. 
We can now  say that a problem $\mcK$ can be 
solved in a proof system $\Prop$ if, and only if, it is definable in 
$\FO(\Prop)$.
For instance, we saw that $\mcK_{\text{Reach}}$ is definable in 
the logics $\FO(\HRes)$ and $\FO(\kResx 2)$.

\bigskip
We proceed to describe our main results and give a rough sketch of the 
structure of this article.
This work is based on our conference paper~\cite{GrPaPa17}.
However, the present article also contains some new results and substantial 
generalisations of our results from~\cite{GrPaPa17} on the polynomial calculus.

In Section~\ref{sec:resolution}, we study the resolution 
proof system and its aforementioned restrictions Horn-Resolution 
($\HRes$) and Bounded-width-$k$ Resolution
($\kRes$), for $k \geq 2$.
It turns out that $\HRes$ can express precisely the problems that 
are definable in least-fixed point logic (\LFP), that is $\FO(\HRes) = 
\LFP$. This readily follows by the well-known fact 
that the problem of computing winning positions in reachability games 
(known as GAME or  alternating reachability)  is complete for $\LFP$ with 
respect to 
$\FO$-reductions.
More interestingly, we proceed to show that $\kRes$, for every $k \geq 2$, 
is less powerful than $\HRes$. In fact, $\FO(\kResx 2) 
= \FOTC$, where $\FOTC$ is the extension of first-order logic by a
transitive closure operator. Moreover, we prove that, for every $k \geq 
3$, 
$\FO(\kRes) = \EFP$, where $\EFP$ is the \emph{existential} fragment of 
least fixed-point logic which is known to be a strict fragment of full 
least fixed-point logic. 
We can also show that the Lindstr\"om extensions for Horn resolution and 
width-$k$ resolution have different structural properties. While for 
$\FO(\HRes)$ a single application of a $\mcQ_{\HRes}$ quantifier suffices 
to obtain the full expressive power, nesting of 
$\mcQ_{\kRes}$ quantifiers is needed for the logics $\FO(\kRes)$. 

In Section~\ref{sec:fpc}, we then turn our attention to the polynomial 
calculus (\PC),  a propositional proof system which is based on algebraic 
reasoning techniques. The polynomial calculus manipulates polynomial 
equations over an underlying field $\field F$. A $\PC$-refutation is a 
derivation of the equation $0 = 1$. 
As in the case of bounded-width resolution, if one restricts the 
degree of the polynomials in all equations to some constant $k \geq 1$,
then one can search for PC-proofs in polynomial time (when working over the field of rationals, the bit-complexity of the coefficients must also be restricted to binary representations of polynomial length).
Besides restricting the degree, one can also vary the underlying 
field~$\field F$. Specifically, we consider the cases where $\field F$ is 
the 
field of rationals (or reals) or a finite field. Moreover, the polynomial 
calculus can also be restricted by weakening its proof rule for 
multiplication, which defines a variant known as the monomial-PC (\monPC).
We denote its corresponding restriction to degree $k$ by $\monPCx k$.

For the case of the polynomial calculus over $\mbQ$ we 
show the following. 
First of all, if we consider the monomial-PC restricted to some degree $k 
\geq 2$, then this proof system $\monPCx k$ has precisely the same 
expressive power as fixed-point logic with counting (\FPC), which is a 
very expressive logic 
well-studied in descriptive complexity theory \cite{Dawar15,Otto97}; formally, 
we show that $\FOnum(\monPCx k) = \FPC$ for $k \geq 2$ where $\FOnum$ 
denotes the extension of $\FO$ by a numeric sort to match 
the setting of $\FPC$. 
In particular, this separates the (monomial-)PC from the resolution proof 
system since $\FPC$ is known to be much stronger than $\EFP$ and $\LFP$.
In a second step, we generalise this characterisation for the monomial-PC 
to the full polynomial calculus (PC). To deal with the already mentioned phenomenon of potentially exceedingly large coefficients, we restrict the degree-$k$ PC further and define, for any $b \in \bbN$, the proof system $\PC_{k,b}$ as the degree-$k$ PC with the limitation that all coefficients occurring in a proof must be representable as fractions of binary numbers with at most $n^b$ bits each ($n$ refers to the number of variables in the input polynomials). Then we prove that for any constants $k,b$, the proof system $\PC_{k,b}$ captures $\FPC$, just like $\monPCx k$ does. From there, we move on to the more common $\PC_k$ with \emph{unrestricted bit-complexity}, and observe that we can define the existence of $\PC_k$-proofs in the infinitary counting logic $\Cinfx{k}$. This logic is strictly more expressive than $\FPC$. The question whether $\PC_k$-proofs with unbounded rational coefficients are also definable in $\FPC$ remains open. Yet, we can say that a positive answer to it seems unlikely because it is not even clear that this problem is decidable in polynomial time: As shown by Hakoniemi \cite{Hakoniemi21}, there exists a set $Q_n$ of polynomials over Boolean variables that has a refutation in the degree-$2$ polynomial calculus, but none that requires less than exponentially many bits for the coefficients; it is doubtful that such a refutation can be computed in polynomial time.\\
On our way we prove a result which is of independent interest, 
namely that $\FPC$ can define solution spaces of linear equation systems 
over the rationals. We need this in order to express $\PC_{k,b}$ in $\FPC$, and indirectly also to express $\PC_k$ in $\Cinfx{k}$. The latter result allows us to answer an open question by Grohe and Berkholz 
from~\cite{BerkholzGro15} about the relative power of $\monPCx{k}$ and $\PC_k$ with respect to the graph isomorphism problem.

In Section~\ref{sec:PCfiniteFields}, we turn our attention to the 
polynomial calculus over finite 
fields. It is easy to see that the connection between $\FPC$ and the 
(monomial-)PC breaks down. We set out to establish 
criteria on the characteristic of the underlying finite field and certain 
finite-model-theoretic properties of polynomial equation systems that 
allow us to retain 
$\FPC$-definability of bounded-degree PC-refutations.
This result proves to be very useful in 
order to derive lower bounds for the polynomial calculus over finite 
fields.
There are also technical results in this section which should be 
of independent interest. For example, we show that classes of 
CFI-structures over expander graphs are homogeneous with respect to 
$\FPC$-definability. 

Finally, in Section~\ref{sec:lowerboundsPC}, we discuss how we 
can apply our $\FPC$-definability results in order to prove 
lower bounds for the polynomial calculus.
We give examples including the graph isomorphism problem and constraint 
satisfaction problems. 
Although most (but not all) of the lower bounds have been known before, we 
present new proofs which only use finite-model-theoretic arguments.
Our novel, uniform approach to these lower bounds, also suggests a way to 
capture a common weakness of many propositional proof systems: 
whenever a proof system has a stratification which allows for  
\emph{symmetric} refutations that can be described and verified in counting 
logic with a bounded number of variables, our lower bounds techniques can be 
applied. For illustration, we discuss the example of the Positivstellensatz 
proof system in Section~\ref{sec:discussion}.

\subsection*{Related work}
Let us discuss some related work. 
The most relevant result to mention here is the 
characterisation by Atserias and Dalmau of resolution width in 
terms of the number of pebbles required to win an existential pebble 
game played on a given CNF-formula and a 
structural encoding of truth assignments~\cite{At04, AtseriasDal08}. This 
resembles our result that bounded-width resolution corresponds to 
\emph{existential} least fixed-point logic.
Using their game-theoretic characterisation, Atserias 
and Dalmau can reprove many of the known lower bounds on 
resolution width. Again, this is similar to the applications we give in
Section~\ref{sec:applications:low}.

However, what makes our setting different from the approach of 
Atserias and Dalmau is that we always consider the power of proof 
systems only \emph{up to logical reductions}. This reflects, for 
example, in our result saying that $\FO(\kResx 3) = \FO(\kResx 4)$, 
i.e.\ that $\kResx 3$ has the same expressive power as $\kResx 4$. But, 
certainly, 
this only holds if we allow first-order reductions to transform inputs 
between $\kResx 4$ and $\kResx 3$. 
Hence our characterisation of resolution width is ``coarser'' than that of 
Atserias and Dalmau. But it has the advantage of being more robust.
For instance, 
in the situation of lower bound proofs, we can avoid playing 
pebble games directly on the inputs to proof systems, such as 
CNF-formulas, 
but instead it suffices to play suitable games on pairs of structures in 
which 
these inputs interpret. This can make the 
description of winning strategies much simpler. Furthermore, our setting 
allows us to prove lower bound results not depending on specific encodings 
of a problem, since our logics are closed under interpretations, see 
Section~\ref{sec:lowerboundsPC}.
 
Besides this, we want to mention the 
series of 
papers~\cite{AtMa12, BerkholzGro15, Malkin14,GrOt15} which establish
surprisingly tight connections between the equivalence of graphs in 
counting logic and their indistinguishability by linear programming 
techniques (Sherali-Adams relaxiations of graph isomorphism 
polytopes) and algebraic propositional proof system. Similar to our 
applications, these results also allow the transfer of 
lower bounds from finite model theory to get lower 
bounds on proof complexity. 
In particular, we use notions and ideas of \cite{BerkholzGro15} in 
Section~\ref{sec:fpc}. 
Let us also point to the excellent work of Dawar and 
Wang~\cite{DawarW16,DawarW17} which connect finite model theory with 
semi-algebraic proof systems.
This work, and our own, has certainly also interesting connections 
to the very recent work by Atserias and Ochremiak \cite{AtseriasOch18}
showing by means of finite-model-theoretic arguments that the Sums-of-Squares 
proof system can be simulated in $\Cinf$.
Surely these connections deserve to be explored further.

\section{Preliminaries}
\label{sec:prel}

This is a paper in finite model theory. All structures are 
\emph{relational} and \emph{finite} if not explicitly stated otherwise.
We assume that the reader has a solid background in logic. To fully 
understand
and appreciate our results, familiarity with the ideas and techniques of 
finite
model theory will be necessary 
(see \cite{ebbflu99,Im99,lib04,grakollib+07}).  
In particular, a good knowledge of fixed-point logic with counting 
is needed in order to understand our definability results for the 
polynomial calculus in 
Sections~\ref{sec:fpc},\ref{sec:PCfiniteFields},\ref{sec:lowerboundsPC}, 
see the above references plus~\cite{Otto97,Dawar15}.

\subsection{Finite Relational Structures.}
Given a (finite, relational) \emph{vocabulary} 
(or \emph{signature})~$\tau$, a \emph{$\tau$-structure} $\mfA$ consists of a 
finite 
\emph{universe} $A$ and a
relation $R^A\subseteq A^k$ for each $k$-ary relation symbol $R$ in 
$\tau$. 
If we consider (undirected) graphs, that is structures 
over the vocabulary $\tau = \{ E \}$, then we usually use a different
notation and denote graphs by $G = (V,E)$. In particular, we denote the 
vertex set of a graph $G$ by $V = V(G)$ and the set of edges $E$ by $E = 
E(G)$.
The class of all (finite)
$\tau$-structures is denoted by $\Str(\tau)$. 
Sometimes we want to distinguish certain constants in $\tau$-structures 
$\mfA$. For a tuple of parameters $\tup z$ we denote by $\Str(\tau, \tup 
z)$ the class of all pairs $(\mfA, \tup z \mapsto \tup a)$ where $\mfA\in 
\Str(\tau)$.
For our applications in Section~\ref{sec:PCfiniteFields} and 
Section~\ref{sec:lowerboundsPC}, we also fix an encoding of
ordered pairs $(\mathfrak A,\mathfrak B)$ of $\tau$-structures as 
structures
$(\mathfrak A,\mathfrak B)$ of
some vocabulary $\tau_{\textup{pair}}$.

\subsection{Logics without Counting.}
We assume that the reader is familiar with \emph{first-order logic} (\FO) 
and \emph{least and inflationary fixed-point logic} ($\LFP$ and $\IFP$).
\emph{Infinitary finite variable logic} $\Linf$ extends $\FO$ by 
infinite conjunctions and disjunctions in formulas, but with the 
additional requirement that formulas only contain a finite number of 
variables.
More precisely, if we denote by $\Linfx k$ the \emph{$k$-variable 
fragment} of $\Linf$, then we have $\Linf = \bigcup_k \Linfx k$.
Formulas of $\LFP$ with $k$ variables can be translated 
into equivalent formulas of $\Linfx{k}$. In particular, $\LFP \leq
\Linf$ (in this article we use the notation $\Ll_1 \leq \Ll_2$ to say that 
every class of structures that is $\Ll_1$-definable is also 
$\Ll_2$-definable, that is $\leq$ refers to semantic inclusion of 
logics wrt.\ sentences).

\subsection{Logics with Counting.}
\emph{(Infinitary) counting logic} $\Cinf$ is the extension of $\Linf$ that 
allows counting
quantifiers $\exists^{\ge m}x$ (``there exist at least $m$ values for
$x$'') for each $m$ (with the same restriction on the number of variables 
as before, that is each $\Cinf$-formula only contains a finite number of 
variables).
Note that each individual quantifier $\exists^{\ge m}x$ can be expressed 
using $m$ first-order quantifiers and
$m$ distinct variables for $x$. However, the translation leads to formulas 
with a higher quantifier rank and, moreover, it increases the number of 
required variables. Analogous to the above, we denote by $\Cinfx k$ the 
fragment of $\Cinf$ consisting of all formulas with at most $k$ (free or 
bound) variables. Then $\Cinf = \bigcup_k \Cinfx k$.
For two structures $\mfA, \mfB$ (of the same vocabulary) we write $\mfA 
\equiv^k \mfB$ if the structures cannot be distinguished by any formula of 
$\Cinfx k$.

We now recall the definition of \emph{fixed-point logic 
with counting} ($\FPC$). In a nutshell, $\FPC$ is the 
extension of inflationary fixed-point logic ($\IFP$) by \emph{counting 
terms}.
Formulas of $\FPC$ are evaluated over the \emph{two-sorted extension} of 
an 
input structure $\mfA$ by a copy of the natural numbers.  
Following~\cite{DawarGroHolLau09} we denote by  $\mfA^{\#}$ the two-sorted 
extension 
of a $\tau$-structure $\mfA=(A,R_1,\dots,R_k)$ by $\mfN=(\N,+,\cdot,0,1)$, 
that is the two-sorted structure $\mfA^{\#} = ( A, R_1,\dots , R_k , \N, +, 
\cdot, 0, 1)$ where the universe of the first sort (also referred to as 
\emph{vertex sort}) is $A$ and the universe of the second sort (also 
referred to as \emph{number sort} or \emph{counting sort}) is $\N$.
For both, the vertex and the number sort, we have a collection of 
\emph{typed first-order variables}, that is the domain of any variable $x$ 
(over the input structure $\mfA$) is either $A$ or $\N$.
Similarly, for second-order variables $R$ we allow mixed types, that is a 
relation symbol $R$ of type $(k , \ell) \in \N \times \N$ stands for a 
relation 
$R \subseteq A^k \times \N^\ell$. 

Of course, already $\FO$ is undecidable over the class of two-sorted 
structures $\mfA^{\#}$. To obtain a logic with polynomial-time data 
complexity, we have to restrict the range of 
quantifiers over the numeric sort by fixed polynomials.
More precisely, $\FPC$-formulas can use quantifiers over the numeric sort 
only in the form $Qx \leq n^q . \phi$ where $Q \in \{ \exists, \forall 
\}$ and where $q \geq 1$ is a fixed constant.
The range of the quantifier $Q$ is $\{0, \dots, n^q\}$ where $n$ denotes 
the 
size of the input structure~$\mfA$. To simplify notation, we henceforth 
assume that each numeric variable $x$ comes with a built-in restricted 
range 
polynomial, that is $x = (x \leq n^q)$. For better readability, we 
usually omit this range polynomial in our notation.
By this convention, each variable $x$ has a 
predefined range in any input structure $\mfA^{\#}$ of polynomial size 
(which is either $A$ or $\{0, \dots, n^q\}$ for a fixed $q \geq 1$).
We denote this range by $\dom(\mfA,x)$ (or just by $\dom(x)$ if 
$\mfA$ is clear from the context). Analogously, for a tuple of variables 
$\tup x=(x_1, \dots, x_k)$ we set $\dom(\tup x) = \dom(x_1) \times \cdots 
\times \dom(x_k)$.
By this, we also obtain polynomial bounds for numeric components in 
fixed-point definitions $\left[ \ifp \, R\bx \,.\, \phi(R,\bx) 
\right](\bx)$.
Indeed, the inflationary fixed-point defined by this formula is of the 
form $R \subseteq \dom(\tup x)$. 

The crucial elements of $\FPC$ are \emph{counting terms} which 
allow to define cardinalities of sets. Starting with an arbitrary
$\FPC$-formula $\phi(x)$ one can form a new \emph{counting term} $s = [\# 
x \,.\, \phi]$ whose value in $\mfA$ is just the size of the set defined 
by $\phi$ in $\mfA$.
In particular, the term $s$ is a \emph{numeric term}, that is $s$ takes 
its 
value in the number sort. More precisely, for an input structure~$\mfA$, 
the  value $s^\mfA \in \N$ of $s$ in $\mfA$ is the number of elements 
$a \in A $ such that $\mfA \models \phi(a)$.
One can allow counting terms of a more general form 
without increasing the expressive power of \FPC. In particular, counting 
terms $[\# \bx \,.\, \phi]$
over mixed tuples of variables can be simulated with unary counting 
terms  and fixed-point operators; we refer to~\cite{Otto97} for more 
details 
and background on fixed-point logic with counting.

An important fact that we are going to use frequently is that formulas of 
$\FPC$ with $k$ variables can be rewritten as equivalent 
$\Cinfx k$-formulas. In particular, we have that $\FPC \leq \Cinf$.
In Section~\ref{sec:PCfiniteFields}, we also make use of the fact that 
for every $k \geq 1$, there exists an $\FPC$-formula $\phi$ with $\mcO(k)$ 
many variables such that $(\mfA, \mfB) \models \phi$ if, and only if, 
$\mfA \equiv^k \mfB$, see e.g.\ \cite{Otto97}.

In Section~\ref{sec:fpc}, we also make use of the \emph{numeric
extension of first-order logic}, denoted by $\FOnum$, which is defined 
as $\FPC$, but without the rule for forming (inflationary) fixed points.

\subsection{Logical Interpretations and Lindstr\"om Quantifiers}
The logical counterpart of the notion of an (algorithmic) reduction is 
the 
notion 
of a \emph{logical interpretation}. A logical interpretation $\mcI$ 
transforms an input structure $\mfA$ into a new structure $\mfB = 
\mcI(\mfA)$ 
and 
this transformation is defined by formulas of some logic $\Ll$.
In this article we consider $\Ll$-interpretations with respect to 
different underlying logics $\Ll$, such as $\FO, \FOnum, \LFP, \FPC, 
\Cinf$. Basically, the definition of an $\Ll$-interpretation is uniform for 
all of these logics. However, there is one exception for the case of 
$\FOnum$ and $\FPC$ where we have the special situation that formulas 
can use numeric variables. As a consequence, the interpreted 
structures $\mcI(\mfA)$ can contain such numeric elements.
In this section, we further introduce \emph{Lindström 
quantifiers}, also known as \emph{generalised quantifiers}, which capture 
the notion of \emph{oracles} in the realm of finite model theory.

Let us start with the case of single-sorted logics 
$\Ll$, such as $\FO, \LFP$, or $\Cinf$. Let $\sigma, \tau$ be signatures 
with $\tau = 
\{S_1,...,S_{\ell}\}$. Let $s_i$ denote the arity of $S_i$. An 
$\Ll[\sigma,\tau]$-\emph{interpretation} is a tuple
\[ I(\tup{z}) = 
(\varphi_{\delta}(\tup{x},\tup{z}),\varphi_{\approx}(\tup{x}_1,\tup{x}_2,
\tup{z}),\varphi_{S_1}(\tup{x}_1,...,\tup{x}_{s_1},\tup{z}),...,\varphi_{S_
{\ell}}(\tup{x}_1,...,\tup{x}_{s_{\ell}},\tup{z}))\]
where 
$\varphi_{\delta},\varphi_{\approx},\varphi_{S_1},...,\varphi_{S_{\ell}} 
\in \Ll[\sigma]$ and $\tup{x},\tup{x}_1,...,\tup{x}_{s_{\ell}}$ are tuples 
of pairwise distinct variables of the same length $d$ and $\tup{z}$ is a 
tuple of variables pairwise distinct from the $x$-variables. We call $d$ 
the \emph{dimension} and $\tup{z}$ the \emph{parameters} of $\I(\tup{z})$.

A $d$-dimensional $\Ll[\sigma,\tau]$-interpretation $\I(\tup{z})$  defines 
a partial mapping $\I \colon \Str(\sigma,\tup{z}) \rightarrow \Str(\tau)$ 
in the 
following way:  For $(\AA,\tup{z} \mapsto \tup{a}) \in 
\Str(\sigma,\tup{z})$ we obtain a $\tau$-structure $\BB$ over the universe 
$\{\tup{b} \in A^d \ | \ \AA \models \varphi_{\delta}(\tup{b},\tup{a})\}$, 
setting $S_i^{\BB} = \{(\tup{b}_1,..,\tup{b}_{s_i}) \in B^{s_i} \ | \ \AA 
\models \varphi_{S_i}(\tup{b}_1,...,\tup{b}_{s_i},\tup{a})\}$ for each 
$S_i 
\in \tau$.
Moreover let $\Ee = \{(\tup{b}_1,\tup{b}_2) \in A^d \times A^d \ | \ \AA 
\models \varphi_{\approx}(\tup{b}_1,\tup{b}_2,\tup{a})\}$. Now we define
\begin{equation*}
\I(\AA,\tup{z} \mapsto \tup{a}) :=
\begin{cases}
\BB/ \Ee &\text{if $\Ee$ is a congruence relation on } \BB\\
\text{undefined} & \text{otherwise.}  
\end{cases}
\end{equation*}	 
We say that $\I$ interprets $\BB/\Ee$ in $\AA$. 

Let us briefly discuss the case of two-sorted logics. If $\Ll$ is 
$\FPC$ or $\FOnum$, then we have the same definition of an 
$\Ll$-interpretation as above. However, note that now the variable tuples 
$\tup x, \tup z, \dots$ may contain numeric variables. Recall that each 
numeric variable $x$ has an explicit polynomial range bound $\dom(\mfA, 
x)$ which is either $A$ or $\{0, \dots, n^q\}$ for a fixed $q \geq 1$. As 
a consequence, the domain of the structure $\mcI(\mfA, \tup z \mapsto \tup 
a)$ does not longer consist of equivalence classes of tuples in $A^d$ but, 
more generally, it consists of equivalence classes of elements in
$\dom(\mfA, \tup x)$ (and note that these tuples may contain numeric 
components). 

Next, we introduce Lindström quantifiers. Let $\Ll$ be a logic and $\Kk 
\subseteq Str(\tau)$ a class of 
$\tau$-structures with $\tau = \{S_1,...,S_{\ell}\}$. The \emph{Lindström 
extension} $\Ll(\Qq_{\Kk})$ of $\Ll$ by \emph{Lindström quantifiers} for 
the class $\Kk$ is obtained by extending the syntax of $\Ll$ by the 
following formula creation rule:
\begin{quote}
	Let 
$\varphi_{\delta},\varphi_{\approx},\varphi_{S_1},...,\varphi_{S_{\ell}}$ 
be formulas in $\Ll(\Qq_{\Kk})$ that form an 
$\Ll[\sigma,\tau]$-interpreta-tion $\I(\tup{z})$. Then $\psi(\tup{z}) = 
\Qq_{\Kk}\I(\tup{z})$
	is a formula in $\Ll(\Qq_{\Kk})$ over the signature $\sigma$, with
$(\AA,\tup{z} \mapsto \tup{a}) \models \Qq_{\Kk}\I(\tup{z})$, if, and only 
if,  $\BB := 
\I(\AA,\tup{z} \mapsto \tup{a})$ is defined and $\BB \in \Kk$. 
\end{quote}  

As we see, adding the Lindström quantifier $\Qq$ to the logic $\Ll$ is 
the most direct way to make the class $\mcK$ definable in $\Ll$.
We use this key notion to capture the power of propositional proof systems 
up to first-order definable transformations.

\subsection{Representing Propositional Formulas as Relational 
Structures}
As always when we are dealing with logics in algorithmic contexts, 
we have to agree on encodings of (non-structural) inputs as 
relational structures.
In this article such inputs are, for instance, propositional formulas, 
polynomial equation systems, matrices and vectors over fields.
In all of these cases, it is straightforward to come up with natural 
structural representations. Most 
often, we refrain from describing such encodings explicitly. For it is 
rather tedious, and, more importantly, the concrete details do not matter 
too much: all (natural) encodings will be interdefinable in first-order 
logic.

To get a better intuition, let's go through one encoding explicitly. Let 
us briefly discuss two ways to represent propositional formulas (in CNF) 
as finite relational structures.
Perhaps the most obvious representation of a CNF-formula $\psi$  
as a structure $\AA(\psi)$ is based on the vocabulary 
$\{C,V,P,N\}$;
the universe of $\AA(\psi)$ consists of the variables and the clauses of 
$\psi$, the monadic relations
$V$ and $C$ identify the variables and clauses, respectively, and the 
binary relations 
$P$ and $N$ specify which variables appear positively and negatively in 
which clauses;
so $Pvc$ is true in $\AA(\psi)$ if the variable $v$ appears positively in 
the clause $c$,
and analogously for $N$.
A different representation, that sometimes leads to more elegant logical 
descriptions works
with the set $L$ of literals and with a self-inverse bijection $\neg: L \ra 
L$, so that $\psi$ would be
represented by $\AA(\psi)=(A,C,L,\neg,\in)$ where $A$ is the set of
clauses and literals, $\neg(x)$ is the complementary literal to $x$, and 
$x\in c$ means
that the literal $x$ occurs in the clause $c$ (note that, formally, we do 
not allow function symbols in our vocabularies, but, of course, we can 
substitute function symbols by their graph relations).

\section{Resolution and (Existential) Least Fixed-Point Logic}
\label{sec:resolution}

In this section we study the resolution proof system.
We start by showing that Horn-Resolution ($\HRes$) is complete for least 
fixed-point logic ($\LFP$) wrt.\ (many-to-one) first-order 
interpretations, see Theorem~\ref{thm:hres:lfp}.
In a second step, we consider bounded-width resolution ($\kRes$, for $k 
\geq 2$). We show that bounded-width resolution is strictly weaker than 
Horn-Resolution from the viewpoint of finite model theory. Specifically, 
we prove that $\kResx 2$ is complete for transitive closure logic 
$\FOTC$ (Theorem~\ref{thm:2res:tc}) and that for every $k \geq 3$, $\kRes$ 
is complete for the existential fragment of least fixed-point logic 
($\EFP$), see Theorem~\ref{ResKinEFP}. Since it is known that $\FOTC < 
\EFP < \LFP$, this separates the power of these polynomial-time 
restrictions of the resolution proof system.

\subsection{Horn Resolution Captures Least Fixed-Point Logic}

Let $\posLFP$ be the fragment of \LFP-formulas that are in
negation normal form (i.e.\ negation is applied only to 
input atoms), in which each fixed-point variable is bound only once, and 
that do not make use of greatest fixed points.
Further, let $\EFP_0$ be the basic existential fragment of LFP;
it consists of those formulas in $\posLFP$ whose quantifiers are all
existential.

It is known that, on finite structures (but not in general), every 
LFP-formula can be effectively
translated into an equivalent one in $\posLFP$. On the other side $\EFP_0$ 
is strictly weaker;
it has the same expressive power as Datalog with negation of input atoms.

\begin{thm}\label{LFP-Res}  For every $\phi\in\posLFP[\tau]$ 
there is a
first-order interpretation $I_\phi$ that maps  finite $\tau$-structures
to propositional Horn formulas  $\psi_{\AA,\phi}$ such that
$\AA\models\phi$ if, and only if, $\psi_{\AA,\phi}$ is unsatisfiable.
Further, if $\phi$ is in $\EFP_0$ then
all clauses in $\psi_{\AA,\phi}$ have width at most three.   
\end{thm}

\lueg  Fix a formula $\phi\in \posLFP[\tau]$. For every finite 
$\tau$-structure $\AA$, with universe $A$,
we construct the propositional Horn formula   $\psi_{\AA,\phi}$ as follows.
An \emph{instantiated subformula} of $\phi$ is an expression $\beta(\bar 
a)$ 
which is obtained by taking some subformula $\beta(\bar x)$ of $\phi$ and 
by instantiating every free variable $x$ by some element $a\in A$.
We now take for every instantiated subformula $\beta$ of $\phi$ a 
propositional variable $X_\beta$, and inductively define a set 
$C(\AA,\phi)$ of clauses as follows.
\begin{enumerate}
\item If $\beta$ is a $\tau$-literal, then we add $1\ra X_\beta$ 
in case that
$\AA\models\beta$ and $X_\beta\ra 0$ in case $\AA\not\models \beta$.
\item If $\beta=\eta\lor\theta$, then we add the clauses $X_\eta\ra 
X_\beta$ and
$X_\theta\ra X_\beta$.
\item If $\beta=\eta\land\theta$, then we add the clause $X_\eta\land 
X_\theta\ra 
X_\beta$.
\item If $\beta=\E x\eta(x)$, then we add all clauses $X_{\eta(a)}\ra 
X_\beta$ for $a\in A$.
\item If $\beta=\A x\eta(x)$, then we add the clause $(\bigwedge_{a\in 
A}X_{\eta(a)})\ra X_\beta$.
\item If $\beta=[\lfp R\bar x \,.\, \eta ](\bar a)$ or $\beta=R\bar a$, then we 
add the clause $X_{\eta(\bar a)}\ra X_\beta$.
\end{enumerate}
  
By induction, it readily follows that the minimal model of all these 
clauses sets the variable
$X_\beta$ to true if, and only if,  $\AA\models\beta$
(with fixed-point variables interpreted by their least fixed-point on 
$\AA$).
Let now $\psi_{\AA,\phi}$ be defined as the conjunction of all clauses in 
$C(\AA,\phi)$
together with  $X_\phi\ra 0$.
Then $\psi_{\AA,\phi}$ is unsatisfiable if, and only if, $\AA\models\phi$.

We observe that the only clauses of size larger than three are those 
coming
from universal quantifiers. Hence, if there are no universal quantifiers, 
the
formula only has clauses of size at most three.
Finally it is clear that, for every fixed $\phi\in \posLFP[\tau]$,  we can 
interpret  (a representation of) the formula $\psi_{\AA,\phi}$ inside 
$\AA$,  by using an $\FO$-interpretation $I_\phi$.
\hx

\noindent
This shows that $\LFP \leq \FO(\HRes)$. Actually we 
established a stronger result.

\begin{thm}	\label{LFPinIntHornRes}
For every formula $\phi \in \LFP$ there exists a first-order 
interpretation 
$J_\phi$ 
such that $\Q{\HRes}( J_\phi)$ is equivalent to $\phi$ on 
finite structures. In particular, each 
$\LFP$-formula can be translated into an equivalent  $\FO(\HRes)$-formula 
with a single 
application of the generalised quantifier $\Q{\HRes}$.
\end{thm}

We are ready to prove that $\FO(\HRes)$ has the same expressive 
power as $\LFP$.

\begin{thm}\label{thm:hres:lfp}
On finite structures,  $\LFP=\FO(\HRes)$.
\end{thm}

It remains to  show that $\FO(\HRes) \leq \LFP$, that is we have to express
Horn resolution in \LFP. Recall that a propositional
Horn formula $\psi$  admits a derivation of the empty clause if,
and only if, $\psi$ contains a clause in which
all variables appear negatively, written $X_1 \wedge \dots \wedge X_k\ra 
0$,
such that all unit clauses $\{X_i\}$ for $i=1,\dots, k$
can be derived from $\psi$ by Horn resolution.

Let $\psi$ be presented as a structure $\AA(\psi)$ with universe
$C\cup V$ and vocabulary $\{C,V,P,N\}$.
Let $D$ be the set of variables $v\in V$ such that
the clause $\{v\}$ can be derived from $\psi$ by
Horn resolution. Then $\psi$ is unsatisfiable if, and only if,
$\AA(\psi)\models  \E c (Cc\land \neg\E x Pxc \land \A x(Nxc\ra Dx))$.
The set $D$ is definable by the \LFP-formula
$ [  \lfp   Dx\,.\,  \E c (Pxc  \land \A y (Nyc \ra Dy))](x)$.

\subsection{Bounded-Width Resolution and Existential Least Fixed-Point 
Logic}
\label{sec:kres:efp}
Intuitively, \emph{existential least fixed-point logic} (EFP) extends 
$\EFP_0$ by stratified negation.
This means that it permits fixed-point formulas over existential formulas 
which may depend on
closed fixed-point relations, defined in a lower stratum, and these can be 
used also in negated form.
Thus, negation (and hence, implicitly, also universal quantifiers) are 
present in a limited form, but
least fixed-point recursions may never go through negation or universal 
quantification.
In fact, EFP is equivalent to Stratified 
Datalog and is weaker than full \LFP~\cite{Dahlhaus87,Kolaitis91}.

\begin{defi} Existential fixed-point logic
$\EFP:=  \bigcup_{\ell \geq 0} \EFP_{\ell}$ 
generalises $\EFP_0$ as follows.  The stratum $\EFP_{\ell+1}$ is the 
closure under disjunction, conjunction and existential quantification
of formulas of the form $[\lfp \ R\bar{x}. \E \bar y 
\varphi(R,\bar{x},\bar 
y)](\bar{x})$ 
where $\phi(R,\bar x,\bar y)$ is obtained from a quantifier-free formula, 
that may contain positive and
negative occurrences of additional relations $S_1,\dots,S_m$, by 
substituting these relations by
formulas from $\EFP_\ell$. 
\end{defi}
Let us remark that the logic $\EFP$ is known under different names 
(we stick to the term $\EFP$ which was used 
in~\cite{Dahlhaus87,Kolaitis91}). For instance, in~\cite{ebbflu99}, the 
term \emph{bounded fixed-point logic} (BFP) is used to refer to the same 
logic. Another common name for this logic  is \emph{stratified fixed-point 
logic} (SFP) in reference to its equivalence with Stratified Datalog.  

Notice that first-order logic $\FO$ is contained in $\EFP$, but not in any 
bounded level $\EFP_\ell$, because every quantifier alternation in $\FO$
must be simulated by an additional level of stratified
negation, again see~\cite{Dahlhaus87,Kolaitis91}. 
For the same reason $\EFP$, but none of its levels $\EFP_\ell$, is closed 
under first-order operations. As a consequence of Theorem~\ref{LFP-Res} we 
can infer
\begin{thm}
	On finite structures, $\EFP \leq \FO(\kResx 3)$.
\end{thm}

\lueg Theorem~\ref{LFP-Res} directly establishes this for $\EFP_0$.
So assume that the claim is established for $\EFP_\ell$.
Every formula in $\EFP_{\ell +1}$ can be written as an
$\EFP_0$-formula over predicates that are $\EFP_\ell$-definable.
Hence, by applying Theorem~\ref{LFP-Res} once more, it can be rewritten 
as an $\FO(\kResx 3)$-formula over predicates that are themselves
definable in $\FO(\kResx 3)$.
Since Lindstr\"om extensions of $\FO$ are closed under nesting of 
generalised
quantifiers, it follows that also $\EFP_{\ell+1}\leq  \FO(\kResx 3)$.
\hx   

We require clauses of width 3 for translating $\EFP$-formulas into Horn 
formulas.
In fact, if we restrict to clauses of width 2, then we obtain the power of 
first-order logic with a transitive closure operator $\FOTC$. This 
immediately follows from the fact that satisfiability of $2$-CNF formulas 
reduces to graph reachability, and from the reduction of graph 
reachability 
to the non-satisfiability problem for a $2$-CNF formula that we described 
in the introduction.
\begin{thm}
\label{thm:2res:tc}
On finite structures, $\FO(\kResx{2}) = \FOTC$.
\end{thm}

\subsection{Simulating Bounded-Width Resolution in EFP}
To express width-$k$ resolution, for fixed $k \geq 1$, in 
\EFP, we 
shall use
the representation of a CNF-formula $\psi$ by structures 
$\AA(\psi)=(A,C,L,\neg,\in)$
where $C$ is the set of clauses and $L$ is the set of literals, and the 
universe
is $A=C\cup L\cup\{0\}$.
Further we shall describe the set of all derivable clauses
of size at most $k$ as a $k$-ary relation $D\subseteq (L\cup\{0\})^k$,
that contains those $k$-tuples $(x_1,\dots,x_k)$ for which $\{ x_i: i\leq 
k, x_i\neq 0\}$ is
a clause that is derivable from $\psi$.
This relation $D$ is defined by a fixed-point formula  $[\lfp D\tup x\,.\,
\phi(D,\tup x)](\tup x)$
where $\phi(D,\tup x)$ expresses the following. Either
\begin{enumerate}
\item there exists a clause $c\in C$ such that $c= 
\{x_1,\dots,x_k\}\setminus\{0\}$, or
\item there exist tuples $\tup y, \tup z\in D$ such that, for some $i,j$, 
the literal $z_j$ is the
negated literal to $y_i$, and $(\{ 
y_1,\dots,y_k\}\cup\{z_1,\dots,z_k\})\setminus\{y_i,z_j,0\}=\{x_1,\dots,
x_k\}\setminus\{0\}$.
\end{enumerate}

When spelling out these equations in first-order logic, we can express 
$\phi(\tup x,D)$ by an existential $\FO$-formula
$\E \tup y \,  \alpha(\tup x,\tup y,D,Q)$ where $Q$ is $\FO$-definable 
by a formula (with quantifier prefix $\E^*\A$) that does not depend 
on $D$. This yields a formula in $\EFP_1$.
Since $\EFP$ is closed under $\FO$-operations, this proves 

\begin{thm}
	\label{ResKinEFP}
	On finite structures, $\FO(\kRes) = \EFP$ for all $k \geq 3$.
\end{thm}

Another interesting observation is that if we restrict the nesting depth 
of $\kRes$-quantifiers in $\FO(\kRes)$-formulas to some constant $d 
\geq 1$, then we obtain a fragment $\FO(\kRes)^{\leq d}$ of $\FO(\kRes)$ 
which is strictly 
less expressive.
This follows from the results of Grädel and McColm~\cite{GrMc96} and the 
observation that formulas in $\FO(\kRes)^{\leq d}$ can be written as 
$\Linf$-formulas with at most $d$ many nested unbounded quantifier 
blocks.
However, as Grädel and McColm show there are formulas of transitive 
closure logic $\FOTC$ which require more than $d$ many such blocks when 
expressed as equivalent $\Linf$-formulas.
Since $\FOTC \leq \FO(\kRes)$, it follows that for every $d \geq 1$ we 
have $\FO(\kRes)^{\leq d} < \FO(\kRes)$.
Note that this is different from the case of Horn-Resolution 
where nesting of $\HRes$-quantifiers was not necessary.
In other words, while Horn-Resolution $\HRes$ is many-to-one complete for 
$\LFP$ wrt.\ first-order interpretations, $\kRes$ is only complete wrt.\  
first-order Turing reductions for $\FOTC$ and $\EFP$, respectively.

\section{The Polynomial Calculus over the Field of Rationals and 
Fixed-Point Logic with Counting}
\label{sec:fpc}
We now turn our attention to the \emph{polynomial calculus} (\PC). The 
polynomial calculus is an important and well-studied propositional proof 
system that is based 
on algebraic reasoning techniques. The idea is to represent Boolean 
formulas by polynomial equation systems over some field $\field F$ and to 
show that, by 
manipulating these polynomial equations, one can derive an inconsistent 
equation such as $0 = 1$.
Analogous to the case of bounded-width resolution $\kRes$, it 
is 
possible to stratify the polynomial calculus along a parameter $k \geq 
2$ to obtain polynomial-time fragments. More precisely, if we restrict the 
degree of all polynomials in $\PC$-refutations to some constant $k \geq 
2$, then we obtain (incomplete) fragments $\PCx k$ of the full $\PC$ in 
which proofs can be found in polynomial time (over the rationals, we must also restrict the bit-complexity of the coefficients to ensure this).
One can define the polynomial calculus with respect to any underlying 
field~$\field F$. Throughout this section, this underlying field~$\field 
F$ will always be the field of rationals $\mbQ$. In the following 
Section~\ref{sec:PCfiniteFields}, we turn our attention to the case of finite 
fields.

Another important fragment of the polynomial calculus is the so-called 
\emph{monomial-PC} (\monPC). This restricted variant of the full PC 
was introduced by Berkholz and Grohe in~\cite{BerkholzGro15}. Their 
intention was to precisely characterise the power of a 
combinatorial graph isomorphism test, the so-called Weisfeiler-Leman 
algorithm~\cite{CFI92},
in terms of propositional proof complexity.
Specifically, they proved that two graphs $G$ and $H$ can be distinguished 
by the $k$-dimensional Weisfeiler-Leman algorithm if, and only 
if, the $k$-dimensional monomial-PC ($\monPCx k$) can refute the 
solvability of a certain system of polynomial equations $\IsoForm(G,H)$ 
over $\mbQ$ which encode the graph isomorphism problem for $G$ and $H$.
In this section, we analyse the power of the monomial-PC and the (full) PC 
from the perspective of finite model theory.
In our main result we are going to show that both proof systems, the 
bounded-degree monomial-PC and the bounded-degree and bounded bit-complexity PC, have precisely the same expressive power as 
fixed-point logic with counting (\FPC), which is a 
natural and powerful logic of great importance in the area of descriptive 
complexity theory (see Theorem~\ref{thm:FPCgleichPC}). For the bounded-degree PC over $\bbQ$ without any restriction on the complexity of the coefficients, we show that it is contained in finite variable infinitary counting logic.
As a consequence, the correspondence between the Weisfeiler-Leman 
algorithm 
and the monomial-PC can be generalised to the full PC (though we have to 
sacrifice the 
tightness of the connection between the degree of polynomials 
and the dimension of the Weisfeiler-Leman algorithm), see 
Theorem~\ref{thm:monPCvsPC:GI}.

Our proof consists of three parts.
First of all, we show that proofs in the monomial-PC can be expressed in 
$\FPC$, see Subsection~\ref{subsec:monpc:infpc}. This shows that 
$\FOnum(\monPCx k) \leq \FPC$.
After that, we show in Subsection~\ref{subsec:fpc:inmonpc} that 
$\FPC$-iterations can be simulated using the monomial-PC. Taken together 
this shows that $\FOnum(\monPCx k) = \FPC$.
As the final step, in Subsection~\ref{subsec:fullpc:infpc}, we show that 
$\FPC$ can also express degree-$k$ refutations of bounded bit-complexity in the (full)
polynomial calculus over $\mbQ$, which also entails that $\PC_k$ with unbounded coefficients can be simulated in $\Cinf$.

\subsection{The Polynomial Calculus}
We start with some background on the polynomial calculus and its 
restricted variant, the monomial-PC.
Both systems refute the solvability of a given set of \emph{(multivariate) 
polynomial equations} over some  field $\mbF$ using proof rules 
that manipulate such equations.  
In this paper, $\mbF$ will either be the field of rationals $\mbQ$ or a 
finite field $\mbF_{p^n}$ of size $p^n$ for $p \in \Primes$, where $\Primes$ 
denotes the set of primes and where $n \geq 1$.
We denote by $\field F[\mcX]$ the ring of polynomials in 
variables $\mcX = \{ X_j : j \in J\}$, for some 
index set $J$ and with 
coefficients in $\mbF$.
For an ``exponent''  $\alpha: J \to \mbN$ we let the \emph{monomial} 
$X^\alpha$ be defined as $X^\alpha = \Pi_{j \in J} X_j^{\alpha(j)}$.
Then polynomials $f \in \field F[\mcX]$ can be written as $f = 
\sum_{\alpha} f_\alpha \cdot X^\alpha$ where the $f_\alpha \in \mbF$ are 
coefficients from the field $\mbF$ and such that $f_\alpha \neq 0$ for 
finitely many $\alpha$ only. 
The \emph{degree} $\deg(X^\alpha)$ of a monomial $X^\alpha$ is defined as 
$|\alpha| = \sum_{j \in J} 
\alpha(j)$, and the degree  of a polynomial is defined as the maximal 
degree of its monomials. 
A \emph{polynomial equation} is an equation of the form $f = 0$ for a 
polynomial $f \in \field F[\mcX]$. For better readability, we usually omit 
the 
equality ``$= 0$'' when we specify polynomial equations, that 
is we identify polynomials $f \in \field F[\mcX]$ with the corresponding 
normalised polynomial equations $f = 0$.
A \emph{system of polynomial equations} is a set $\mcP = \{ f_i : i \in I 
\}$ 
consisting of polynomials $f_i \in \field F[\mcX]$ for all $i \in I$ where 
$I$ is an (unordered) index set.
A \emph{solution} of $\mcP$ is a common zero $a \in \mbF^J$ of all 
polynomials in $\mcP$.
In what follows, we only consider systems $\mcP = \{ f_i : i \in I 
\}$ which 
contain for every variable $X = X_j$, $j \in J$, the polynomial equation
$(X^2 - X) = 0$.
The axioms $(X^2 - X) = 0$ enforce that each variable $X=X_j$, 
$j 
\in J$, can only take values $0$ or~$1$.
These equations encode the Boolean setting (truth values) that we are 
interested in.

The polynomial calculus is based on the following result from algebra 
which is known as \emph{Hilbert's Nullstellensatz}.
It says that the non-solvability of the system $\mcP=\{f_i:i\in I\}$ of 
polynomial equations is equivalent to the existence of polynomials $g_i 
\in 
\field F[\mcX], i \in 
I$, such that $\sum_{i \in I} g_i \cdot f_i = 1.$
The polynomials $g_i$ are called a \emph{Nullstellensatz 
refutation} for the system $\mcP$. 
The idea of the polynomial calculus is to search for such polynomials 
$g_i$ in a sequential way. 

\begin{defi} 
The inference rules of the \emph{polynomial calculus} ($\PC$) 
over the ring of polynomials $\field F[\mcX]$ are as follows:
\smallskip
\begin{align*}
&\text{(Multiplication)}\qquad &&\quad\frac{\, f \,}{Xf} 
&&\text{ where } X \in \mcX \\ 
&\text{(Linear Combination)}\qquad  
 && \frac{g , \, f}{ag+bf}
 && \text{ where } a, b \in \mbF
\end{align*}
The goal of the polynomial calculus is to derive with these rules, 
from a collection $\mcP$ of \emph{axioms} $p \in \mcP$,
the constant polynomial $1 \in \mbF[\mcX]$,  in order to prove that the 
polynomials in~$\mcP$ have no common zero.

The \emph{monomial-PC}  ($\monPC$) is the restriction of the (full) 
$\PC$ that permits the use of the multiplication rule only in the cases 
where $f$ is either a 
monomial or the product of a monomial and an axiom.	
A polynomial equation system $\mcP$ has a \emph{refutation} of 
degree $k \geq 1$
in the $\PC$ (or $\monPC$)  if the constant
polynomial $1 \in \field F[\mcX]$ can be derived from $\mcP$ using only 
polynomials of degree at most $k$.
\end{defi}

The polynomial calculus, and the monomial-PC, are clearly sound and, by 
Hilbert's 
Nullstellensatz, also complete proof systems. However, completeness 
requires unbounded degree in refutations. In fact, as we 
indicated before, the ``degree of polynomials'' for 
the $\PC$ ($\monPC$) is a complexity measure with very similar properties 
as the 
``width of clauses'' measure for the resolution proof system.
If we restrict the $\PC$ ($\monPC$) to polynomials of degree at most $k$, 
for some fixed $k\geq 1$,
then the systems become incomplete, but admit proof search in polynomial 
time (again, for $\PC$ over $\bbQ$, we must also restrict the bit-complexity of the coefficients).
In what follows, whenever we speak of the monomial-PC or the (full) PC, 
then we usually refer to a variant with restricted degree $k \geq 1$. If 
we want to make 
this constant $k$ explicit, then we 
denote the corresponding proof system by $\monPCx k$ and $\PCx k$, 
respectively.
Another fact which we use implicitly throughout this section is 
that the axioms $(X^2 - X)$ guarantee that in (monomial-)PC 
proofs we can restrict ourselves to \emph{multilinear} 
polynomials. 
To see this, say that we were able to derive the polynomial 
$p = X^2 Y + Z$ within some (monomial)-PC proof. Of course, $p$ is not 
multilinear. 
However, we can use the axiom $(X^2- X)$ together with the ``linear 
combination''-rule to reduce this polynomial to the corresponding 
multilinear polynomial $p' = X Y + Z$. Indeed, $p' = p - Y(X^2-X)$.
Hence, restricting to multilinear polynomials, and modifying the 
multiplication rule accordingly with implicit linearisation, does not
change the power of the corresponding proof systems.
For a polynomial $p \in \field F[\mcX]$ we denote its 
\emph{multilinearisation} by $\ML(p)$.
So, from now on we stick to the setting of implicitly multilinearising all 
polynomials which precisely captures the semantics of the polynomial 
equations $(X^2 - X) =0$.

We remind the reader that in this section the underlying 
field~$\field F$ for the (monomial-)PC is always the field of rationals 
$\mbQ$.

\subsection{Monomial-PC in Fixed-Point Logic with Counting}
\label{subsec:monpc:infpc}
Our first aim is to show that $\FPC$ can express $\monPCx k$-refutations
over the rationals using only $\mcO(k)$ many variables.
Of course, in order to obtain such a definability result, we have to agree 
on an encoding of sets~$\mcP$ of 
rational, multilinear polynomials
as finite relational structures.
Similar to our representation of CNF-formulas described in 
Section~\ref{sec:prel}, a natural encoding
can be based on a many-sorted structure $\mfA_\mcP$ whose universe is 
partitioned into sets 
of polynomials, (multilinear) monomials, variables, and rational 
coefficients that occur in $\mcP$.
As usual, we represent rationals as fractions of integers using 
binary encoding. Hence, $\mfA_\mcP$  provides a linear order of 
sufficient length to encode these binary strings.
Again, the exact technical details  are not important, as long as the 
encoding has some natural properties, such as $\FO$-definability of the 
class of valid encodings.
By a slight abuse of notation, we also denote by $\monPCx k$ the class of 
structures $\mfA_\mcP$ which encode a system $\mcP$ of polynomials over 
$\mbQ$ which can be refuted in  $\monPC_k$.

\begin{thm}\label{thm:monpc:in:fpc}
 For every $k \geq 1$, $\monPC_k \in \FPC$.
\end{thm}

Given a set of multilinear polynomials $\mcP$ of degree at most $k$, we 
consider the set $V_\mcP = \monPCx k(\mcP)$
of multilinear polynomials 
which can be derived 
from $\mcP$ using $\monPCx{k}$.
The first observation is that $V_\mcP$ is a $\mbQ$-linear space.
This easily follows since we can take $\mbQ$-linear combinations of 
polynomials that we derived.
Now, since this vector space $V_\mcP$ only contains multilinear 
polynomials of degree at 
most~$k$, 
we can naturally associate polynomials $p \in V_\mcP$ with vectors $p \in 
\mbQ^{M_k}$ where the index set $M_k$ denotes the set of all multilinear 
monomials of degree at most $k$.
For fixed $k \geq 1$, this set $M_k$ is of polynomial size $n^{\mcO(k)}$.

To prove Theorem~\ref{thm:monpc:in:fpc} we are going to express in $\FPC$
an inductive algorithm,  that is based on a similar algorithm for 
the full 
polynomial calculus from~\cite{CleggEdmImp96},
for computing a generating set for the $\mbQ$-linear space 
$V_\mcP$. Then, in order to see whether $\monPCx{k}$ can refute the 
system~$\mcP$, we simply check whether the constant polynomial $1$ 
is contained in $V_\mcP$, see Figure~\ref{basisAlgo}.

\begin{figure}[h]
\begin{algorithmic}

  \STATE
  \REQUIRE Set of multilinear polynomials $\mcP \subseteq \mbQ^{M_k}$
   
  \ENSURE $\mcB \subseteq \mbQ^{M_k}$ such that $\langle \mcB \rangle = 
\monPCx k(\mcP)$
\STATE~\COMMENT{where $\langle \mcB \rangle$ denotes the $\mbQ$-linear 
subspace generated by $\mcB$}

  \STATE $\mcB := \{\ML(m \cdot p) \ | \ p \in \mcP, m \text{ a monomial 
such that} \,\deg(\ML(m \cdot p)) \leq k\}$
\STATE~\COMMENT{Initialisation (lift all axioms in $\mcP$ up to degree $k$)}

 \REPEAT 
  \FORALL{monomials $m  \in \langle \mcB \rangle$, 
$\deg(m) < k$} 
\STATE{$\mcB := \mcB \cup \{ \ML(X \cdot m) : \text{for some variable 
$X$} \}$} 
\ENDFOR 
 \UNTIL{$\mcB$ remains unchanged}
  \RETURN $\mcB$
\end{algorithmic}
\caption{$\FPC$-procedure to define generating set for $V_\mcP = 
\monPCx{k}(\mcP)$}
\label{basisAlgo}
\end{figure}

During the run of the algorithm we iteratively construct a set $\mcB 
\subseteq V_\mcP$ of polynomials such that $\langle \mcB \rangle \leq 
V_\mcP$. Here, $\langle \mcB \rangle$ denotes the $\mbQ$-linear subspace 
generated by the polynomials in $\mcB$ (considered as $M_k$-vectors over 
$\mbQ$). Moreover, we ensure that at termination we have $\langle 
\mcB \rangle = V_\mcP$, see~Figure~\ref{basisAlgo}.
One important point to observe is that after the initialisation step we 
only add 
\emph{monomials} to the set $\mcB$. This closure operation is sufficient 
for the monomial-PC, since, except for the given axioms in $\mcP$ of which
we take care at initialisation, we can only use the multiplication (or 
lifting) rule for monomials.
Since there are only polynomially many 
different monomials of degree at most $k$, for a fixed $k$, this means 
that 
the algorithm is guaranteed to terminate after a polynomial number of 
iterations.

It is not obvious how to express this algorithm in $\FPC$. Most 
steps, such as 
the representation of the set $\mcB$ and the multilinearisation of 
polynomials, are easy to formalise, but there is a severe obstacle 
hidden in the condition for the main loop. Here, we want to iterate, in 
parallel, through all monomials $m \in \langle \mcB \rangle$. This 
condition ``$m \in \langle \mcB \rangle$''  translates to solving a linear 
equation system over $\mbQ$.
Although it is provably impossible to express the method of Gaussian 
elimination in 
$\FPC$, since it requires arbitrary choices during its computation, and 
although 
$\FPC$ cannot define 
the solvability of linear equation systems over finite 
fields \cite{AtseriasBulDaw09}, it is known \cite{DawarGroHolLau09}
that $\FPC$ can indeed express solvability of
linear equation systems over the rationals, see also 
Subsection~\ref{subsec:deflinFPC}.

\begin{thmC}[\cite{DawarGroHolLau09}]
\label{thm:solvq:fpc}
The solvability of linear equation systems over $\mbQ$ is definable in 
$\FPC$. 
\end{thmC}

Using this result we can express the algorithm from Figure~\ref{basisAlgo} 
in $\FPC$. In order to 
complete our proof of 
Theorem~\ref{thm:monpc:in:fpc} we just need to recall that $\monPCx k$ can 
refute 
$\mcP$ 
if, and only if, $1 \in \langle \mcB \rangle = V_\mcP$. 
This last assertion, again, reduces to deciding the solvability of a 
linear equation system over $\mbQ$ and it can thus, by 
Theorem~\ref{thm:solvq:fpc}, be defined in $\FPC$.

\subsection{Monomial-PC captures Fixed-Point Logic with Counting}
\label{subsec:fpc:inmonpc}
Next we show that the monomial-PC can simulate fixed-point logic with 
counting.
We first observe, however, that the logic $\FO(\monPCx k)$ does 
\emph{not} suffice for this purpose. This is due to the fact that $\FPC$ 
has access to the
second (numeric) sort, on which it can perform arbitrary polynomial-time 
computations, whereas $\FO(\monPCx k)$ is 
evaluated over standard single sorted input structures. 
To overcome this mismatch we have to extend the logic $\FO(\monPCx k)$ 
to the second-sorted framework as well. 
We denote this extension of $\FO(\monPCx k)$ by 
$\FOnum(\monPCx k)$.
As in the case of $\FPC$, this means that formulas are evaluated over 
extensions $\mfA^+$ of relational 
structures $\mfA$ by a numeric sort, as defined in 
Section~\ref{sec:prel}. In particular, interpretations for the Lindström 
quantifiers can make use of the second numeric sort, and we require this 
capability in the proof of our following result.

\begin{thm}\label{thm:fpc:in:monPC}
 For every $k \geq 2$, $\FPC \leq \FOnum(\monPCx k)$.
\end{thm}

An elegant way to prove Theorem~\ref{thm:fpc:in:monPC} is to use  a 
game-theoretic characterisation of $\FPC$ which was recently established  
in~\cite{GraedelHeg16}.
It is based on the notion of so-called \emph{threshold games}.
A \emph{threshold game} is a two-player game  played on a 
directed graph $G = (V,E)$ that is equipped with a 
\emph{threshold function} $\theta \colon V \to \mathbb{N}$. 
This function satisfies that $\theta(v)\leq \delta(v)+1$ for 
all $v \in V$, where $\delta(v)$ denotes the out-degree of $v$ in $G$.
Moreover, there is a designated vertex $s \in V$ at which each play 
starts.
A play is a sequence of $G$-nodes that arises 
according to the following rules. At the current position $v \in 
V$, Player~0 first selects a set $X \subseteq vE = \{ w : (v,w) \in E\}$ 
with $|X| \geq \theta(v)$. Then Player~1 chooses a node $w \in 
X$ and the play moves on to $w$. A player who cannot move loses. Hence 
Player 0 wins at all nodes in $T_0 := \{v \in V \ | \ \theta(v) = 0\}$ and 
Player 1 at all nodes in $T_1 := \{v \in V \ | \ \delta(v) < 
\theta(v)\}$.

In~\cite{GraedelHeg16} it is shown that threshold games provide 
appropriate model-checking games $\Tt(\AA,\phi)$
for any finite structure $\AA$ and any formula $\phi\in \FPC$. 
Since fixed-point evaluations on finite structures can be 
uniformly unraveled to first-order evaluations, we can in fact assume that 
the game graphs 
of these threshold games are acyclic. 
For any fixed \FPC-formula $\phi$, these model checking games
are polynomially bounded in the size of the input structure and can, in 
fact,
be interpreted in (two-sorted) input structures using a first-order 
interpretation.
This is related to the transformation of $\FPC$-formulas into
uniform families of polynomial-size threshold circuits, as used 
for instance in \cite{Otto97} and \cite{AndersonDaw17}.

\begin{thmC}[\cite{GraedelHeg16}]
\label{thm:fpc:to:games}
For every $\FPC$-formula $\phi$ there is a 
first-order interpretation $I_\phi$ which, for every finite structure 
$\mfA$, 
interprets in $\AA^+$ an acyclic threshold game 
$\Gg(\AA,\phi)$ such that
$\mfA \models \phi$ if, and only if, Player 0 has a winning strategy for  
$\Gg(\AA,\phi)$.
\end{thmC}

It remains to show that the monomial-PC can define winning regions in 
acyclic threshold games. 
Given an acyclic threshold game $\mcG = (G = (V,E),\theta)$, we 
construct an axiom system $\mcP(\mcG)$ which consists of polynomial 
equations of degree at most two.
For every node $v \in V$ in the threshold  game $\mcG$, the 
system $\mcP(\mcG)$ contains a variable $X_v$. 
Let us denote by $W_\sigma^\mcG$ the winning region of Player~$\sigma$ in 
$\mcG$. Then $\mcP(\mcG)$ satisfies the following:
\begin{itemize}
 \item if $v \in W^\mcG_0$, then $X_v = 1$ is derivable from $\mcP(\mcG)$ 
in $\monPCx 2$;
  \item if $v \in W^\mcG_1$, then $X_v = 0$ is derivable from $\mcP(\mcG)$ 
in 
$\monPCx 2$;
  \item $\mcP(\mcG)$ is consistent; in particular, either $X_v = 1$ or 
$X_v = 0$ is derivable for every $v \in V$;
\end{itemize}

\noindent
If we can construct such a system $\mcP(\mcG)$ via an
$\FO$-interpretation in $\mcG$, then this completes 
our proof of Theorem~\ref{thm:fpc:in:monPC}. In fact, it then
follows that $\FOnum(\monPCx 2)$ can define winning regions in acyclic 
threshold  games: a node $v \in V$ is in the winning 
region of Player~0 if, and only if, the system $\mcP(\mcG) \cup \{X_v = 
0\}$ can be refuted in $\monPCx 2$.

Recall that $vE = \{ w \in V: (v,w) \in E\}$, for $v \in V$, denotes 
the set of successors of~$v$. Further, we let $\ns(v)$ denote the number 
of 
successors of $v$, and we let $\nws(v)$ denote the number of successors of 
$v$ which 
are in the winning region of Player~0, that is 
$\ns(v) = | vE |$ and 
$\nws(v) = | vE \cap W_0^\mcG |$.
We denote the set of non-terminal positions by $\NT = \{ v \in V : \ns(v) 
> 
0\}$.
The system $\mcP(\mcG)$ uses the following set of variables:
\begin{itemize}
 \item a variable $X_v$, for every $v \in V$,
 \item a variable $Y_{v}^{m}$ for every $v \in \NT$, 
and $0 \leq m \leq \ns(v)$,
 \item a variable $Z_{v}^{m}[u \mapsto j]$ for every $v \in \NT$, $1 \leq 
m \leq \ns(v)$, $1 \leq j \leq m$, $u \in vE$.
\end{itemize}
The intuition is that the variables $X_v$ encode the winning regions of 
both players, as described above. Moreover, the variables $Y_v^{m}$ 
should indicate whether $\nws(v) = m$, in the following way: if $\nws(v) 
\neq m$, then $Y_v^{m} = 0$ is derivable, and if $\nws(v) = m$, then 
$Y_v^{m} = 1$ is derivable.
The variables $Z_{v}^{m}[u \mapsto j]$ are auxiliary variables used to   
encode this last condition, cf.~\cite{BerkholzGro15}.
The system $\mcP(\mcG)$ consists of the following axioms:
\begin{align*}
 \text{(T)}\qquad &\text{For } v \in T_0:  X_v = 1  \text{ and for  } v 
\in T_1:  X_v = 0 && \\
 \text{(C)}\qquad
 & \text{For } v \in \NT, 1 \leq m \leq \ns(v), u \in vE :
 &&\sum_{j = 1}^{m} Z^m_{v}[u \mapsto j] - Y_{v}^m = 0 \\ 
 & \text{For } v \in \NT , 1 \leq m \leq \ns(v), 1 \leq j 
\leq m:
 &&\sum_{u \in vE} X_u Z^m_{v}[u \mapsto j] - Y_{v}^m = 0 \\
 &\text{For } v \in \NT:
 &&\sum_{u \in vE} X_u \cdot Y_{v}^0 = 0 \\
 \text{(E)}\qquad &
 \text{For } v \in V: 
(1 - X_v) - \sum_{m = 0}^{\theta(v)-1} Y_{v}^{m} = 0  \,  \text{   
 and  } \,   
X_v \hspace{-3pt}  && \hspace{-3pt}-   \sum_{m = \theta(v)}^{\ns(v)}   
Y_{v}^{m} = 0
 \end{align*}

\noindent
We also add for each variable $X=X_v$, $v \in V$, a syntactic dual 
variable $\bar X$ together with the axiom 
\[ \text{(N)}\,\, 1 - X - \bar{X} = 0.\]
These axioms enforce that each dual variable $\bar X$ takes as value $1 - 
X$.
Note that the system $\mcP(\mcG)$ only contains axioms of degree at most 
$2$.

\begin{lem} The system $\mcP(\mcG)$ is consistent.
\end{lem}  
\begin{proof}
 We define an intended model of $\mcP(\mcG)$.
For   $X$-variables, we set $X_v := 1$, if $v \in W_0^\mcG$, and $X_v := 
0$, 
if $v \in W_1^\mcG$. 
For   $Y$-variables, we set $Y_{v}^{m} := 1$, if $\nws(v) = m$, and 
$Y_{v}^{m} := 0$ if $m \neq \nws(v)$.
For   $Z$-variables, we set $Z_{v}^{m}[u \mapsto j] := 
0$ for all non-terminal positions $v \in V$, $u \in vE$, and $j \in \{1, 
\dots, m \}$, if $m \neq \nws(v)$. 
For $m = \nws(v) > 0$, we let $vE \cap W_0^\mcG = \{u_1, \dots, u_m\}$.
We then set $Z_v^{m}[u_i \mapsto j] = 1$ if $j = i$, and $Z_v^{m}[u_i 
\mapsto j] = 0$ for $j \neq i$. Moreover, for $u \in vE \setminus 
W_0^\mcG$, we set $Z_v^{m}[u \mapsto 1] = 1$, and $Z_v^{m}[u \mapsto j] = 
0$ for $j \in \{2, \dots, m\}$. 
\end{proof}

\begin{lem}
If $v \in W_0^\mcG$, then we can derive the polynomial $X_v - 1$ (that is 
the equation $X_v = 1$) from $\mcP(\mcG)$ in 
$\monPCx 2$; and if $v \in W_1^\mcG$, then the polynomial
$X_v$ (that is the equation $X_v = 0$) can be derived from $\mcP(\mcG)$ 
in $\monPCx 2$.
\end{lem} 
\begin{proof}
We start with a small remark. Assume that we can derive $(1 - X)$   for a 
variable $X = X_v$, $v \in V$. We show how to derive $W(1 - X)$ 
for any variable $W$.
 This is clearly possible in the full polynomial calculus: we just have to 
multiply by $W$. In the monomial-PC, however, 
we cannot multiply $(1 - X)$ by $W$, since $(1-X)$ is neither a monomial 
nor an axiom. Instead, we use our negation axioms.
Starting from $1 - X$, we can derive $\bar{X}$ by subtracting (N) from 
$1-X$.
 Since (N) 
is an axiom, we can multiply it by $W$; also, $\bar{X}$ is a 
monomial and so we can multiply it by $W$. Thus, $W(1 - X - 
\bar{X}) + W \bar{X} = W(1 - X)$ can be derived, as claimed. We make use 
of this trick in the following.

Our proof is by induction on the height of the subgame rooted 
at $v \in V$ (recall that $\mcG$ is acyclic).
For terminal positions $v \in V$, the assertion is immediate from
axioms (T).

Assume $v \in V$ is a non-terminal position. 
Let $W_0(v) = vE \cap W_0^\mcG$ and $W_1(v) = vE \cap W_1^\mcG$.
By the induction hypothesis we know that we can derive in $\monPCx{2}$ for 
every $u \in W_0(v)$ the equation $X_u = 1$ and for every $u \in W_1(v)$ 
the equation $X_u = 0$. 

Let $m > 0$. 
Consider an equation of the form 
$\sum_{u \in vE} X_u Z^m_{v}[u \mapsto j] - Y_{v}^m = 0$ for $j \in \{1, 
\dots, m\}$ of type (C).
We have $vE = W_0(v) \uplus W_1(v)$. For every $Z$-variable and for every 
$u\in W_0(v)$ we can derive $ZX_u = Z$ in 
$\monPCx 2$, and for every $u \in W_1(v)$ we can derive $ZX_u = 0$ in 
$\monPCx 2$.
Hence, we can simplify these equations of type (C) as
$\sum_{u \in W_0(v)} Z^m_{v}[u \mapsto j] - Y_{v}^m = 0$ for $j \in 
\{1, \dots, m\}$ in $\monPCx{2}$.

Next, we consider for every $u \in W_0(v)$ the equations 
$\sum_{j = 1}^{m} Z^m_{v}[u \mapsto 
j] - Y_{v}^m = 0$, again of type (C). We combine these two sets of 
equations as follows:
\[ \sum_{j =1}^m \left(\sum_{u \in W_0(v)} Z^m_{v}[u 
\mapsto j] - Y_{v}^m \right ) 
- \sum_{u \in W_0(v)} \left( \sum_{j = 1}^{m} Z^m_{v}[u \mapsto 
j] - Y_{v}^m
\right) = 0.
\]
We can further simplify this equation (the variables $Z_v^m[u \mapsto 
j]$ cancel out) and we get 
\[ (m - \nws(v)) Y_v^m = 0.\]
Hence, for every $m > 0$, $m \neq \nws(v)$, we can derive $Y_v^m = 0$ in 
$\monPCx{2}$.
Indeed, also in the case where $m = 0 < \nws(v)$ we can derive $Y_v^m = 0$.
In this case we just use the equation
$\sum_{u \in vE} X_u Y_v^0 = 0$. 
Using the same arguments as above, this equation simplifies to $\nws(v) 
\cdot Y_v^0 = 0$. Hence, if $\nws(v) > 0$, we can also derive $Y_v^0=0$.
Note that the two equations of type (E) can be combined to the equation
$\sum_{m=0}^{s(v)} Y_v^m = 1$.
Hence, altogether we showed the following. For all $0 \leq m \leq s(v)$ it 
holds that:
\begin{itemize}
\item if $m = \nws(v)$, then we can derive $Y_v^m = 1$ in $\monPCx 2$; 
and
 \item if $m \neq \nws(v)$, then we can derive $Y_v^m = 0$ in $\monPCx 2$.
\end{itemize}
Having this, the claim follows immediately by using the equations 
of type $(E)$.
\end{proof}

In summary, we have seen that defining the winning regions in acyclic 
threshold games is an $\FPC$-complete problem, with respect 
to $\FOnum$-reductions, and that the winning regions in 
such games can be defined in $\FOnum(\monPCx{2})$. 
Furthermore, it is easy to see that the system $\mcP(\mcG)$ can be 
obtained from the game $\mcG$ by means of an $\FO$-interpretation.
This completes the proof of Theorem~\ref{thm:fpc:in:monPC} and, together 
with Theorem~\ref{thm:monpc:in:fpc}, establishes our first main theorem of 
this section.

\begin{thm}
 For every $k \geq 2$, $\FPC = \FOnum(\monPCx{k})$.
\end{thm}

\subsection{\texorpdfstring{$\FPC$}{FPC}-Definability of Refutations in the (Full) Polynomial 
Calculus}
\label{subsec:fullpc:infpc}
Next, we are going to lift our result concerning the degree-$k$ monomial-PC to the full degree-$k$ polynomial calculus. As we mentioned before, it seems implausible that proof-search for $\PC_k$ can be implemented in $\FPC$, since there are instances where such refutations necessarily contain polynomials with coefficients of super-polynomial bit-complexity (and $\FPC \leq \ptime$). Nevertheless, we will provide an $\FPC$-definable proof search procedure, similar to the one in the previous section, but it will only be able to deal with coefficients of restricted size. To this end, we define for each constant $b \in \bbN$, $\PC_{k,b}$ as the fragment of degree-$k$ polynomial calculus over $\bbQ$ where all coefficients are representable as fractions of binary numbers with at most $n^b$ many bits. In a next step, we see that, if we drop the restriction on the coefficients, we can still define the proof search algorithm in $\Cinfx{k}$.\\
It follows that the degree-$k$ variants $\monPCx k$ 
and $\PC_{k,b}$ (for each constant $b$) of the 
monomial-PC and the full polynomial calculus have the same 
expressive power (with respect to $\FOnum$-interpretations), that is for all 
$k\geq 2$ and $b \in \bbN$, we have
\[ \FOnum({\monPCx k}) = \FPC = \FOnum(\PC_{k,b}).\]
\smallskip
Our result provides an interesting new characterisation of the power of 
the (full) polynomial calculus from the perspective of finite model theory. In 
particular, it allows us to use techniques from finite model theory 
to answer open questions about the (relative) power of the two 
variants of the polynomial calculus. 
As indicated before, one example is given in Section~\ref{sec:gi} where we 
use 
our new characterisation of the polynomial calculus to answer an open 
question posed by Grohe and Berkholz in~\cite{BerkholzGro15}, see 
Question~\ref{ques:monPCvsPC} and Theorem~\ref{thm:monPCvsPC:GI}.

To prove the equivalence of $\FPC$ and $\FOnum(\PC_{k,b})$, the first 
important step is to understand why it is more difficult to 
express $k$-dimensional refutations in the (full) polynomial calculus in 
$\FPC$ rather than 
in its restricted variant $\monPC$.
Basically, this comes down to the following problem: in order to find 
proofs in the monomial-PC it suffices to decide the \emph{solvability 
problem} for linear equation systems over $\mbQ$ (this is a Boolean 
decision problem; the output is either \emph{solvable} or \emph{not 
solvable}). However, in order to search for proofs in the full PC we need 
to express the functional problem of \emph{computing solution 
spaces} of linear equation system over $\mbQ$ in $\FPC$. 
However, while it was known that 
$\FPC$ can define the (Boolean) solvability problem over $\mbQ$, it was 
not known whether solution spaces of linear equation systems over 
$\mbQ$ can be expressed in $\FPC$. Luckily, as we show in 
Theorem~\ref{thm:fpc:definesolutions}, this is indeed the case.

Let us now elaborate more on how to find refutations in the full PC.
To this end, we recall the procedure from Figure~\ref{basisAlgo} to find 
$k$-dimensional proofs in the monomial-PC.
Given a set of multilinear polynomials $\mcP \subseteq \mbQ^{M_k}$ of 
degree at most $k$, the idea is to construct a set $\mcB \subseteq 
\mbQ^{M_k}$ of (multilinear) polynomials of degree at most $k$ which 
generate (as $\mbQ$-linear combinations) the set of \emph{all} polynomials 
$\DeriveMonPCx k(\mcP)$ that can be derived 
in the $k$-dimensional $\monPC$ (starting from the given set of 
polynomials~$\mcP$). 
At the beginning, $\mcB$ is set to contain 
all (linearised and) lifted versions $\ML(m \cdot p)$ of the given 
polynomials $p \in \mcP$ up to degree~$k$.
Subsequently, the set $\mcB$ is  closed under liftings by variables $X$.
More precisely, in each iteration, the set $\mcB$ is extended by 
\emph{all} possible 
(linearised) liftings $X \cdot m$ of \emph{monomials} $m$ of degree at 
most $k-1$ 
that can be derived up to this stage, i.e.\ for which $m \in \langle \mcB 
\rangle$ holds (here, $\langle \mcB \rangle$ denotes the set of all 
polynomials that can be derived from polynomials in $\mcB$ using 
$\mbQ$-linear combinations).
The crucial observation is that this simple inductive lifting step is 
sufficient for the monomial-PC,  because, indeed, by its rules we are only 
allowed to lift monomials and the initial polynomials $p \in \mcP$.
The set $\mcB$ is extended in this way until $\langle \mcB 
\rangle$ remains stable.

In order to adapt this algorithm to the (full) PC, we need to make the 
following changes.
Most importantly, instead of lifting all monomials $m  \in \langle 
\mcB 
\rangle$, $\deg(m) < k$, during the iteration, for the full PC we have to 
take all \emph{(multilinear) polynomials} $p  \in 
\langle \mcB \rangle$, $\deg(p) < k$ into account, and make sure that 
their liftings $X \cdot p$ are contained in $\langle \mcB \rangle$.
This is more difficult for the following two reasons.
First of all, we cannot go through \emph{all} such polynomials 
$p$, simply because their number is exponential in the number of 
variables. To overcome this obstacle, we have to use linear-algebraic 
preprocessing which enables us to lift a generating set for the set of 
polynomials $p \in \langle \mcB \rangle, \deg(p) < k$, instead. Note 
that this was not necessary in the setting of the monomial-PC: 
here, the number of possible $k$-dimensional monomials is bounded 
polynomially in 
the number of variables (for fixed $k\geq 2$).
There is a second problem. During the iteration, for the monomial-PC we 
could repeatedly add \emph{all} lifted variants $X \cdot m$ of all 
monomials $m 
\in \langle \mcB \rangle, \deg(m) < k$ to our partial generating set 
$\langle \mcB \rangle$. This is because the (linearised) 
version of a lifted monomial remains a monomial and we just said that the 
number of all $k$-dimensional monomials is polynomially bounded. Hence, we 
never obtain generating sets $\mcB$ of super-polynomial 
size in this way.
In contrast, for the setting of the (full) PC, assume that at some stage 
during the iteration we have a small generating set $\mcC$ for the set of 
all 
polynomials $p \in \langle \mcB \rangle$ of degree at most $k-1$. If 
we now lift \emph{all} polynomials $p \in \mcC$ in all possible ways $X 
\cdot p$, then clearly the size of the resulting set $\mcC'$ 
increases 
by a factor which corresponds to the number of variables (and there is no 
global polynomial upper bound for $\mcC'$ as in the case of the 
monomial-PC).
Hence, before each lifting step, we have to ensure that the size of the 
generating set $\mcC$ of polynomials that we lift is (globally) bounded by 
a polynomial. We can invoke standard linear-algebraic algorithms to 
achieve this. More specifically, we construct $\mcC$ in such a way that 
its 
size 
does not exceed $|M_k|$, that is the number of different multilinear 
monomials of degree at most~$k$. Note that a generating set of this size 
exists, since each $\mbQ$-linear subspace of $\mbQ^{M_k}$ is of 
dimension at most $|M_k|$. Moreover, as we mentioned before, $|M_k|$ is of 
polynomial size for any fixed $k$.
We summarise the adapted algorithm for finding $k$-dimensional refutations 
in 
the (full) polynomial calculus in Figure~\ref{fullPCAlgo}.

\begin{figure}[h]
\begin{algorithmic}

  \STATE
  \REQUIRE Set of multilinear polynomials $\mcP \subseteq \mbQ^{M_k}$
   
  \ENSURE $\mcB \subseteq \mbQ^{M_k}$ such that $\langle \mcB \rangle = 
\DerivePC(\mcP)$.

 \STATE $\mcB := \{\ML(m \cdot p) \ | \ p \in \mcP, m \text{ a monomial 
such that}\, \deg(\ML(m \cdot p)) \leq k\}$
 \STATE\COMMENT{Initialisation (lift all axioms in $\mcP$)}
 
 \REPEAT 
  \STATE Construct set $\mcC \subseteq \mbQ^{M_k}$ of size at most 
$|M_k|$ such   that 
$ \langle \mcC \rangle = \{ p \in \langle \mcB \rangle : \deg(p) < k \}$
\FORALL{polynomials $p  \in \mcC$} 
\STATE{$\mcB := \mcB \cup \{ \ML(X \cdot p) : \text{for some variable 
$X$} \}$}
\ENDFOR 
 \UNTIL{$\langle \mcB \rangle$ remains unchanged}
  \RETURN $\mcB$
\end{algorithmic}
\caption{\FPC-procedure to construct generating set $\mcB$ for the set 
$\DerivePC(\mcP)$ of all polynomials that can be derived in $\PCx k$ 
starting from the given set of polynomials $\mcP$}
\label{fullPCAlgo}
\end{figure}

To see how we can implement the algorithm from Figure~\ref{fullPCAlgo} in 
polynomial time, let us have a closer look at the construction of the set 
$\mcC$ during the iteration.
First of all note that the set $\{ p \in \langle \mcB \rangle : \deg(p) < 
k \}$ is indeed a $\mbQ$-linear subspace of $\langle \mcB \rangle$ which, 
in turn, is a $\mbQ$-linear subspace of $\mbQ^{M_k}$.
Hence, it is clear that a generating set $\mcC$ of size at most $|M_k|$ 
exists.
Moreover, we can easily obtain $\mcC$ as the solution space of a linear 
equation system.
Indeed, let $M$ be the $M_k \times \mcB$-matrix over $\mbQ$ whose columns 
correspond to the polynomials in~$\mcB$. Then $\im(M) = \langle \mcB 
\rangle$.
Hence, if we let $x$ and $p$ denote a $\mcB$-vector and an $M_k$-vector 
of variables ranging over $\mbQ$, respectively, then the solution space of 
the linear 
equation system determined by the equation $Mx = p$ is 
$\langle \mcB \rangle$ when we project it to the variables~$p$.
Hence, by adding extra constraints $p(m) = 0$ for all monomials $m \in 
M_k$ with $\deg(m) = k$, we obtain a linear equation system whose 
solution space, projected to variables in $p$, is a generating set 
for $\{ p \in \langle \mcB \rangle : \deg(p) < k \}$.
Clearly, solution spaces for such systems can be computed in polynomial 
time.

Before we discuss the $\FPC$-definability of this procedure, let us observe 
that there is a small caveat with the approach above. So far, the 
generating set for $\{ p \in \langle \mcB \rangle : \deg(p) < k \}$ that we 
obtain is not of size at most $|M_k|$. Indeed,
by our construction, which relies on the final projection step, 
the size of the generating set depends on $|\mcB|$ (because 
the vector of variables $x$ is indexed by $\mcB$).
Hence, we need to make a second important observation.
Say we were able to construct an $M_k \times J$-matrix $N$ over $\mbQ$
with the property that $\im(N) = \{ p \in \langle \mcB \rangle : \deg(p) < 
k \}$ (that is the columns of $N$ form a generating 
set for the solution space of the above linear equation system projected to 
$p$).
We would like to transform this matrix $N$ into a ``smaller'' $M_k \times 
M_k$-matrix $\hat N$ such that $\im(\hat N) = \im(N)$. This is clearly 
possible simply because the dimension of the space $\im(N)$ is at most $|M_k|$. 
However, the question is about how difficult it is to obtain such a
``more compact'' version $\hat N$ of~$N$. Specifically, for our
$\FPC$-definability proof, we need to express this ``compression 
transformation'' in $\FPC$ as well.

Fortunately, the step from $N$ to $\hat N$ is surprisingly easy to realise. 
As we will see in the following subsection, it holds that the 
$(M_k \times M_k)$-matrix $\hat N  := N \cdot N^T$ has the same 
image as the matrix $N$, see 
Lemma~\ref{lem:bcmmt}. Hence, we obtain a \emph{small} generating set for 
$\im(N)$ by 
taking 
the columns of $\hat N = N \cdot N^T$.
This shows that we can, in general, quite easily transform an arbitrary 
generating set for a $\mbQ$-linear subspace of $\mbQ^{M_k}$ into a 
small generating set of size at most $|M_k|$. Moreover, this 
transformation only relies on simple matrix operations, such as 
transposition and matrix multiplication. As such operations are well-known 
to be definable in $\FPC$, see e.g.\ 
\cite{Holm10}, this transformation is $\FPC$-definable. However, it is at this point that the bit-complexity of the coefficients has to be taken into account. Since $\hat N$ is computed by squaring $N$, the bit-complexity of the coefficients can increase in this step. If this happens repeatedly, then the required number of bits may become greater than the maximum number of bits that our $\FPC$-sentence can handle. In this case, the computation has to be aborted. This maximum number of bits depends on the number of variables of the $\FPC$-sentence that we are constructing: We want our sentence to be able to find refutations in $\PC_{k,b}$, for a fixed value of $b$. That is, the coefficients occurring in a refutation can be written as fractions of binary numbers of length $\leq n^b$, where $n$ is the size of the input structure. These coefficients, i.e.\ the entries of the matrices that we are manipulating in the fixed-point computation, are represented as follows: We use a tuple of $b$ variables ranging over the $n$ ordered elements of the number sort in order to index the positions of a binary string. Relations are used to mark the positions that are $0$ and $1$, respectively, and to specify the position of the coefficient in the matrix (see \cite{Holm10} for more details). Therefore, it is possible to construct for every fixed $b$ an $\FPC$-sentence that performs the matrix manipulations mentioned above using binary numbers of length $n^b$, but no fixed $\FPC$-sentence can deal with coefficients of unbounded length.

\medskip
Altogether, this means that the only difficulty we face is to define 
solution spaces of linear equation system over $\mbQ$ in $\FPC$.
Recall that by the result of Dawar, Grohe, Holm, and Laubner we know
that $\FPC$ can express the \emph{Boolean} solvability problem of linear 
equation 
systems over $\mbQ$, see Theorem~\ref{thm:solvq:fpc}.
However, this does not give direct evidence for 
$\FPC$ being able to express the more general \emph{functional} problem of 
defining solution spaces over $\mbQ$. For the sake of illustration, consider rank logic over finite fields. Rank logic can define the Boolean solvability problem for linear equation systems over finite fields but it is not to be expected that it can also define vectors in the solution space: This is because any solution vector to a linear equation system obtained from CFI-graphs has an orbit of exponential size, and rank logic is isomorphism-invariant and in $\ptime$.\\ 
Luckily, over $\bbQ$, the situation turns out to be different. Not only can we define the Boolean solvability problem in $\FPC$ but we can also define the corresponding solution spaces as we show in this article (see 
Theorem~\ref{thm:fpc:definesolutions}).
From this result and our preceding discussion it easily follows that the 
algorithm in Figure~\ref{fullPCAlgo} (with bounded bit-complexity) is definable in $\FPC$. Beyond this 
application, we believe that 
Theorem~\ref{thm:fpc:definesolutions} is interesting in its own right and 
might prove useful in other contexts.
Let us conclude by stating our main result of this subsection (where we rely on 
the yet to be proven
Theorem~\ref{thm:fpc:definesolutions}).

\begin{thm}
\label{thm:PCinFPC}
 For every $k \geq 2$ and $b \in \bbN$, there exists an $\FPC$-sentence $\phi$ with 
$\mcO(k+b)$ many variables such that given (a structural encoding of) a 
system $\mcP$ of polynomials  over $\mbQ$ as input for 
the $k$-dimensional polynomial calculus of degree $k$, $\phi$ expresses 
whether $\mcP$ can be refuted in $\PC_{k,b}$, that is $\phi$ 
defines whether $1 \in \PC_{k,b}(\mcP)$. 
\end{thm}

A more commonly studied version of the polynomial calculus is $\PC_k$, that is, the degree-$k$ PC over $\bbQ$ without any restriction on the bit-complexity. The procedure we described above in principle also works for the $\PC_k$. It would be $\FPC$-definable, even without a bound on the bit-complexity, if $\FPC$-sentences were evaluated in structures with larger number sorts. Recall that we use the elements of the number sort to index the positions of the binary strings. In $\FPC$, the number sort always has the same size $n$ as the structure itself, but if we imagine the number sort to be of some size $f(n)$, for a sufficiently large function $f$, then our sentence can deal with $f(n)^b$ many bits instead of $n^b$. Our algorithm involves squaring a matrix polynomially many times. Hence, if $f(n)$ is greater than the largest possible growth of the bit-length that can occur in this number of squaring operations, then the procedure could be implemented in $\FPC$ with number sorts of size $f(n)$ and it would always correctly decide the existence of $\PC_k$-refutations, regardless of any bit-complexity issues. A fixed-point logic with such big number sorts does not really exist but instead, we can use $\Cinfx{k}$. The standard translation of $\FPC$-sentences into $\Cinf$-sentences does not increase the number of variables. It simulates the number-sort-variables of the $\FPC$-sentence with large disjunctions or conjunctions over all elements of the number sort. This idea works regardless of the size of the number sort. These considerations directly lead to the following result for the degree-$k$ polynomial calculus:

\begin{thm}
	\label{thm:PCinCinf}
	For every $k \geq 2$, there exists a $\Cinf$-sentence $\phi$ with 
	$\mcO(k)$ many variables such that given (a structural encoding of) a 
	system $\mcP$ of polynomials  over $\mbQ$ as input for 
	the $k$-dimensional polynomial calculus of degree $k$, $\phi$ expresses 
	whether $\mcP$ can be refuted in $\PC_{k}$, that is $\phi$ 
	defines whether $1 \in \PC_{k}(\mcP)$. 
\end{thm}

\begin{thm} 
\label{thm:FPCgleichPC}
For all $k \geq 2$, $b \in \bbN$:
 \[ \FOnum(\monPCx k) = \FOnum(\PC_{k,b}) = \FPC \overset{?}{<} \FOnum(\PC_k) < \Cinfx{\Oo(k)}.   \]
\end{thm}
By $\FPC \overset{?}{<} \FOnum(\PC_k)$, we mean that we do not know whether or not $\FOnum(\PC_k) \leq \FPC$ holds. However, there are some reasons why we suspect that $\FPC$ is strictly weaker than $\FOnum(\PC_k)$. First of all, to the best of our knowledge, it is an open problem whether or not there exists a PTIME-algorithm that decides the existence of $\PC_k$-refutations (for unbounded coefficients). The well-known Groebner basis algorithm certainly fails \cite{Hakoniemi21}, so if this problem is in P, then there must be some way to avoid explicit computation of the coefficients in the refutation. Since $\FPC\leq \ptime$, it is ``even more open'' if the problem is in $\FPC$.\\
Secondly, our result $\FOnum(\PC_{k,b}) = \FPC$ has the following consequence: If it were possible to compute $\PC_k$-refutations with arbitrarily large coefficients in $\FPC$, then there would be a numeric FO-interpretation that reduces any input polynomial equation system to one that can be decided in $\PC_{k,b}$, that is, with small degree and small coefficients. This seems to be a very strong statement because it means that the necessity to use large coefficients in refutations can be circumvented with simple FO-definable preprocessing of the input polynomials. This would be quite surprising, so it seems more reasonable to believe that $\FPC \lneq \FOnum(\PC_k)$.

\subsection{Definability of Solution Spaces of Linear Equation Systems 
over \texorpdfstring{$\mbQ$}{Q}}
\label{subsec:deflinFPC}
To complete our proof of Theorem~\ref{thm:PCinFPC}, we proceed to show 
that $\FPC$ can define solution spaces of linear equation systems 
over $\mbQ$.
Formally, our main result in this subsection reads as follows.

\begin{thm}
\label{thm:fpc:definesolutions}
 There exist $\FPC$-formulas which define the following:
given (a structural encoding of) a linear equation system $M\cdot x = b$ 
over $\mbQ$, for $M\colon I \times J \to \mbQ$ and $b\colon I \to \mbQ$,
they express whether $M \cdot x = b$ is solvable, and in this 
case, 
define (structural encodings of) a matrix $S\colon J \times J \to \mbQ$ 
and a vector $c: J \to \mbQ$ such that $\im(S) = \kernel(M)$ and $M \cdot 
c 
= b$, i.e.\ such that 
$\im(S) +c $ is the solution space of $M \cdot x = b$.
\end{thm}

In order to prove Theorem~\ref{thm:fpc:definesolutions}, we make use of 
the following linear-algebraic properties of matrices over the 
rationals. 
For completeness, and since it is central for our application, we 
present short proofs to recall the underlying algebraic arguments.
From now on, let us fix a linear equation system $M \cdot x = b$ with 
$M\colon I \times J \to \mbQ$ and $b\colon I \to \mbQ$.
The key is to consider the following matrices over $\mbQ$:
\begin{align*}
 B &:= M \cdot M^T \in \mbQ^{I \times I} \\
 C &:= M^T \cdot M \in \mbQ^{J \times J}.
\end{align*}
In Figure~\ref{fig:summary:bcmmt} on the next page we summarise what we are going to show.
\pagebreak

\begin{figure}[h]
\begin{tikzpicture}
 
  \node[font=\large] (rj) at (-3.5,0) {$\mbQ^J$};
  \node[below of=rj, node distance=0.7cm] {\rotatebox{90}{$=$}};
  \node[below of=rj, node distance=1.3cm] (kerm) {$\kernel(M)$};
  \node[below of=rj, node distance=2cm] (oplusrj) {$\oplus$};
  \node[below of=rj, node distance=2.7cm] (immt) {$\im(M^T)$};
  
  \node[left of=kerm, node distance=1.2cm] (kerme) {$=$};
  \node[left of=immt, node distance=1.2cm] (immte) {$=$};

  \node[left of=kerm, node distance=2.4cm] (kerme) {$\kernel(C)$};
  \node[left of=immt, node distance=2.4cm] (imc) {$\im(C)$};
  \node[left of=oplusrj, node distance=2.4cm] (oplusrj2) {$\oplus$};

  \node[font=\large] (ri) at (3.5,0) {$\mbQ^I$};
  \node[below of=ri, node distance=0.7cm] {\rotatebox{90}{$=$}};
  \node[below of=ri, node distance=1.3cm] (kermt) {$\kernel(M^T)$};
  \node[below of=ri, node distance=2cm] (oplusri) {$\oplus$};
  \node[below of=ri, node distance=2.7cm] (imm) {$\im(M)$};
  
  \node[right of=kermt, node distance=1.2cm] (kermte) {$=$};
  \node[right of=imm, node distance=1.2cm] (imme) {$=$};

  \node[right of=kermt, node distance=2.4cm] (kerme) {$\kernel(B)$};
  \node[right of=imm, node distance=2.4cm] (imb) {$\im(B)$};
  \node[right of=oplusri, node distance=2.4cm] (oplusrj2) {$\oplus$};
  
  \path[->=stealth, bend left=10] (immt) edge node[ fill=white, 
anchor=center, pos=0.5,font=\bfseries] {$\cong$} (imm);
  \path[->=stealth, bend left=10] (imm) edge node[ fill=white, 
anchor=center, pos=0.5,font=\bfseries] {$\cong$} (immt);

  \node[font=\large] (m) at (0.1,-1.6) {$M$};
  \node[font=\large] (m) at (0.1,-3.7) {$M^T$};
  

  \path (imc) edge [->=stealth, loop below, out=310, in=230,looseness=5] 
node[ fill=white, 
anchor=center, pos=0.5,font=\bfseries] {$\cong$} (imc);

\path (imb) edge [->=stealth, loop below, out=310, in=230,looseness=5] 
node[ fill=white, 
anchor=center, pos=0.5,font=\bfseries] {$\cong$} (imb);

 \node at (-6.7,-3.5) {$C$};
 \node at (6.7,-3.5) {$B$};

\end{tikzpicture}
\caption{Linear-algebraic structure induced by the matrix $M\colon I 
\times J \to \mbQ$ where $B = M \cdot M^T \colon I \times I \to \mbQ$ and 
$C = M^T \cdot M\colon J \times J \to \mbQ$}
\label{fig:summary:bcmmt}
\end{figure}

\begin{lem}[Properties of $B, C$]
\hfill
\label{lem:probBC}
\begin{enumerate}
 \item $B^T = B$ and $C^T = C$, that is $B$ and $C$ are symmetric.
 \item $\kernel(B) = \kernel(B^i)$ and $\kernel(C) = \kernel(C^i)$ for all 
$i \geq 1$.
 \item $\im(B) = \im(B^i)$ and $\im(C) = \im(C^i)$ for all $i \geq 1$.
 \label{lem:probBC:3}
 \item $\mbQ^I = \kernel(B) \oplus \im(B)$ and $\mbQ^J = \kernel(C) \oplus 
\im(C)$.
\label{lem:probBC:4}
 \item $B$ is an automorphism of $\im(B)$ and $C$ is an automorphism of 
$\im(C)$.
\label{lem:probBC:5}
 \end{enumerate}
\end{lem}
\begin{proof}
 The arguments for $B$ and $C$ are completely symmetric, so let us 
consider the case of $B\colon I \times I \to \mbQ$.
First of all, $B^T = (M \cdot M^T)^T = M \cdot M^T = B$.
For the second claim, we proceed by induction on $i$. It is clear that 
$\kernel(B) \subseteq \kernel(B^i)$ for all $i \geq 1$, so it suffices to 
show 
$\kernel(B^i) \subseteq \kernel(B)$. 
For $i= 1$, the claim is trivial, so assume that $i \geq 2$ and for some 
$x \in \mbQ^I$ we have $B^i x = 0$.
Then also  $B^{i-1} \cdot B^{i-1} x = 0$. Since $B^T = B$, this means that 
also $x^T \cdot (B^{i-1})^T \cdot B^{i-1} x = 0$. 
Hence, $|B^{i-1} x|^2 = 0$, which implies that $B^{i-1}x = 0$. We get $x 
\in \kernel(B^{i-1})$ and by the induction hypothesis $x \in \kernel(B)$.

Let's consider~(\ref{lem:probBC:3}). Again, it is easy to see that 
$\im(B^i) \subseteq \im(B)$ for all $i \geq 1$.
Now let's choose a basis $Be_1, \dots, B{e_\ell}$ for $\im(B)$ and let $i 
\geq 2$. Then $B^ie_1, \dots, B^ie_\ell$ is a generating set for 
$\im(B^i)$.
We claim that $B^ie_1, \dots, B^ie_\ell$ is a basis for $\im(B^i)$ which 
would prove our claim.
Indeed, assume that for some non-zero $(a_j) \in \mbQ^\ell$ we had 
$\sum_j a_j \cdot B^i e_j = 0$.
Since $\kernel(B^i) = \kernel(B)$ it follows that $\sum_j a_j B e_j = 0$, 
a 
contradiction. 
We turn our attention to~(\ref{lem:probBC:4}).
We have to show two things, namely that every vector in $\mbQ^I$ can be 
written as a linear combination of elements in $\kernel(B)$ and $\im(B)$ 
and 
that this expression is unique.
Let us start with the latter claim. Assume that $Bx +  y = 0$ for $x, y 
\in 
\mbQ^I$, $y \in \kernel(B)$. We have to show that $Bx = y = 0$.
From $Bx + y = 0$ we can conclude that $B^2x = 0$ since $By = 0$. Since 
$\kernel(B^2) = \kernel(B)$ it follows 
that $Bx = 0$ which yields $y = 0$.
To complete the proof, let $x \in \mbQ^I$. 
Consider the cyclic space generated by $B$, that is $\langle \{ x, Bx, 
B^2x, \cdots, B^nx \} \rangle$ where $n=|I|$. Note that $\{ x, Bx, 
B^2x, \cdots, B^nx \}$ is linearly dependent.
If $x \in \im(B)$, then there is nothing to show.
Otherwise, we know that $x \not\in \langle \{ Bx, B^2x, \cdots, B^nx \} 
\rangle$.
We choose a non-zero vector $(a_j) \in \mbQ^n$ such that
$\sum_{j=1}^n a_j B^jx = 0$. Let $k \geq 1$ be minimal such that $a_k \neq 
0$. 
Then $B^k (a_k x + \sum_{j=k+1}^n a_j B^{j-k} x) = 0$, hence $a_k x + z 
\in \kernel(B^k) = \kernel(B)$ for some $z \in \im(B)$. This is what we 
wanted 
to show.
Finally,  (\ref{lem:probBC:5}) follows from (\ref{lem:probBC:3}).
\end{proof}

\begin{lem}[Relating $B,C$ and $M, M^T$]
\label{lem:bcmmt}
\hfill
\begin{enumerate}
 \item $\kernel(B) = \kernel(M^T)$ and $\kernel(C) = \kernel(M)$.
 \label{lem:bcmmt:1}
 \item $\im(B) = \im(M)$ and $\im(C) = \im(M^T)$.
  \label{lem:bcmmt:2}
 \item $M$ is an isomorphism from $\im(M^T)$ to $\im(M)$ and $M^T$ an 
isomorphism from $\im(M)$ to $\im(M^T)$. In particular, $\rank(C) = 
\rank(M^T) = \rank(M) = \rank(B)$.
 \label{lem:bcmmt:3}
\end{enumerate}
\end{lem}
\begin{proof}
 Again, as the arguments are symmetric, we only consider the case of $B$.
 For~(\ref{lem:bcmmt:1}), note that $\kernel(M^T) \subseteq \kernel(M\cdot 
M^T) 
= \kernel(B)$ for trivial reasons. 
Moreover, if $M\cdot M^T x = 0$ for some $x \in \mbQ^I$, then also $x^T M 
M^T x = 0$, that is $| M^T x |^2 = 0$ which implies $M^Tx = 0$. Hence, 
$\kernel(B) \subseteq \kernel(M^T)$.
For~(\ref{lem:bcmmt:2}), note that $\im(B) = \im(M \cdot M^T) \subseteq 
\im(M)$, again for trivial reasons. To verify the other direction, let $x 
\in \mbQ^J$ and consider the element $Mx \in \im(M)$.
By Lemma~\ref{lem:probBC}, we can write $Mx$ as $Mx = By + z$ for some 
$y,z \in \mbQ^I$ and $z \in \kernel(B)=\kernel(M^T)$.
Hence $M^TMx = (M^TM)M^Ty$, that is $Cx = CM^Ty$. From this we get that 
$C(x - M^Ty) = 0$. Using $\kernel(C) = \kernel(M)$, we get $M(x-M^Ty) = 0$.
This implies that $Mx = MM^Ty = By$ which proves our claim.
Finally, (\ref{lem:bcmmt:3}) follows immediately 
from~(\ref{lem:bcmmt:1}),(\ref{lem:bcmmt:2}) and Lemma~\ref{lem:probBC}.
\end{proof}

We are ready to establish the following central criterion for the 
solvability of linear equation systems over $\mbQ$. 
\begin{lem}[Solvability of linear equation system, see 
also~\cite{GroheP17}]
\label{lem:solvFPC}
 Let $M \cdot x = b$  be a linear equation system over $\mbQ$ where with 
$M\colon I \times J \to \mbQ$ and $b \colon I 
\to \mbQ$.
Let $B = M \cdot M^T$ as above. Let $n = \minimum\{|I|, |J|\}$.
Then the linear equation system $M \cdot x = b$ is solvable if, and only 
if, $b$ can be written as a $\mbQ$-linear combination of vectors in 
$\Gamma = \{ Bb, B^2b, 
\dots, B^{n+1}b \}$, that is if $b \in \langle \Gamma \rangle$.
\end{lem}
\begin{proof}
 First, note that $\langle \Gamma \rangle \subseteq \im(B)$. Hence, if 
$b \in \langle \Gamma \rangle$, then clearly the linear equation system 
$M \cdot x = b$ is solvable.
For the other direction, assume that $b \in \im(M) = \im(B)$. Then for 
some $c \in \mbQ^I$ we have $Bc = b$.
Let $\Delta = \{Bc, B^2c, \dots, B^{n+1}c \} = \{ b, Bb, \dots, B^nb \}$. 
This set $\Delta$ is linearly dependent, because it is a subset of 
$\im(B)$ 
 and we have established before that
the dimension of $\im(B)$ coincides with $\rank(M)$ which is at most $n = 
\minimum\{ |I|, |J| \}$.
It easily follows that $\langle \Delta 
\rangle$ is $B$-invariant, that is $\langle B\Delta \rangle \subseteq 
\langle \Delta \rangle$. 
Moreover, by Lemma~\ref{lem:probBC}, $B$ is an automorphism of $\im(B)$ 
which implies that $\langle B\Delta \rangle = \langle \Delta 
\rangle$. However, this shows that $b \in \Delta$ can be written as a 
linear combination of elements in $B\Delta = \Gamma$ which proves our 
claim.
\end{proof}

Using  Lemma~\ref{lem:solvFPC}, it is easy to show that $\FPC$ can 
define the solvability problem for linear equation systems over
$\mbQ$.
\begin{proof}[Proof of Theorem~\ref{thm:fpc:definesolutions} - Part 1/2]
 Given a linear equation system $M \cdot x = b$ over $\mbQ$ for $M\colon I 
\times J \to \mbQ$ and $b\colon I \to \mbQ$, we first define $B = M M^T$ 
as above and the \emph{ordered} set of vectors 
$\Gamma = \{ Bb, \dots, B^{n+1}\}$ as in  Lemma~\ref{lem:solvFPC}. 
This can be done in $\FPC$, since matrix multiplication over $\mbQ$ is 
well-known to be definable in $\FPC$, see e.g.~\cite{Holm10}. 

Since $\Gamma$ is an ordered set we can use the Immerman-Vardi Theorem to 
define the problem $b \in \langle \Gamma \rangle$ in $\FPC$. More 
precisely, let $N$ be the $I \times \{1, \dots, n+1\}$-matrix over 
$\mbQ$ whose $i$-th column is the vector $B^ib\colon I \to \mbQ$.
Then $b \in \langle \Gamma \rangle$ if, and only if, $N \cdot x = b$ is 
solvable. 
Note that $N$ can be written as $N = B \cdot \hat N$ where $\hat N$ is 
the $I \times \{1, \dots, n+1\}$ matrix whose $i$-th column is $B^{i-1}b$.
Moreover, since $B = M \cdot M^T$, we have transformed our 
original system $M \cdot x = b$ into the system $M \cdot (M^T \cdot \hat 
N) \cdot x = b$ which is solvable if, and only if, $M \cdot x = b$ is 
solvable. 
Furthermore, a solution $c\colon \{ 1, \dots, n+1 \} \to \mbQ$ for $M 
\cdot (M^T \cdot \hat N) \cdot x = b$ readily defines the solution $(M^T 
\cdot \hat N \cdot c)\colon J \to \mbQ$ for $M \cdot x = b$

To solve the system $N \cdot x = b$ in $\FPC$, first note that $N$ has an 
ordered set of columns. However, the set of rows $I$ is not ordered. To 
obtain an ordered linear equation system, we can consider the 
lexicographical ordering on the set of rows of $N$ induced by the linear 
order on the set of columns and on $\mbQ$. This results in a linear 
preorder which merges columns that are identical (such columns correspond 
to repeated linear equations). By merging identical columns we obtain a 
fully ordered system. By the Immerman-Vardi Theorem such systems can 
be solved in $\FPC$.
\end{proof}

We are left with the second claim of 
Theorem~\ref{thm:fpc:definesolutions}, 
namely that, given $M \cdot x = b$ with $M\colon I \times J \to \mbQ$ and 
$b \colon I \to Q$, we can define in $\FPC$ a matrix $S\colon J \times 
J \to \mbQ$ whose columns form a generating set for $\kernel(M)$, that is 
$\im(S) = \kernel(M)$.
For this we make use of the structure induced by the linear transformation 
$C = M^T M\colon J \times J \to \mbQ$ on $\mbQ^J$.
\begin{proof}[Proof of Theorem~\ref{thm:fpc:definesolutions} - Part 2/2]
 We established in Lemma~\ref{lem:probBC} that $\mbQ^J = \kernel(C) \oplus 
\im(C)$ and in Lemma~\ref{lem:bcmmt} that $\kernel(M) = \kernel(C)$.
Hence, we can equivalently define a generating set for $\kernel(C)$ in 
$\FPC$.
Let us denote by $e_j$ the $j$-th standard basis vector on $\mbQ^J$.
We can clearly define the vector $e_j\colon J \to \mbQ$ in $\FPC$ using $j 
\in J$ as a parameter.
Since $\mbQ^J = \kernel(C) \oplus \im(C)$ we can write each $e_j$ uniquely 
as 
$e_j = k_j + c_j$ for $k_j \in \kernel(C)$ and $c_j \in \im(C)$.
It is easy to see that the set $\{k_j : j \in J\}$ forms a generating set 
for $\kernel(C)$.
Hence, our aim is to define this set in $\FPC$.

To obtain the projections $k_j$ of $e_j$ onto $\kernel(C)$, we make use of 
the fact that $\FPC$ can solve linear equation systems over $\mbQ$ and 
define single solutions.
Indeed, $k_j$ is the \emph{unique} vector $k_j \in \mbQ^J$ such that 
$e_j = k_j + c_j$ and $Ck_j = 0$ and $Cz = c_j$ for some $c_j,z \in 
\mbQ^J$ (where we treat $k_j, z, c_j$ here as $J$-vectors of variables 
ranging over $\mbQ$).
Since in each solution of this system the projection onto $k_j$ is unique, 
we can define $k_j$ in $\FPC$ as we saw before.
Note that in order to define these linear equation systems we use $j \in 
J$ as a parameter, so we really solve $|J|$-many linear equation systems 
in parallel.
Given the vectors $k_j\colon J \to \mbQ$, we can define the matrix 
$S\colon J \times J \to \mbQ$ as the matrix whose $j$-th column is the 
vector~$k_j$. Then $\im(S) = \ker(C) = \ker(M)$.
This completes our proof of Theorem~\ref{thm:fpc:definesolutions}.
\end{proof}

\begin{rem}
In fact, by going through our proof once again, one can show that 
Theorem~\ref{thm:fpc:definesolutions} can be strengthened to the 
extent that the $\FPC$-formulas only use fixed-point operators that 
converge after a polylogarithmic number of steps, cf.~\cite{GroheP17}.
\end{rem}

\section{Definability of Polynomial Calculus Refutations over Finite 
Fields}
\label{sec:PCfiniteFields}
In Section~\ref{sec:fpc} we proved that fixed-point logic with counting 
and the ($k$-dimensional) 
polynomial calculus over $\mbQ$ have the same expressive power if we restrict the coefficients that may occur in a refutation. In this section we study the polynomial calculus not over~$\mbQ$, but over 
finite fields. That means we do not need to worry about the representation of coefficients any more. Yet, we cannot hope to express the degree-$k$ PC over finite fields in fixed-point logic with counting in the general case: It is easy to show that the problem of 
solving linear equation systems 
over a field $\mbF$ can be reduced (in first-order logic) to finding 
proofs 
in the ($k$-dimensional) $\PC$ over~$\mbF$. However, $\FPC$ cannot 
define the 
solvability problem for linear equation systems over finite fields $\mbF$, 
see~\cite{AtseriasBulDaw09}.

Instead of giving up completely, we set out to explore  
certain (interesting) 
situations in which we can establish the same strong connections between 
$\FPC$ 
and the polynomial calculus that we discovered over $\mbQ$ also over 
finite 
fields. 
To identify these, we take a closer look at typical settings where 
the connection breaks down, that is where we encounter linear equation 
systems over finite fields that cannot be solved by $\FPC$.
To generate such hard linear equation systems, a common approach is to 
use the Cai-Fürer-Immerman (CFI) construction~\cite{CFI92}.
Specifically, the CFI-construction yields for every prime $p \in 
\Primes$ a family of structures $(\mfA^p_n)_{n \geq 1}$ 
of size $\mcO(n)$ such that the solvability problem for linear equation 
systems over finite fields $\mbF$ of characteristic $p$ that are 
defined in CFI-structures $\mfA^p_n$ (via $\FO$-interpretations)
cannot be expressed in $\FPC$.
Clearly, over these families of \CFI-structures $(\mfA^p_n)_{n \geq 1}$ 
there is no hope to 
express $\PCx k$-proofs in $\FPC$ over fields $\mbF$ of characteristic $p$.

However, what happens if we consider the following slightly more 
asymmetric situation. As before, we consider equation systems over a 
finite 
field $\mbF$ that are interpreted 
in CFI-structures $\mfA^p_n$. But, in contrast to the above, we make the 
additional assumption that the characteristic $q = \characteristic(\mbF)$ 
of the 
finite field $\mbF$ does not match the prime $p$ that was 
used for the \CFI-construction, that is we assume that $q \neq p$.
In this case, as we show in our main result of this section 
(Theorem~\ref{thm:main:PCfinitefields}), we can 
express $\PCx k$-refutations in fixed-point logic with counting.
Although this result only gives limited insight into the logical 
expressiveness of the polynomial calculus over finite fields, it turns out 
to be extremely useful to prove lower bounds for the polynomial calculus, 
as we demonstrate in Section~\ref{sec:lowerboundsPC}.

This section is structured as follows.
First of all, we recall (a generalised version of) the Cai-Fürer-Immerman 
construction in Subsection~\ref{subsec:CFI} and we analyse automorphism 
groups of \CFI-structures in Subsection~\ref{sec:symmetriesCFI}.
To unfold its full power, the \CFI-construction relies on an underlying 
family of highly connected graphs of bounded degree. To this end, we 
recall the notion of expander graphs in Subsection~\ref{subsec:expander}. 
We prove our first main technical result in 
Subsection~\ref{sec:homo} where we show that \CFI-structures over expander 
graphs are \emph{homogeneous wrt.\ $\FPC$}, which means that $\FPC$ can 
describe elements (and tuples of elements) in \CFI-structures up to 
automorphisms. An important consequence is that $\FPC$ can linearly order 
orbits of elements (and tuples of elements) with a bounded number of 
variables.
In Subsection~\ref{subsec:cyclic} we establish another important property 
of \CFI-structures which extends homogeneity: we show that \CFI-structures 
are \emph{cyclic (wrt.\ to $\FPC$)} which means that $\FPC$ can linearly 
order orbits of elements (and tuples of elements) by fixing a single 
parameter in this orbit. We also show that this property is closed under 
taking $\FPC$-interpretations and ordered pairs.
Building on this, we establish our key technical result in 
Subsection~\ref{subsec:cocylicLES}: we show that $\FPC$ can define 
solution spaces of linear equation systems over finite fields $\mbF$ that 
are interpreted in cyclic background structures with the additional 
assumption that the characteristic of the field $\mbF$ does not divide the 
size of the (Abelian) automorphism group of the cyclic structure (we say 
that such linear equation systems are \emph{cocyclic}).
Finally, in Subsection~\ref{subsec:cocycPC} we use this result in 
order to prove our main Theorem~\ref{thm:main:PCfinitefields}: $\FPC$ 
can 
express $k$-dimensional \PC-refutations for polynomial equation systems 
that are 
defined in cyclic structures over finite fields $\mbF$ if this same 
condition on the characteristic for $\mbF$ holds.

\subsection{Cai-Fürer-Immerman Construction}
\label{subsec:CFI}
For notational convenience, we introduce the \CFI-construction 
only for connected (undirected) graphs $G$ which are \emph{3-regular} and 
\emph{ordered}. The assumption that $G$ is ordered means that additional 
to the set of vertices $V = V(G)$ and the (symmetric) edge relation $E = 
E(G) \subseteq V(G) \times V(G)$ we assume that $G$ contains a 
linear order $\Ord \,\,\,=\, \Ord(G)$ on its set of vertices $V$. 
In fact, this is the original setting as it was introduced by Cai, Fürer, 
and Immerman in~\cite{CFI92}.

Let $p \in \Primes$ be a prime. For every vector $\lambda \in \mbF_p^V$ we 
construct the \emph{CFI-structure} $\CFIgraph{G}{p}{\lambda}$ over the 
(connected and ordered) graph $G$, the finite field $\field F_p$, and 
with \emph{load} $\lambda$ as the following relational structure with 
signature $\taucfi = \{ \preceq, R, C, I \}$ where $R$ is a ternary 
relation symbol and where $\preceq, I, C$ are binary relation symbols.
The universe $A$ of the CFI-structure $\mfA = \CFIgraph{G}{p}{\lambda}$ is 
$A = E(G) \times \mbF_p$.
The linear order $\Ord(G)$ on the vertex set $V(G)$ of $G$ 
extends to a linear order on the edge set $E(G)$.
We use this linear order on $E(G)$ to define the following 
\emph{total preorder} $\preceq$ on $A$: $(e,x) \preceq (f,y)$ if $e \leq 
f$. Note that $\preceq$ induces a linear order on the corresponding 
equivalence classes $e^p = e \times \mbF_p$. Clearly, each of these 
classes $e^p$ is of size $p$. 
Since $G$ is undirected every edge $e=(v,w) \in E$ comes with its
corresponding \emph{dual edge} $f=(w,v)\in E$. In what follows, we use
the notation $e^{-1} = f$ to denote the dual of the edge $e \in E$.
The relations $I$ and $C$ are defined follows.
\begin{itemize}
 \item The \emph{cycle relation} $C$ defines the cyclic structure 
 of the additive group of $\field  F_p$ on each of the equivalence classes 
$e^p$.
 More precisely, 
 \[ C = \bigcup_{e \in E} \{ ((e,x), (e, x+1 \text{ mod } p) )
 : x \in \field F_p\}.\]
 \item The \emph{inverse relation} $I$ relates additive inverses for 
dual edges. Formally,
 \[ I = \bigcup_{e \in E} \{ ( (e,x), (e^{-1}, -x) : x \in \field F_p 
\}.\]
\end{itemize}
Note that while the cycle relation $C$ defines a \emph{directed} cycle, the 
inverse relation $I$ is symmetric. Furthermore, observe that the 
relations $\preceq, C$ and $I$ are defined independently of the load 
vector $\lambda$ and so do only depend on the underlying graph $G$ and 
the prime field $\field F_p$.
In contrast, the \emph{CFI-relation} $R=R^\lambda$ is defined using the 
load vector $\lambda$ as follows.
For each $v \in V$, we let $vE \subseteq V$ denote the set of neighbours 
of $v$ in $G$, that is $E(v) = \{ v \} \times vE \subseteq E$ is the set 
of edges outgoing from $v$.
Since $G$ is $3$-regular we have that $|vE| = 3$ for each $v \in V$.
For $v \in V$ let $E(v) = \{ w_1, w_2, w_3 \}$ where $w_1 < w_2 < w_3$. The 
\emph{CFI-relation} $R^\lambda(v)$ at vertex $v$ is defined as follows:
\[ R^\lambda(v) = \{ ((w_1, x_1), (w_2, x_2), (w_3,x_3)) : x_1 + x_2 + x_3 
= \lambda(v)  \}. \]
The full CFI-relation $R^\lambda$ of the structure 
$\CFIgraph{G}{p}{\lambda}$ is given as $R^\lambda = \bigcup_{v \in 
V} R^\lambda(v)$.

\subsection{Symmetries of \CFI-Structures}
\label{sec:symmetriesCFI}
It turns out that the automorphism group $\Gamma$ of a CFI-structure
$\CFIgraph{G}{p}{\lambda}$ only depends on $G$ and $p$, but not on 
$\lambda$. To see this, first observe that every automorphism $\pi \in 
\Gamma$ has to maintain the linear preorder $\preceq$ which means that 
$\pi(e^p) = e^p$ for all $e \in E$. Moreover, $\pi$ has to maintain the 
cycle relation~$C$. This means that the action of $\pi$ on an 
edge class $e^p$ is a cyclic shift in~$\field F_p$. Let us write $\pi(e) 
\in \field F_p$ to denote this cyclic shift of $\pi$ on $e^p$ for $e \in 
E$.
Then, because of the inverse relation $I$, we have $\pi(e) + \pi(e^{-1}) = 
0$.
Altogether this shows that 
\[ \Gamma \leq \{ \pi \in \field F^E_p : \pi(e) + 
\pi(e^{-1}) = 0 \text{ for } e \in E \} \leq \field F^E_p.\] 
However, so far we have not taken the CFI-relation $R^\lambda$ into 
account.
Again, because of the linear preorder $\preceq$, for each $\pi \in \Gamma$ 
we have $\pi(R^\lambda(v)) = R^\lambda(v)$ for all $v \in V$.
Let  $v \in V$ and $vE = \{ w_1, w_2, w_3 \}$ and let 
$((w_1, x_1), (w_2, x_2) , (w_3, x_3)) \in R^\lambda(v)$, that is 
$x_1 + x_2 + x_3 = \lambda(v)$.
From our earlier observations we know that
\[\pi( (w_i, x_i)) =(w_i, 
x_i+\pi(v,w_i)).\]
Hence, the condition $\pi(R^\lambda(v))=R^\lambda(v)$ implies that
\[ x_1 + \pi(v,w_1) +  x_2 + \pi(v, w_2) + x_3 + \pi(v,w_3) = \lambda(v).\]
This implies that $\pi(v,w_1) + \pi(v,w_2) + \pi(v,w_3) = 
\sum_{e \in E(v)} \pi(e) = 0$. In fact, this condition is not only 
necessary but also sufficient for $\pi$ to preserve
$R^\lambda(v)$. Moreover, note that all of this holds independent of what
$\lambda$ is.
Altogether this gives us a characterisation of the automorphism group 
$\Gamma$ of $\CFIgraph{G}{p}{\lambda}$ as a subgroup of the vector space 
$\field F_p^E$ that is determined by the following set of linear equations 
in variables $\pi(e)$ for $e \in E$:
\begin{align*}
 \tag{\text{Inv}}\pi(e) + \pi(e^{-1}) &=0 && \text{ for } e \in 
E \label{aut-constr-inv}
\\
 \tag{\text{CFI}}\pi(v) := \sum_{e \in E(v)} \pi(e) 
&=0 && \text{ for 
} v \in V.\label{aut-constr-CFI}
\end{align*}

More generally, we can apply each vector $\pi \in \field F_p^E$, that
satisfies the constraints (Inv), to a CFI-structure 
$\CFIgraph{G}{p}{\lambda}$ and obtain a new CFI-structure over the same 
underlying graph~$G$.
As it turns out the resulting structure is $\CFIgraph{G}{p}{\lambda+\pi}$ 
where $(\lambda + \pi) (v) = \lambda(v) + \pi(v)$ for all $v \in V$.
Let us denote by $\Inv(\field F_p^E) \leq \field F_p^E$ the set of all 
vectors $\pi$ that satisfy the $(\Inv)$-constraints.

\begin{rem}
 For every $G=(V,E,\leq)$ (connected, ordered, $3$-regular) and each 
prime $p \in \Primes$, the group $\Delta = \Inv(\field F_p^{E}) \leq 
\field F_p^{E}$ acts on the set of CFI-structures over 
$G$ 
that is on $\CFIgraph{G}{p}{\star} = \{ \CFIgraph{G}{p}{\lambda} : \lambda 
 \in\field F_p^V\}$, and this action partitions the set into precisely $p$ 
orbits~(see below).
\end{rem}

Clearly, the set $\CFIgraph{G}{p}{\star}$ has size $p^n$ where $n = 
|V(G)|$. However, if we consider this set up to isomorphism, there are 
only $p$ different \CFI-structures over a fixed graph $G$: 
\begin{thmC}[\cite{CFI92, Holm10, Pakusa16}]
Two CFI-structures $\CFIgraph{G}{p}{\lambda}, \CFIgraph{G}{p}{\sigma} \in 
\CFIgraph{G}{p}{\star}$ are isomorphic if, and only if, 
\[\sum \lambda = \sum_{v \in V} \lambda (v) = \sum_{v \in V} \sigma(v) = 
\sum \sigma.\]
\end{thmC}

Let us remark that, for technical convenience, we have introduced 
CFI-structures as relational 
structures. However, it is easy to encode them as usual (unordered) 
graphs, and, in fact, this is the way in which they were originally 
defined in \cite{CFI92}. The main step is to introduce for each 
CFI-constraint $i = ((e_1, x_1), (e_2, x_2), (e_3,x_3)) \in 
R^{\lambda(v)}$, $e_i \in E(v)$, $x_i \in \field F_p$, a new node 
$i^{\lambda(v)}$ and to connect it to the edge nodes $(e_i,x_i) \in E(v) 
\times \field F_p$ 
accordingly
(these additional constraint nodes $i^{\lambda(v)}$ are called \emph{inner 
nodes} in 
the original 
construction of Cai, Fürer, and Immerman). 
Furthermore, we can replace the linear preorder by a path of the
appropriate length and connect vertices in the edge classes to positions 
on this path accordingly.
All of these simple transformation steps are clearly definable in \FPC.

\begin{lem}
\label{lem:encodeCFI:asgraph}
 There exist \FPC-interpretations 
 $\mcJ$ and $\mcJ^{-1}$ 
 such that $\mcJ$ maps CFI-structures $\mfA = \CFIgraph{G}{p}{\lambda} 
\in \CFIclass{\mcF}{p}$ to  graphs $\mcJ(\mfA)$ of degree $\mcO(p^2)$ and 
with $\mcO(p^2 \cdot n)$ many  vertices, where $n = |V(G)|$, and such that 
$\mcJ^{-1}$, which maps graphs to CFI-structures, is the inverse of $\mcJ$ 
in the sense that for all $\mfA \in \CFIgraph{G}{p}{\lambda}$ we have that
$\mcJ^{-1} ( \mcJ(\mfA) )$ is isomorphic to $\mfA$, that is $ \mcJ^{-1} 
( \mcJ(\mfA) ) \cong \mfA$.
\end{lem}

\subsection{Expander Graphs and CFI-Classes}
\label{subsec:expander}
Let us briefly recall the definition of \emph{expander graphs} based on 
the exposition in~\cite{HoLiWi06}.
In the following $G=(V,E)$ denotes an undirected $d$-regular 
graph (in this paper we only consider the case $d = 3$). 
For two subsets of vertices $S, T \subseteq V$ in $G$ we denote the set of 
(directed) edges from $S$ to $T$ by  $E[S;T] = E \cap (S \times T)$.
The \emph{edge boundary} of a set $S \subseteq V$ is $\partial S = E[S;V 
\setminus S]$ and the \emph{expansion ratio} $h(G)$ is defined as:
\[ h(G) = \min\limits_{\{ S : |S| \leq |V|/2 \}} \frac{|\partial S|}{|S|}. 
\]

\begin{defi}[Expander graphs]
A sequence $\mcF = \{ G_n : n \in \mathbb N\}$ of undirected $d$-regular 
graphs is called a \emph{family of $d$-regular expander graphs} if 
\begin{itemize}
 \item $\mcF$ is \emph{increasing}, that is $|V(G_n)|$ is monotone and 
unbounded, and
 \item $\mcF$ is \emph{expanding}, that is there exists $\varepsilon > 0$ 
such that $h(G_n) \geq \varepsilon$ for all $n \in \mathbb N$.
\end{itemize}
\end{defi}

For our applications in this paper we make use of the existence of a 
family of 
$3$-regular connected expander graphs.

\begin{thm}[see e.g.\ Example~2.2 in~\cite{HoLiWi06}]
There exists a family of 3-regular expander graphs $\mcF = \{ G_n : n \in 
\mathbb N\}$ such that each graph $G_n$, $n \in \mathbb N$, is connected 
and has $\mcO(n)$ vertices. 
\end{thm}

For the rest of this paper let us fix a family $\mcF = \{ G_n : n \in 
\mathbb N\}$ of expander graphs as in the previous theorem.
Of course, we can also assume that the graphs in $\mcF$ are
\emph{ordered} just by adding to each graph $G_n \in \mcF$ an arbitrary 
linear order on $V(G_n)$. 
From this family $\mcF$ of $3$-regular, connected, ordered expander graphs 
$G_n$ with $\mcO(n)$ many vertices we construct, for every $p \in 
\Primes$, the CFI-class $\CFIclass{\mcF}{p}$ consisting of 
all CFI-structures over graphs from $\mcF$ that is
\[ \CFIclass{\mcF}{p} = \bigcup_{n \in \mathbb N} 
\CFIgraph{G_n}{p}{\star}.\]
The \emph{CFI-problem} (over $\mcF$ and $p \in \Primes$) is to 
decide, given a structure 
$\CFIgraph{G}{p}{\lambda} \in \CFIclass{\mcF}{p}$ whether $\sum\lambda = 
0$. It was shown by Cai, Fürer, and Immerman that this problem is 
undefinable in counting logic with sublinearly many variables.
\begin{thmC}[\cite{CFI92}]
 For $\CFIgraph{G_n}{p}{\lambda}, \CFIgraph{G_n}{p}{\sigma}  \in 
\CFIclass{\mcF}{p}$ we have
\[ \CFIgraph{G_n}{p}{\lambda} \equiv^{\Omega(n)} 
\CFIgraph{G_n}{p}{\sigma}.\]
\label{fpc-cfi}
\end{thmC}

\subsection{Homogeneity of \CFI-Structures}
\label{sec:homo}
In this section we establish a key technical result:
We show that 
CFI-structures (over ordered expander graphs) are 
($\FPC$-)\emph{homogeneous}. This means that the orbits of $k$-tuples in 
CFI-structures can be uniquely described in $\FPC$ by using only $\mcO(k)$ 
many variables.
This is extremely useful as it implies that we can actually order the set 
of orbits on $k$-tuples in $\FPC$ using only $\mcO(k)$ many variables. 
To put it differently, we show that in CFI-structures over 
expander graphs, $\FPC$ 
can describe tuples up to their automorphism type without using too many 
resources (variables). 
We believe that this result is of independent interest and should prove 
useful in
other applications (one example is discussed in 
Section~\ref{sec:discussion}).

\begin{thm}(Homogeneity)\label{thm:homogeneity}
There is a constant $\ell \geq 1$ such that for every $k \geq 1$, $p \in 
\Primes$, and every $\mfA = \CFIgraph{G}{p}{\lambda} \in 
\CFIclass{\mcF}{p}$ with 
automorphism group $\Gamma \leq \field F_p^{E(G)}$ and every $\tup a, \tup 
b \in A^k$ we have that 
$(\mfA, \tup a) \equiv^{\ell \cdot k} (\mfA, \tup b)$ if, and only if, 
$\Gamma(\tup a) = \Gamma(\tup b)$.
\end{thm}

Let $1 > \varepsilon > 0$ be the expander constant corresponding to the 
class $\mcF$, that is for all $G \in \mcF$ we have that $h(G) \geq 
\epsilon$.
We will prove Theorem~\ref{thm:homogeneity} for 
\[ \ell \geq \frac{12}{\epsilon}. \]
To this end, we inductively show the following for all $k \geq 0$: 
For  $\mfA = \CFIgraph{G}{p}{\lambda} \in \CFIclass{\mcF}{p}$ with 
automorphism group $\Gamma \leq \field F_p^{E(G)}$, every $\tup a \in 
A^k$, 
and every $a, b \in A$ such that $(\mfA, \tup a, a) \equiv^{\ell \cdot 
(k+1)} (\mfA, \tup a, b)$ we can find an automorphism $\pi \in \Gamma$ 
such that $\pi(\tup a, a) = (\tup a, b)$.

If we have shown this, then the above theorem easily follows. 
Indeed, let $\tup a, \tup b \in A^k$ such that $(\mfA, \tup a) 
\equiv^{\ell \cdot k} (\mfA, \tup b)$. 
We have to show that in this case $\tup a$ and $\tup b$ are in the same 
orbit, i.e.\ $\Gamma(\tup a) = \Gamma(\tup b)$. In other words, we have to 
show that every $\equiv^{\ell \cdot k}$-class on $A^k$ is a single 
$\Gamma$-orbit (clearly, each such class is a union of orbits).
For the sake of contradiction, assume that $\Gamma(\tup a) \neq 
\Gamma(\tup b)$.
Let $r \geq 0$, $r < k$ be maximal with respect to the following property:
there exists $\tup c \in \Gamma(\tup a)$ such that $\tup b$ and 
$\tup c$ share a prefix of length $r$, that is $\tup b =(b_1, \dots, 
b_r, b_{r+1}, \cdots b_k)$, and $\tup c = (b_1, \dots, b_r, c_{r+1}, \dots 
c_k)$, and $b_{r+1} \neq c_{r+1}$.
But then $(\mfA, a_1, \dots, a_{r}, a_{r+1}) \equiv^{\ell \cdot k} (\mfA, 
b_1, \dots, b_r, c_{r+1})$ (since $\tup a$ and $\tup c$ are in the same 
orbit) and $(\mfA, a_1, \dots, a_{r}, a_{r+1}) \equiv^{\ell \cdot k} 
(\mfA, b_1, \dots, b_{r}, b_{r+1})$ (by the assumption).
Hence \[ (\mfA, b_1, \dots, b_{r}, b_{r+1}) \equiv^{\ell \cdot k} 
(\mfA, b_1, \dots, b_{r}, c_{r+1}).\]
By the above proposition, we can find $\pi \in 
\Gamma$ such that $\pi(b_1, \dots, b_r, c_{r+1}) = (b_1, \dots, b_r, 
b_{r+1})$. This means that $\pi(\tup c) \in \Gamma(\tup a)$ is a tuple in 
$\Gamma(\tup a)$ sharing a longer prefix with $\tup b$ than $\tup c$ which 
leads to the desired contradiction.

Hence, let us now focus on proving the above proposition. 
For this let $k \geq 0$, $\tup a \in A^k$, and $a, b \in A$ such 
that $(\mfA, \tup a, a) \equiv^{\ell \cdot 
(k+1)} (\mfA, \tup a, b)$ where $\mfA = \CFIgraph{G}{p}{\lambda} \in 
\CFIclass{\mcF}{p}$. Recall that $\Gamma \leq \Inv(\field F_p^{E(G)}) 
\leq \field F_p^{E(G)}$ denotes the automorphism group of $\mfA$.
We have to show the existence of some $\pi \in \Gamma$ such that 
$\pi(\tup a, a) = (\tup a, b)$.
To do this, we establish two key properties of the elements $a, b \in A$ 
using the fact that $(\mfA, \tup a, a) \equiv^{\ell \cdot 
(k+1)} (\mfA, \tup a, b)$. Clearly, we can assume that $a \neq b$, because 
the claim is trivial otherwise.

\begin{enumerate}[label=(P1)]
 \item The elements $a, b$ are in the same edge class in $\mfA$, i.e.\ 
there exists 
$e \in E(G)$ such that $a, b \in e^p = e \times \field F_p \subseteq A$.
\label{itm-hom-p1}
\end{enumerate}

This easily follows from the fact that each edge class $e^p$ can 
be identified in counting logic by using the preorder $\preceq$ and at most
three variables (which we use to enumerate the edge classes 
starting from the minimal one). Moreover, this formula does not require 
access to any of the parameters from $\tup a$. Hence, a counting 
logic formula with three variables could distinguish between $a$ and $b$ 
in $\mfA$ if they were in different edge classes. Note that $\ell \geq 12$, 
so we clearly have enough variables available.

The next simple observation is that this edge class $e^p$ is 
\emph{free}, a property that we are going to define now.
Let $\Blocked \subseteq E(G)$ be the smallest set such that
\begin{enumerate}[label=(\roman*)]
  \item if $a_i \in f^p$, for some $1 \leq i \leq k$ and for $f \in E(G)$, then 
$f \in \Blocked$, and \label{itm-blocked-1}
 \item if $f \in \Blocked$, then $f^{-1} \in \Blocked$, 
\label{itm-blocked-2}
\end{enumerate}
We say that the edges in $\Blocked$ are \emph{blocked} and the edges $E(G) 
\setminus \Blocked$ are \emph{free}.
For blocked edges it is straightforward to define the individual elements in 
the corresponding blocked edge classes.
\begin{lem}
 For every blocked edge $f \in \Blocked$ and every $c \in f^p = f \times 
\field 
F_p \subseteq A$, there exists a formula of counting logic $\phi(x_1, 
\dots, x_k, y)$ with at most $k+2$  many 
variables which defines 
$c$ in $(\mfA, a_1, \dots, a_k)$, i.e.\ such that for every $d \in A$ we 
have that $\mfA \models \phi(a_1, \dots, a_k, d)$ if, and only if, $c = d$.
\label{lem-blocked-definable}
\end{lem}
 
\begin{proof} First of all, assume that an edge $f \in \Blocked$ is marked 
as 
blocked because for some $1 \leq i \leq k$ we 
have that $a_i \in f^p = f \times \field F_p$, i.e.\ we are in case 
\ref{itm-blocked-1}. Recall that the cycle relation $C$ defines a 
directed cycle on $f^p$. Hence, using $C$ and $a_i$ as a parameter we 
can define every other element $c \in f^p$ in counting logic using 
the parameter $a_i$ and one additional auxiliary variable.

Secondly, assume that $f \in \Blocked$ is blocked, because $g = f^{-1} \in 
\Blocked$, i.e.\ we are in case~\ref{itm-blocked-2}.
By the induction hypothesis we know that we can define in counting logic 
each element in $g^p$ using $k+2$ many variables (and the
parameters in $\tup a$). Let $\phi(\tup x, y)$ be a formula defining such an 
element $c \in g^p$. 
Then $\psi(\tup x, y) = \exists z ( I(z,y) \wedge \phi(\tup x, z))$ 
is a formula of counting logic with at most $k+2$ many variables which 
defines an element in $f^p$. Since $I$ is a bijection between $f^p$ and 
 $g^p$ we can define each element in $f^p$ in this way. 
\end{proof}

For obvious reasons, the number of blocked edges is linearly bounded 
in $k$:
\begin{lem}
 The number of blocked edges is linear in $k$: we have $|\Blocked| 
\leq 2 k$ (or $|\Blocked| \leq k$ if we count edges as undirected).
\label{lem:numberblockededges}
\end{lem}

Let us come back to our original goal. 
Recall that we have to show the existence of some $\pi \in \Gamma$ such 
that 
$\pi(\tup a, a) = (\tup a, b)$ where $(\mfA, \tup a, a) \equiv^{\ell \cdot 
(k+1)} (\mfA, \tup a, b)$, $a \neq b$.
With the above preparation it is now easy to see that $a$ and $b$ satisfy 
the following property.

\begin{enumerate}[label=(P2)]
 \item The edge class $e^p$, $e \in E(G)$, that contains the elements 
$a, b$ (see~\ref{itm-hom-p1}), is free, that is $e \in E(G)\setminus 
\Blocked$.
\end{enumerate}
This immediately follows from Lemma~\ref{lem-blocked-definable} (note that 
$k+2 \leq \ell \cdot (k+1)$).

Let us fix a free edge class $e^p$, $e \in E(G)\setminus \Blocked$. 
We are going to construct an automorphism $\pi \in \Gamma$ such that 
$\pi(\tup a) = \tup a$ and such that $\pi(e) = 1$, that is $\pi$ acts as 
a cyclic shift by one on the edge class $e^p$. 
If we can show this, then our original claim follows.
To this end, we distinguish between the following two cases.
We say that the edge $e$ lies on a \emph{free cycle}, if there exist 
edges $e_0=(v_0,v_1), e_1=(v_1,v_2), \dots, 
e_r=(v_r,v_0)$, $r \geq 2$, such that all $v_i$, $0 \leq i \leq r$, are 
distinct and such that $e = e_0$ and $e_i \in E(G) \setminus \Blocked$, $0 
\leq i 
\leq r$.
Indeed, if $e$ lies on such a free cycle, then we can construct an 
automorphism $\pi \in \Gamma$ with the desired properties as follows:
we simply set $\pi(e_i) = 1$ and $\pi(e_i^{-1}) = -1$ for all 
$0 \leq i \leq r$. Note that each of the moved 
edge classes $e_i^p$ is free, so none of the elements in the tuple $\vec 
a$ will occur 
in any of the edge classes moved by $\pi$, that is $\pi(\tup a) = \tup a$.

Hence, the interesting case is that $e$ does not lie on a free cycle.
We show that in this case each element in the edge class 
$e^p$ can be defined in counting logic by fixing elements in a 
bounded number of additional edge classes (more precisely, by using at 
most $\ell \cdot (k+1)$ many variables). Hence, for our original setting 
this would mean that the assumption $(\mfA, \tup a, a) \equiv^{\ell \cdot 
(k+1)} (\mfA, \tup a, b)$ would already imply that $a = b$.
To prove this, we strongly make use of the fact that the family $\mcF$ of 
graphs from which we constructed the CFI-class $\CFIclass{\mcF}{p}$ is 
an expander family. Let us formulate our claim precisely.

\begin{lem} 
As above, let $\mfA = \CFIgraph{G}{p}{\lambda} \in \CFIclass{\mcF}{p}$, 
$\tup a \in A^k$, $\ell \geq 12 / \varepsilon$, and let $e \in 
E(G)\setminus B$ be a free edge which does not lie on a free cycle.
Then for every $b \in e^p$ there exists a formula $\phi(\tup x, y)$ of 
counting logic with at most $\ell \cdot (k+1)$ many variables that 
defines $b$ in $(\mfA, \tup a)$, that is for every $c \in A$ we have that 
$\mfA \models \phi(\tup a, c)$ if, and only if, $c = b$.
\end{lem}

\begin{proof}
Let us consider the subgraph $H=(W,F)$ of $G$ that is induced by the free 
edges $E(G) \setminus \Blocked$. Note that since $\Blocked$ is symmetric, 
$H$ is an 
undirected graph.
Let $e = (v,w) \in E(G)\setminus \Blocked$ be the free edge that we 
consider and let $X \subseteq 
V(H)$ and $Y \subseteq V(H)$ denote 
the connected components of $v$ and $w$ in the graph 
$H 
\setminus e$, that is in the graph that results from $H$ by removing the 
(undirected) edge $e$.
Since $e$ does not lie on a free cycle we know that $X$ and $Y$ are 
disjoint. 
Using the expander property of $G$, we now aim to bound the size of $X$ 
or $Y$. 
Clearly, at least one of the two sets contains at most $|V(G)| / 2$ many 
vertices. Without loss of generality, let us assume that $|X| \leq |V(G)| 
/ 2$. 
Then $|\partial X| \geq |X| \cdot \varepsilon$ since $h(G) \geq 
\varepsilon$. 
The important observation is that we can bound $|\partial X|$ in terms 
of $k$.
Indeed, in $G$ every edge leaving the set $X$ (different from $e$) has to 
be blocked, since $X$ is a connected component in $H$. Hence $|\partial X| 
\leq k$.  
This yields the bound of $|X| \leq k/\varepsilon$ on the size of $X$.

In conclusion, the set $X \subseteq V(G)$ in $G$ is a set of vertices of 
size at most $k/\varepsilon$ such that each edge leaving $X$ is blocked 
except for the single free edge $e$. 
We now consider the CFI-substructure $\mfB$ of the input CFI-structure 
$\mfA$ induced on the edge classes incident with vertices in $X$ where in 
every blocked edge class $f^p$ we arbitrarily mark an element $c \in 
f^p$ to be $c = (f,0) \in f^p$ (this choice depends on the 
parameters $\tup a$).
More precisely, let $E_X = \{ e \in E(v): v \in X, e \not\in \Blocked \}$, 
then the universe $B$ of $\mfB$ is the set $B= \bigcup_{e \in E_X} e^p$, 
and the linear preorder $\preceq$, the cycle 
relation~$C$, and the inverse relation $I$ in $\mfB$ are just the 
restrictions of the corresponding relations in $\mfA$ to the 
subuniverse~$B$.
To define the CFI-relation $R_\mfB^\lambda$ on $\mfB$ we distinguish 
between the following cases. Recall that $R^\lambda = \bigcup_{v \in 
V(G)} R^\lambda(v)$. 
First, let us consider vertices $v \in X$ whose neighbours are all 
contained in $X$. In this case we simply set $R_\mfB^\lambda(v) = 
R^\lambda(v)$.
For vertices $v \in X$ for which some incident edge (classes) are blocked we 
define $R_\mfB^\lambda(v)$ as follows.
Let $F\subseteq E(v)$ denote the set of blocked edges incident with $v$. 
Note that $|F| \leq 2$. We fix $x^f \in f^p$ for every $f \in F$ (we 
can make this choice using the parameters $\tup a$, 
see Lemma~\ref{lem-blocked-definable}).
Then we define $R_\mfB^\lambda(v)$ to be the restriction of $R^\lambda(v)$ 
to those tuples that contain $x^f$ for every $f \in F$. 
If we recall the definition of $R^\lambda(v)$, then this intuitively 
corresponds to declaring $x^f = (f,0)$. In particular, note 
that the arity of $R_\mfB^\lambda(v)$ is $3-|F|\geq 1$.
Finally, the CFI-relation $R_\mfB^\lambda$ in $\mfB$ is defined as 
$R_\mfB^\lambda = \bigcup_{v \in X}{R^\lambda(v)}$.

If follows from our preparations that $\mfB$ can be defined in 
$\mfA$ by using a parametrised, one-dimensional interpretation $\mcI(\bar 
x)$ in counting logic, i.e.\ $\mcI(\mfA, \tup a) = \mfB$.
Moreover, $\mcI$ can be constructed by using, as a rough estimate, at most 
$k + 6$ many variables. The most important thing to observe is that we 
can use Lemma~\ref{lem-blocked-definable} in order to fix elements in all 
blocked edge classes as required.

We now want to argue that every possible automorphism $\pi \in \field 
F_p^{E_X}$ of $\mfB$ will fix the edge class $e^p$ (recall that $e$ 
denotes 
the single free edge $e$ that leaves the set $X$).
Recall that the inverse constraints~(\ref{aut-constr-inv}) enforce that 
for each pair of dual edges $f, g \in E_X$ we have $\pi(f) + \pi(g) = 0$.
Note that for each edge $f \in E_X$ we have $f^{-1} \in E_X$ except for 
the single 
edge $e$ for which $e^{-1} \not\in E_X$. Hence $\sum_{f \in E_X} \pi(f) = 
\pi(e)$.
Moreover, recall that the CFI-constraints~(\ref{aut-constr-CFI}) enforce 
that for each $v \in X$ we have $\sum_{f \in E_X(v)} \pi(f) = 0$.
Hence, $\sum_{v \in X} \sum_{f \in E_X(v)} \pi(f) = 0$.
Since $\sum_{f \in E_X} \pi(f) = \sum_{v \in X} \sum_{f \in E_X(v)} 
\pi(f)$ we 
conclude that $\pi(e) = 0$, that is $\pi$ fixes the edge class $e^p$, as 
claimed.

It follows that in $\mfB$ every element $x \in e^p$ can be defined 
in counting logic by using roughly $|E_X| + 6$ many variables. Indeed, for 
some $c \in e^p$, we can use $|E_X|$ many variables to fix elements in all 
other edge class and then describe the isomorphism type of the structure 
$(\mfB,c)$. Since each $c \in e^p$ is in a singleton orbit, these 
isomorphism types will be different for all elements $c \in e^p$.
Note that $|E_X| + 6 \leq 3k / \varepsilon + 6$.
If we translate the resulting formulas back to $\mfA$ via $\mcI(\tup x)$, 
then we obtain a formula in counting logic that defines $c \in e^p$ in 
$\mfA$ and which uses at most $3k / \varepsilon + 6 + k + 6 \leq 4k / 
\varepsilon + 12 \leq (12k + 12) / \varepsilon \leq \ell \cdot (k+1)$ many 
variables.
\end{proof}

This completes the proof of Theorem~\ref{thm:homogeneity}.

\begin{defi}
 For $\ell \geq 1$, we say that a structure $\mfA$ with automorphism group 
$\Gamma$ is \emph{$\ell$-homogeneous} if for all $k \geq 1$ and all 
$k$-tuples $\tup a, \tup b \in A^k$ we have that
\[ (\mfA, \tup a) \equiv^{\ell \cdot k} (\mfA, \tup b) \text{ if, and only 
if, } \Gamma(\tup a) = \Gamma(\tup b).\]
Moreover, we say that a class $\mcK$ of structures is \emph{homogeneous} 
if 
there is an $\ell \geq 1$ such that each structure $\mfA \in \mcK$ is 
$\ell$-homogeneous.
\end{defi}

\begin{cor}
 The class of CFI-structures $\CFIclass{\mcF}{p}$ is homogeneous.
\label{cor:cfihom}
\end{cor}

As mentioned before, an important consequence of homogeneity is that 
$\FPC$ can (uniformly) 
define a total preorder on the set $A^k$, for each $k \geq 1$, which 
orders $k$-tuples up to orbits. Moreover, the number of variables 
required by such an $\FPC$-formula is linear in~$k$.
To see this, we make use of the well-known fact that for every $\ell \geq 
1$ 
there exists an $\FPC$-formula $\Typ{\ell}(\tup x, \tup y)$ with 
$\mcO(\ell)$ many variables which defines on each input structure $\mfA$ a 
linear preorder on $A^\ell$ which distinguishes between all pairs of tuples 
$\tup a, \tup b \in A^\ell$ for which $(\mfA, \tup a) \not\equiv^\ell (\mfA, 
\tup b)$ holds, see e.g.~\cite{Otto97}. That is $\Typ \ell(\tup x, \tup y)$ 
defines in each input structure 
$\mfA$ a linear order on the set $\{ [\tup a]_{\equiv^\ell} : \tup a 
\in A^\ell\}$ consisting of $\equiv^\ell$-equivalence classes   
$[\tup a]_{\equiv^\ell} = \{ \tup b \in A^\ell: (\mfA, \tup a) 
\equiv^\ell (\mfA,\tup b) \}$ for $\tup a \in A^\ell$.
Of course, we can also use the formula $\Typ {\ell}(\tup x, \tup y)$  
to define the corresponding preorder on $k$-tuples for lengths $1 
\leq k < \ell$ (a common approach is to extend $k$-tuples to $\ell$-tuples 
by repeating the last component).
We denote the corresponding \FPC-formula by $\Typk{\ell}{k}(\tup x, \tup 
y) = \Typk{\ell}{k}(x_1, \dots, x_k, y_1, \dots, y_k)$.

\begin{thm}\label{thm:preorder}
 Let $\mfA$ be $\ell$-homogeneous with automorphism group $\Gamma$. Then 
the \FPC-formula $\Typk{\ell\cdot 
k}{k}(\tup x, \tup y)$ defines a total preorder $\preceq$ on $A^k$ 
that identifies $k$-tuples which are in the same orbit. In particular, 
$\Typk{\ell\cdot k}{k}(\tup x, \tup y)$ induces a linear order on the set 
of orbits of $k$-tuples $\{ \Gamma(\tup a) : \tup a \in A^k\}$.
\label{thm:order:ellhom}
\end{thm}

\subsection{CFI-structures are Cyclic}
\label{subsec:cyclic}
In Section~\ref{sec:homo} we proved that the CFI-classes 
$\CFIclass{\mcF}{p}$ are homogeneous, which by Theorem~\ref{thm:preorder} 
implies that $\FPC$ can order 
$k$-tuples in structures $\mfA \in \CFIclass{\mcF}{p}$ up to orbits using 
only $\mcO(k)$ many variables.
In this subsection we go one step further and show that, as a result of 
the algebraic properties of the automorphism groups of CFI-structures, 
each individual orbit of $k$-tuples can be linearly ordered in fixed-point 
logic with counting 
by fixing a single $k$-tuple from this orbit as a parameter (and, again, by
using $\mcO(k)$ many variables only). 
Furthermore, we are going to show that this key property of CFI-structures 
remains intact if we apply logical transformations.
Intuitively, our results show that CFI-structures come quite close to 
ordered 
structures: in $\FPC$, one can preorder the elements of 
CFI-structures up to orbits,\emph{and}, secondly, 
each individual orbit can be totally ordered by fixing a 
\emph{single} element as a parameter. Note, however, that this does not 
mean that we can order the full CFI-structure, 
since this would require to 
fix a parameter in \emph{each} of the orbits at the same time. Indeed,
Theorem~\ref{fpc-cfi} implies that CFI-structures can \emph{not} be 
totally 
ordered in $\FPC$ if we restrict ourselves to formulas with a sublinear 
number of variables.

Recall that, for $1 \leq k \leq \ell$, the formulas 
$\Typk{\ell}{k} = \Typk{\ell}{k}(\tup x, 
\tup y)$ define a total preorder that distinguishes $k$-tuples up to 
$\equiv^\ell$-equivalence.
In what follows we make use of parametrised versions of these 
formulas.
More precisely, for a parameter tuple $\tup z$ of length $r \geq 0$ we 
write 
$\Typk{\ell}{k}[\tup z](\tup x, \tup y)$ to denote the formula
$\Typk{\ell}{r + k}(\tup z \tup x, \tup z \tup y)$ (of course this only 
makes sense if $r+k \leq \ell$).
Note that, again, this formula orders $k$-tuples up to 
$\equiv^\ell$-equivalence, but now we consider $\equiv^\ell$-equivalence 
with respect to the additional parameter tuple $\tup z$.
Hence for every structure $\mfA$ and every $\tup c \in A^r$ we have that 
the linear preorder defined by $\Typk{\ell}{k}[\tup c]$ in $\mfA$ refines 
the linear preorder defined by $\Typk{\ell}{k}$ in $\mfA$. 
Note that, in particular, the tuple $\tup c$ will always be in a singleton 
class according to the preorder $\Typk{\ell}{k}[\tup c]$.
Note further that for the special case $r = 0$ we just obtain the formula 
$\Typk{\ell}{k}$.

Given a structure $\mfA$ with automorphism $\Gamma$, we denote for 
a parameter $\tup c \in A^r$ by $\Gamma_{\tup c} \leq \Gamma$ the 
stabiliser subgroup of the tuple $\tup c$, i.e.\ the 
group of all $\pi \in \Gamma$ such that $\pi(\tup c) = \tup c$.
\begin{defi}
\label{def:lpcyclic}
A structure $\mfA$ with automorphism group $\Gamma$ is called 
$\cyc{\ell}{p}$, for $\ell \geq 1$ and $p \in \Primes$, if the 
following holds for every $k \geq 1$:
\begin{enumerate}[leftmargin=*,label=(C-\Roman*)]
 \item $\Gamma$ is an Abelian $p$-group. In particular, for every 
$k$-tuple $\tup 
a \in A^k$, the size of the 
orbit $\Gamma(\tup a)$ of $\tup a$ is a $p$-power, that is $|\Gamma(\tup 
a)| = p^n$ 
for some $n \geq 0$.\label{cyclp-sizeorb}
 \item For every $\tup c \in A^r$, $r \geq 0$, the 
$\FPC$-formula 
 $\Typk{\ell \cdot (k+r)}{k}[\tup c]$ defines a total preorder on 
$A^k$ such that two tuples $\tup a, \tup b \in A^k$ are incomparable if, 
and only if, $\Gamma_{\tup c}(\tup a) = \Gamma_{\tup c}(\tup b)$.
Note that for  $r = 0$ we obtain $\ell$-homogeneity as a special case.
\label{cyclp-ordsingorb}
\end{enumerate}
We say that a class $\mcK$ of structures is $\cyclp$ if every 
structure $\mfA \in \mcK$ is $\cyclp$.
\end{defi}

Actually, if in the above definition, we would only include
item~\ref{cyclp-ordsingorb}, then the resulting notion of 
$\cyclp$ structures would not be very interesting: it would collapse to 
the 
notion of $\ell$-homogeneity, see Theorem~\ref{thm:homab:lpcyc} below.
However, in combination with condition~\ref{cyclp-sizeorb}, we get a 
remarkable effect:
\begin{lem}
In Definition~\ref{def:lpcyclic}, we can add the following to 
item~\ref{cyclp-ordsingorb} without changing the resulting notion:
Assume that $r \leq k$. Then the preorder defined by $\Typk{\ell \cdot 
(k+r)}{k}[\tup c]$ induces a 
linear order on the orbit $\Gamma(\tup c)$ of the parameter $\tup c$.
\end{lem}
\begin{proof} This follows from the fact that, 
by~\ref{cyclp-sizeorb},  $\Gamma$ is an Abelian group (and so the 
induced group action on the orbit is regular).
More explicitly, assume that for some $\pi \in \Gamma$ and
$(\tup c, \tup a) \in \{ \tup c \} \times \Gamma(\tup c)$ it holds that
$\pi(\tup c, \tup a) = (\tup c, \tup b)$.
Choose $\sigma \in \Gamma$ such that $\sigma(\tup a) = \tup c$.
Then $\pi(\sigma(\tup a)) = \tup c$. Since $\Gamma$ is Abelian, it follows 
that $\sigma(\pi(\tup a)) = \tup c$. Hence $\pi(\tup a) = \tup a$, which 
yields $\tup a = \tup b$.
Hence, it follows that for every $\tup a \in \Gamma(\tup c)$ we have 
$|\Gamma_{\tup c}(\tup a)| = 1$.
Having this, item~\ref{cyclp-ordsingorb} implies that $\Typk{\ell 
\cdot (k+r)}{k}[\tup c]$ defines a 
linear order on $\{ \tup c \} \times \Gamma(\tup c)$, as claimed.
\end{proof}

\begin{thm}
\label{thm:homab:lpcyc}
Let $\mfA$ be $\ell$-homogeneous and assume that the automorphism group 
$\Gamma$ of $\mfA$ is an Abelian $p$-group. Then $\mfA$ is $\cyclp$.
In particular, there is $\ell \geq 1$ such that the classes
$\CFIclass{\mcF}{p}$ are
$\cyclp$ for all $p \in \Primes$.
\label{thm:cfi:cyclp}
\end{thm}
\begin{proof}
 We already analysed the automorphism groups of CFI-structures 
$\mfA \in \CFIclass{\mcF}{p}$ in Section~\ref{sec:symmetriesCFI}. In 
particular, we saw that these groups are elementary Abelian $p$-groups, so 
property~\ref{cyclp-sizeorb} holds for CFI-structures in 
$\CFIclass{\mcF}{p}$.
Moreover, Corollary~\ref{cor:cfihom} tells us that classes of 
CFI-structures are homogeneous.

Now, let $\ell \geq 1$ and let $\mfA$ be $\ell$-homogeneous with 
automorphism group~$\Gamma$.
Then, by Theorem~\ref{thm:order:ellhom}, we know that for every $k \geq 
1, r \geq 0$ the formula $\Typk{\ell \cdot (r+k)}{(r+k)}(\tup x_1 \tup 
x_2, 
\tup y_1 \tup y_2)$ defines in $\mfA$ a total preorder on $A^{r+k}$ which 
order $(r +k)$-tuples up to $\Gamma$-orbits.
Since 
\[\Typk{\ell \cdot (r+k)}{k}[\tup z](\tup x, \tup y) = \Typk{\ell \cdot 
(r+k)}{(r+k)}(\tup z \tup x, \tup z \tup y),\]
we know that for every $\tup c \in A^r$ it holds that the total preorder
$\Typk{\ell \cdot (r+k)}{k}[\tup c]$ distinguishes $k$-tuples $\tup a, 
\tup b \in A^k$ if, and only if, $\Gamma(\tup c \tup a) \neq \Gamma(\tup c 
\tup b)$. But this last condition is indeed equivalent to
$\Gamma_{\tup c}(\tup a) \neq \Gamma_{\tup c}(\tup b)$, which completes 
the proof.
\end{proof}

Our next aim is to show that the class of $\cyclp$ structures is 
closed under $\FPC$-transformations.
Unfortunately, stated in this very general form, this claim is clearly wrong. 
For 
example, $\FPC$-transformations can easily generate each fixed finite  
structure (starting from any structure), and so the resulting structures 
will not have Abelian automorphism groups for instance (which is one of 
the requirements for being \cyclp). However, as we 
will show next, one can extend each \FPC-interpretation $\mcI$ to 
an $\FPC$-interpretation $\Norm(\mcI)$ in such a way that the original 
input structure is preserved as a 
substructure. This will enable us to maintain the property of being \cyclp.

Let us be a bit more precise. As said, instead of only interpreting 
$\mcI(\mfA)$ in $\mfA$ we want to
interpret the structure $\Norm(\mcI)(\mfA) = \mcI(\mfA) \uplus \mfA$ in 
$\mfA$, that is the disjoint union of the original structure $\mfA$ and the 
interpreted structure $\mcI(\mfA)$. However, as such, this is not sufficient 
since we can 
still get new automorphisms due to the new substructure $\mcI(\mfA)$.
To overcome this problem, we create additional relations that
indicate from which elements in 
$\mfA$ the newly created elements in $\mcI(\mfA)$ originate. Note that the 
elements in $\mcI(\mfA)$  are equivalence classes of tuples of elements 
from $\mfA$ and we will encode this information in $\Norm(\mcI)(\mfA)$.
Formally, our result is as follows.

\begin{thm}
\label{thm:nfinter}
 Let $\mcI(\tup z) \in \FPC[\sigma \to \tau, \tup z]$ be an 
$\FPC$-interpretation of dimension~$d$ and with $r \geq 0$ parameters 
$\tup z$, $| \tup z | = r$, that 
maps $\sigma$-structures to $\tau$-structures.
Let $\hat \tau = \sigma \uplus \tau \uplus \{ \in \}$ (where $\in$ is a 
fresh binary relation symbol).
Then there exists an $\FPC$-interpretation $\mcJ(\tup z)  \in 
\FPC[\sigma \to \hat \tau, \tup z]$ such that for every structure $\mfA$ 
and every $\tup a \in \dom(\mfA, \tup z)$ the following holds 
for $\mfB = \mcI(\mfA, \tup a)$ and $\mfC = \mcJ(\mfA, \tup a)$:
\begin{enumerate}[label=(\roman*)]
 \item the dimension of $\mcJ(\tup z)$ is at most $d+3$, and
 \item $\mfB \subseteq_{\FO} \mfC|_{\tau}$, that is $\mfB$ is an 
$\FO$-definable substructure of the reduct of $\mfC$ 
to $\tau$, and
\label{itm:modint:substr}
 \item if $\Gamma = \Aut(\mfA, \tup a)$ and $\Delta = \Aut(\mfC)$, then 
$\Gamma \cong \Delta$ (that is, up to isomorphism, the automorphism group 
of the input structure $(\mfA, \tup a)$ is preserved), and
\label{itm:modint:autom}
 \item if $\mfA$ is $\ell$-homogeneous (for $\ell \geq 3$), then the 
structure $\mfC$ is 
$\ell \cdot (d+r)$-homogeneous. 
\label{itm:modint:homogen}
 \end{enumerate}
\end{thm}
\begin{proof}
 Let $\mcI(\tup z)$ be $d$-dimensional with domain formula 
$\lintdom{\phi}(\tup x, \tup z)$ and congruence formula 
$\lintcong{\phi}(\tup x_1, \tup x_2, \tup z)$.
Let $\mfA$ be a $\sigma$-structure and let $\tup a \in \dom(\mfA, \tup z)$.
The elements of the interpreted structure $\mfB = \mcI(\mfA, \tup a)$ are 
equivalence classes of tuples in $\dom(\mfA, \tup x)$.
The idea of the construction of $\mfC$ is as follows.
The universe of $\mfC$ consists of four different sorts $U_A, U_T, U_B, 
U_N$. 
The first sort $U_A$ contains elements that represent the elements in the 
universe of the original structure $\mfA$. 
The second sort contains elements to represent all elements in 
$\dom(\mfA, \tup x)$ that are selected by $\lintdom{\phi}$ and, 
furthermore, a unique element that is used in order to encode the 
parameter tuple $\tup a$. Also $U_T$ contains auxiliary elements to encode 
the structure of tuples, that is the individual entries.
The third sort $U_B$ contains elements to represent the elements 
of the structure $\mfB$, that is the equivalence classes $[\tup b] = \{ 
\tup c \in \dom(\mfA, \tup x) : \mfA \models \lintdom{\phi}(\tup c, \tup 
a) 
\wedge \lintcong{\phi}(\tup b, \tup c, \tup a)\}$ for $\tup b \in 
\dom(\mfA, \tup x)$ with $\mfA \models \lintdom{\phi}(\tup b, \tup a)$.
The last sort $U_N$ is an auxiliary sort which holds a sufficient amount 
of numbers (that is a linearly ordered set) to represent the different 
sorts, their relations, the indices for tuples, and so on.
The binary relation symbol $\in$ is used to relate the different sorts and 
to encode the tuple structure. Relations in $\sigma$ and $\tau$ are 
interpreted on the respective sorts $U_A$ and $U_B$ as in $\mfA$ and 
$\mfB$, respectively.

Let us elaborate more on some technical details (we remark that, as 
usual, there are many different ways to formalise an appropriate encoding; 
in order to verify the properties of $\mcJ$, we describe one of them).
First of all, we extend the dimension of $\mcI$ by three additional 
components $(\mu, \nu, x)$ where the first two variables $\mu, \nu$ range
over the number sort and where $x$ ranges over the vertex sort (we 
remark that it would be sufficient to increase the dimension by at most 
one numeric component, but this would unnecessarily make the following 
description more complicated).
In general, we will use the first component $\mu$ to address different 
sorts. For instance, let us start with the number sort $U_N$. We can use 
the 
congruence formula to merge all tuples $(0, \mu, d, \tup b)$ and 
$(0, \mu, d', \tup c)$ (for $d,d' \in A$ and $\tup b, \tup c \in 
\dom(\mfA, 
\tup x)$) and then use the resulting set $\{ (0, 0, \star), (0,1, \star), 
(0, 2, \star), \dots \}$ to encode the elements in $U_N$. Hereby, we 
choose the 
range of the numeric variable $\nu$ larger than the range of any other 
numeric variable which occurs in the interpretation $\mcI$. 
To identify the numeric sort $U_N$ in the resulting structure we define a 
linear order on $U_N$ using the new relation symbol $\in$.
As a second step, we encode elements $a \in A$ of the original 
structure 
$\mfA$ in $\mfC$ by using elements of the form $(1, \star, a, \star)$ (as 
before, the $\star$'s 
in this notation indicate that we use the congruence formula to merge all 
elements with different $\star$-components).
To identify the first sort $U_A$ in the resulting structure $\mfC$, we 
draw an 
$\in$-edge from the first element in the number sort $U_N$ to all elements 
in 
the first sort $U_A$.
Of course, we define all the relation symbols in $\sigma$ on $U_A$ by 
copying their definition from $\mfA$. 

Thirdly, to encode the elements in $\dom(\mfA, \tup x)$ and the parameter 
tuple 
$\tup a$ we proceed in two steps.
First of all, for every index $1 \leq i \leq \max(k,r)$, we introduce 
component elements $(i, a)$, $a \in A$, and 
$(i, m)$, $m < \dom(\nu)$, to represent all possible components of 
tuples in $\dom(\mfA, \tup x)$. 
Formally, we encode them in $\mfC$ by using elements of the form 
$(2+i,\star,a,\star)$ and $(2+i,m,\star)$. 
To identify them in $\mfC$, we mark them in a similar way as 
before, i.e.\ we introduce $\in$-edges from position $2+i$ in the number 
sort to all component elements $(i, a)$ and $(i,m)$. We also connect 
all component elements $(i,a)$ and $(i,m)$ to their respective values 
$a$ and $m$ via $\in$-edges (which point from component elements to the 
sorts $U_A$ and $U_N$).
We proceed to represent all tuples in $\dom(\mfA, \tup x)$ 
using the original components of the interpretation $\mcI$, that is we use 
elements $(2+\max(k,r)+1, \star, \star, \tup b)$ where $\tup b \in 
\dom(\mfA, \tup x)$ 
and mark them appropriately.
We additionally connect tuples $(2+\max(k,r)+1, \star, \star, \tup b)$ 
with their 
matching component elements, that is with $(i, b_i)$. We then use 
$\lintdom{\phi}$ to select those 
tuples in $\dom(\mfA, \tup x)$ that are in the domain of $\mcI$. 
Also, we add one further special tuple element, say encoded as 
$(2+\max(k,r)+1,0, \star)$, which is meant to encode the parameter tuple 
$\tup a$. This 
special element is thus connected to all component elements $(i,a_i)$.
Finally, we make an additional copy of all tuple elements $(2+\max(k,r)+2, 
\star, 
\star, \tup b)$ that we added, and use $\lintdom{\phi}$ to merge them 
according to $\mcI$. This will give us the sort of elements $U_B$ 
that we use in order to represent the elements of $\mfB$. We mark them 
appropriately in the same way as we did for the other sorts.
Recall that these elements are equivalence classes of elements in 
$\dom(\mfA, \tup x)$, hence we additionally connect them, with 
$\in$-edges, to their representatives in the tuple sort $U_T$.
We define the relations in $\tau$ on $U_B$ according to $\mcI$, that is 
we copy them from $\mfB = \mcI(\mfA, \tup a)$.

From this description it is easy to see that all of the 
required transformations can be expressed by an $\FPC$-interpretation 
$\mcJ$ 
with dimension at most $d + 3$ (and while we increase the number of 
variables by a constant number only).
Also, it should be clear that item~\ref{itm:modint:substr} holds as we can 
very easily define the different sorts in the resulting structure~$\mfC$ 
(in particular, the sort $U_B$ is the maximal sort according to 
our linear order on $U_N$).
Let us now consider item~\ref{itm:modint:autom}.
The main observation is that each automorphism $\pi \in \Delta$ of~$\mfC$ 
is uniquely defined by its projection on the sort $U_A$.
Indeed, this directly follows from the way in which we constructed $\mfC$ 
using the 
new relation symbol $\in$. First note that no automorphism of $\mfC$ can 
move 
elements in the numeric sort $U_N$, since $\in$ defines a linear order on 
$U_N$. In particular it follows that all sorts $U_N, U_A, U_B, U_T$ are 
preserved. Secondly, assume that we have a permutation $\pi$ on $U_A$ that 
can be 
extended to an automorphism of~$\mfC$. Since the $\sigma$-relations on 
$U_A$ in $\mfC$ coincide with the relations in $\mfA$, we know 
that $\pi$ is an automorphism of $\mfA$. 
Moreover, to obtain an automorphism of $\mfC$, there is only one unique 
way in which we can extend $\pi$ to the tuple sort $U_T$ and the sort 
$U_B$ encoding the universe of $\mfB$. 
Indeed, the $\in$-edges enforce that tuple components $(i,a)$ are moved 
to $(i, \pi(a))$ (and tuple components $(i,m)$ cannot be moved) and, 
accordingly, that tuples $\tup b$ in $U_T$ are moved 
to $\pi(\tup b)$. In particular, for the special tuple $\tup a$ this 
means that we have $\pi(\tup a) = \tup a$.
Finally, since $\pi$ extends uniquely to the tuple sort $U_T$, it also 
uniquely extends to the sort $U_B$ of elements of the interpreted 
structure $\mfB = \mcI(\mfA, \tup a)$. Indeed, the elements in $U_B$ are 
sets $[\tup b]$ of tuples $\tup b \in U_T$, and we have connected these 
sets with the elements they contain using $\in$-edges in $\mfC$. Hence, 
for each 
equivalence class $[\tup b] \in U_B$ for $\tup b \in U_T$ we have 
$\pi([\tup b]) = [ \pi(\tup b) ]$. 
So altogether, we can conclude that the extension of $\pi$ from $U_A$ to 
the other sorts $U_N, U_T, U_B$ is unique. On the other hand, note that if 
$\pi$ is an automorphism of $\mfA$ satisfying
$\pi(\tup a) = \tup a$, then the resulting extended $\pi$ is indeed an 
automorphism of 
$\mfC$. To see this, note that all relations in $\tau$ and 
$\sigma$ are preserved under automorphisms of $(\mfA, \tup a)$: for the 
$\sigma$-relations this follows from the 
assumption that $\pi$ is an automorphism of~$\mfA$, and for the 
$\tau$-relation it follows from the fact that they are defined by the 
$\FPC$-interpretation $\mcI$ in $(\mfA, \tup a)$.
This shows that $\Aut(\mfA, \tup a) \cong \Aut(\mfC)$.

Finally, let us consider consider item~\ref{itm:modint:homogen}. Assume 
that $\mfA$ is $\ell$-homogeneous, for $\ell \geq 3$, and let $k \geq 1$.
Let $\tup c = (c_1, \dots, c_k)$ and $\tup d = (d_1, \dots, d_k)$ be two 
$k$-tuples of elements in $\mfC$. 
We have to show that if $(\mfC, \tup c) \equiv^{\ell \cdot (d+r) \cdot 
k} (\mfC, \tup d)$, then there exists an automorphism $\pi$ of $\mfC$ such 
that $\pi(\tup c) = \tup d$.
The main observation is that each element in $\mfC$ is $d$-supported 
by elements of $\mfA$, that is the element can be defined in 
$\FPC$ in $\mfC$ using at most $d$ parameters from $U_A$. More precisely, 
for every element $c$ of $\mfC$ there 
exist at most $d$-many elements $s_1, \dots, s_d \in A = U_A$ for which
there exists an $\FPC$-formula $\psi(x,y_1, \dots, y_d)$ with at most 
$(d+3)$ many variables such that $\psi(x,s_1, \dots, s_d)$ defines $c$ in 
the structure~$\mfC$. For instance, for tuples $c = (b_1, \dots, b_d)$ in 
$U_T$, we can choose $s_1, 
\dots, s_d$ to be the components $b_1, \dots, b_d$ of the tuple, and for 
elements $c = [\tup b] \in U_B$ we can choose the components of some 
representative.
Hence, if $(\mfC, \tup c) \equiv^{\ell 
\cdot (d+r) \cdot k} (\mfC, \tup d)$, then in particular we can find two 
supports $s(\tup c) \in U_A^{d \cdot k}$ for $\tup c$ and $s(\tup d) \in 
U_A^{d \cdot k}$ for $\tup d$ such that
$(\mfC, s(\tup c)) \equiv^{\ell \cdot (d+r) \cdot k} (\mfC, s(\tup d))$.
Moreover, since $\tup a$ is \FPC-definable in $\mfC$ (each component of 
the tuple $\tup a$ is definable using at most three variables), we can 
conclude that  
$(\mfC, s(\tup c), \tup a) \equiv^{\ell \cdot (d+r) \cdot k} (\mfC, s(\tup 
d), \tup a)$.
Since $\mfA$ is $\ell$-homogeneous, we know that 
$\Typk{\ell \cdot (d \cdot k + r)}{d \cdot k}[\tup a]$ defines a 
total preorder on $(d \cdot k)$-tuples in~$\mfA$ which orders tuples up 
to orbits with respect to $\Aut(\mfA, \tup a) \cong \Aut(\mfC)$.
Since $\mfA$ is a definable substructure of $\mfC$ we know that
$(\mfA, s(\tup c), \tup a) \equiv^{\ell \cdot (d+r) \cdot k} (\mfA, s(\tup 
d), \tup a)$.
It follows that we can find a $\pi \in \Aut(\mfC)$ such that $\pi(s(\tup 
c)) = 
s(\tup d)$. Since supports uniquely describe elements, we conclude that 
$\pi(\tup c) = \tup d$, as claimed.
\end{proof}

In what follows, for a given $\FPC$-interpretation $\mcI(\tup z) \in 
\FPC[\sigma \to \tau, \tup z]$, we denote the interpretation $\mcJ(\tup 
z)$ as constructed in Theorem~\ref{thm:nfinter} by $\Norm(\mcI)(\tup z)$.
 
\begin{cor}
 For every $\FPC$-interpretation $\mcI(\tup z)$, there exists $\ell \geq 
1$ such that the class of structures $\{ \Norm(\mcI)(\mfA, \tup a) : \mfA 
\in 
\CFIclass{\mcF}{p}, \tup a \in A^r \}$ is $(\ell, p)$-cyclic.
\end{cor}

Intuitively we showed that the class of $\cyclp$ is closed under 
$\FPC$-interpretations (which, if stated precisely, means that we have 
to rewrite the interpretations in normal form and we have to increase the 
homogeneity constant by a factor depending on the dimension and parameter 
length of the specific interpretation).
We end this section by stating a much simpler observation.
Assume that we have two $\cyclp$-structures $\mfA$ and $\mfB$ of the same 
vocabulary $\tau$. 
Then the ordered pair $(\mfA, \mfB)$ is a $\cyclp$-structure as 
well.
Of course, to some extent this depends on the technical details on how we 
implement ordered pairs as relational structure.
The most important property is that, in the ordered pair
$(\mfA, \mfB)$, we have a simple means to identify the two 
substructures $\mfA$ and $\mfB$, for instance by using additional 
predicate symbols to identify the two universes $A$ and $B$. 
In this article, we agree to understand ordered pairs in this way.
The 
consequence is that the automorphism group of $(\mfA, \mfB)$ is just 
the direct product of the automorphism groups of $\mfA$ and $\mfB$. In 
particular, orbits of (mixed) tuples in $(\mfA, \mfB)$ can be 
described 
in terms of the respective subtuples in $\mfA$ and $\mfB$. Having this, we 
can easily see that the following holds.

\begin{thm}
\label{thm:cyclp:orderedpair}
 Let $\mfA$ and $\mfB$ be two $\cyclp$ structures of vocabulary $\tau$.
 Then the ordered pair $(\mfA, \mfB)$ is an $\cyclp$ structure as 
well.
\end{thm}

\subsection{Solving Cocyclic Linear Equation Systems}
\label{subsec:cocylicLES}
If we want to express $k$-dimensional \PC-refutations over a 
field 
$\mbF$ in $\FPC$, then we need to be able to define solution spaces of 
linear 
equation systems over that field $\mbF$ in $\FPC$, see 
Figure~\ref{fullPCAlgo}.
In Section~\ref{subsec:deflinFPC} we proved that $\FPC$ can 
define solution spaces of linear equation systems over $\mbQ$, and this 
was the key to showing
that $\FPC$ can express $k$-dimensional 
$\PC$-refutations over $\mbQ$ with polynomial bit-complexity, cf.\ Theorem~\ref{thm:PCinFPC}.
Hence, in order to prepare our main result of this section 
(Theorem~\ref{thm:main:PCfinitefields}), we 
are now going to show that $\FPC$ can define solution spaces of linear 
equation systems over a finite field $\mbF$ of characteristic $q$ under 
the assumption that these 
systems are interpreted in a class of $\cyclp$-structures with $q \neq p$. 
Moreover, we show that the number of required variables is
bounded linearly in $\ell$ (with a constant factor that only depends on 
the initial interpretation).
Note that our assumption $q \neq p$ is crucial: the CFI-problem over 
$\field F_p$ cannot be expressed in $\FPC$, 
but it can be reduced (in first-order logic) to the solvability 
problem of linear equation systems over $\field F_p$.

For our proof we make use of a key idea from~\cite{GPCSL15}: in the 
special situation that we consider here, it turns out that 
(solvable) linear  equation system always have symmetric solutions, that 
is solutions which are invariant under all automorphisms of the underlying 
linear equation systems. Together with the property of homogeneity this 
observation allows us to show that $\FPC$ can define such symmetric 
solutions, see~\cite{GPCSL15}.
In this article, we go one important step further.
We not only show that, in this particular setting, we can define the 
\emph{Boolean} solvability problem for 
linear equation systems in $\FPC$, but that we can also define the more 
general 
\emph{functional} problem of expressing solution spaces of given linear 
equation 
systems. 

Let us remark that our results here extend our approach 
from~\cite{GPCSL15} in another crucial way. In~\cite{GPCSL15} we 
considered the CFI-construction with respect to underlying graphs of 
unbounded degree. The reason is that, if we work with such underlying 
graphs, then this considerably simplifies the proof 
of the homogeneity property for CFI-structures. Here, in contrast, we 
consider the ``full'' power of the
CFI-construction, that is with respect to a family of underlying 
three-regular 
expander graphs.  This has the effect that we get much better lower 
bounds on the number of variables, and this makes our separation results 
even stronger. That is to say that the techniques that we develop here can 
readily be used in order to strengthen our separation 
results from~\cite{GPCSL15} to formulas with a sublinear number of 
variables (rather than a constant number as we considered 
in~\cite{GPCSL15}).

As usual, in order to talk about systems of linear 
equations over finite fields in the context of logical definability, we 
first have to agree on an encoding of such systems as finite relational 
structures.
Again, the concrete choice does not matter, so we do not specify 
such an encoding explicitly. Let us rather go through some notation  that 
we use in this section.
We consider (unordered) matrices $M$ over a finite field 
$\field F$ as mappings 
$M \colon I \times J \to \field F$ for two (non-empty) index sets $I$ 
and~$J$.
An (unordered) vector $v$ over a finite field $\field F$ is a 
mapping $v\colon I \to \field F$.
A linear equation system $M \cdot x = b$ over a finite field $\field F$ is 
specified by an $I \times J$-coefficient matrix $M$ over $\field F$ and an 
$I$-constants vector $b\colon I \to \mbF$.
We usually think of the finite field $\field F$ as being part of the 
input.
We are primarily interested in the setting where the 
characteristic $q = \characteristic(\mbF)$ of this field $\field F$ and 
the prime $p$ for
CFI-class $\CFIclass{\mcF}{p}$ are distinct:

\begin{defi}
 Let $\ell \geq 1$. 
 We say that a $\tau$-structure $\mfA$ contains an \emph{\lcocy} 
vector, (or matrix, or linear equation system) over a finite field $\field 
F$ if
\begin{itemize}
 \item the structure $\mfA$ is $\cyclp$ for some prime $p \in \Primes$, and
 \item for some distinguished relation symbol $S \in \tau$, the 
substructure of $\mfA$ induced on $S$ is (the structural encoding of) a 
vector $v\colon I \times \mbF$ (or 
matrix $M \colon I\times J \to \mbF$, or linear equation system $M \cdot x 
= b$) over the finite field $\mbF$ with characteristic different 
from~$p$, that is $\characteristic(\mbF) = q$ for some $q \in \Primes$,
 $p \neq q$.
\end{itemize}
\end{defi}

We proceed to show that $\FPC$ can express solution spaces of 
$\ell$-cocyclic linear equation systems using $\mcO(\ell)$ many 
variables only.
The proof consists of two steps. First of all, we show that $\FPC$ 
can define a single solution of a (solvable) \lcocy linear equation 
system (Theorem~\ref{thm:fpc-singsol}).
In a second step we then show that $\FPC$ can also define (small) 
generating 
sets for kernels of \lcocy matrices (Theorem~\ref{thm:fpc-ker}). By 
putting these two results together, we obtain the desired result. 
The main idea for this second step is to repeatedly make use of the 
$\FPC$-formula from Theorem~\ref{thm:fpc-singsol} for solving \lcocy 
linear 
equation systems and the fact 
that \cyclp structures can be linearly ordered locally in $\FPC$.

\begin{thm}\label{thm:fpc-singsol}
 For every $\ell \geq 1$ there exists an $\FPC$-formula $\phi$ with 
$\mcO(\ell)$ many variables such that $\phi$ defines in every structure 
$\mfA$ that contains
a solvable \lcocy linear equation system $M \cdot x 
= b$ over a finite field $\mbF$, where $M \colon I \times J \to \mbF$ and 
$b\colon I \to \mbF$, a solution to $M \cdot x = b$, that is $\phi$ 
defines a vector $v\colon J \to \mbF$ such that $M \cdot v = b$ (and, if 
$M 
\cdot x = b$ is not solvable, then, by convention, $\phi$ defines the 
all-$0$-vector in $\mfA$).
\end{thm}
\begin{proof}
 Let $\mfA$ be an $\cyclp$ structure which contains a linear equation 
system $M \cdot x  = b$ for a matrix $M \colon I \times J \to \mbF$ and a 
constants vector $b\colon I \to \mbF$ over a finite field $\mbF$ of 
characteristic $\characteristic(\mbF) = q$, $p \neq q$.
Let $\Gamma = \Aut(\mfA)$ denote the automorphism group of $\mfA$.
Then $\Gamma$ acts on the solution space of $M\cdot x = b$.
We know that this space (in case that it is non-empty) has size $q^i$ for 
some $i \geq 0$, since we are 
dealing with a linear equation system over a field of characteristic $q 
\in \Primes$.
On the other hand, recall that $\Gamma$ is a $p$-group which means that 
each orbit of the action of $\Gamma$ on the solution space of $M \cdot x = 
b$ has size $p^j$ for some $j \geq 0$.
We conclude that there has to be at least one orbit of size one.
This, however, means that there is a solution $v\colon J \to \field F$ 
such 
that 
$\pi(v) = v$ for all $\pi \in \Gamma$.
We call a vector $v\colon J \to \field F$ which satisfies this property 
\emph{symmetric}. Note that a symmetric vector $v\colon J \to \field F$ is 
constant on 
the orbits induced by $\Gamma$ on the set $J$ since $\pi(v)(j) = 
v(\pi^{-1}(j))$.
By our assumption that $\mfA$ is $\cyclp$, we know that the formula 
$\Typk{\ell}{1}(x,y)$ defines a linear preorder $\preceq$ on $J$ which 
linearly orders $J$ up to $\Gamma$-orbits. 
Moreover, recall that this $\FPC$-formula $\Typk{\ell}{1}$ only uses 
$\mcO(\ell)$ variables.
Let $J = J_0 \preceq J_1 \preceq \cdots \preceq J_{n-1}$ denote the 
$\Gamma$-orbit partition of $J$.

For $0 \leq i < n$ let $t_i\colon J \to \field F$ denote the $J$-vector 
which 
is the identity on the $i$-th $J$-orbit, that is $t_i(j) = 1$ for $j \in 
J_i$ and $t_i(j) = 0$ for $j \not\in J_i$.
Let $T$ denote the $J \times \{0, \dots, n-1\}$-matrix which has $t_i$ as 
its $i$-th column.
Then for every symmetric $v\colon J \to \field F$ we can find a vector
$w \colon \{ 0, \dots, n-1 \} \to \field F$ such that $Tw = v$. Indeed, 
just 
choose $w(i) = v(j)$ for (some) $j \in J_i$.
We conclude, that the linear equation system $M\cdot x = b$ is solvable 
if, and only if, the system $M \cdot T \cdot x = b$ is solvable.
Clearly, every solution of $M \cdot T \cdot x = b$ gives rise to a 
solution of $M \cdot x = b$.
Hence, it suffices to define a solution of $M \cdot T \cdot x = b$ in 
fixed-point logic with counting.
However, this is very easy because $M \cdot T$ is an $I\times \{0, \dots, 
n-1\}$-matrix which has a linearly ordered set of columns.
Moreover, if we drop duplicates of rows, then the order on the columns 
also induces a (first-order definable) linear order on the rows, namely 
the lexicographical ordering (note that there exists an $\FO$-definable 
order on the finite field $\field F$). 
It follows by the Immerman-Vardi Theorem that fixed-point logic can
define a solution of the system $M \cdot T \cdot x = b$ or determine that 
the original system was not solvable. This solution 
can be lifted to a solution of $M \cdot x = b$ by multiplying 
by $T$. Finally, observe that the number of variables in the resulting 
formula is independent of $\ell$ except for the subformula 
$\Typk{\ell}{1}$ 
which defines the linear order on the orbit-partition of $J$. Hence, the 
required number of variables is indeed $\mcO(\ell)$. 
\end{proof}

\begin{thm}\label{thm:fpc-ker}
 For every $\ell \geq 1$ there exists an $\FPC$-formula $\phi$ 
with $\mcO(\ell)$ variables which defines in every structure $\mfA$ 
that contains an \lcocy matrix $M\colon I \times J \to \field F$ over a 
finite field $\field F$, a 
matrix $\phi^\mfA\colon J \times (J \times |J|) \to \field F$ such that 
$\im(\phi^\mfA) = \kernel(M)$.
\end{thm}
\begin{proof}
 Let $\mfA$ be an $\cyclp$ structure with automorphism group 
$\Gamma=\Aut(\mfA)$, and assume that $\mfA$ contains a matrix $M\colon I 
\times J 
\to \field F$ over a finite field $\field F$ of characteristic 
$\characteristic(\field F) = q \neq p$.
 First of all, we again use the formula $\Typk{\ell}{1}(x,y)$ to define a 
total 
preorder $\preceq$ on $J$ which orders the indexing elements in $J$ up to 
$\Gamma$-orbits.
Let $J = J_0 \preceq J_1 \preceq \cdots \preceq J_{n-1}$.
Recall that $\Typk{\ell}{1}$ is an $\FPC$-formula with $\mcO(\ell)$ 
variables.
Our plan is as follows. We aim to define a generating set for $\kernel(M)$ 
which consists of \emph{$i$-homogeneous} vectors for $0 \leq i < n-1$.
Here we say that a vector $v\colon J \to \field F$ is $i$-homogeneous if 
$v(j) 
= 0$ for all $j \in J_{i'}$ for $i' < i$. That is an $i$-homogeneous 
vector is zero on all $\Gamma$-orbits on $J$ which precede the $i$-th 
orbit~$J_i$.
Our plan is to define in $\FPC$, for every $0 \leq i < n$,  sets 
$K_i$ consisting of $i$-homogeneous vectors $v\colon J \to \field F$, $v \in 
\kernel(M)$, such that the 
projections of $K_i$ to $J_i$ yield generating sets for the projections 
of $\kernel(M)$ to $J_i$, which means that $\bigcup_{i < n-1} K_i$ is a 
generating set for $\kernel(M)$.

We will index the elements in $K_i$ by elements in $J_i 
\times |J|$. The crucial insight is that $\cyclp$ structures satisfy the 
additional property that for each fixed parameter $j \in J_i$  the 
formula $\Typk{\ell \cdot 2}{1}[j](x,y)$ defines a linear order $<_j$ on 
$J_i$.
Hence, if we have a fixed $j \in J_i$, then it makes sense to speak of 
an \emph{$m$-th echelon vector} $(0, \dots, 0, 1, \star, \star)$ of the 
projection of $\kernel(M)$ to $J_i$. Here an $m$-th echelon vector has 
entry 
$1$ in the $m$-th component and is zero at all preceding components (and 
its length is $|J_i|$). Clearly, for some $0 \leq m \leq |J_i|$ such a 
vector may not exist, but if we collect a set of (existing) $m$-th echelon 
vectors for $0 \leq m \leq |J_i|$, then we obtain a generating set for 
the projection of $\kernel(M)$ to $J_i$.
With this preparation, we can describe our strategy more 
precisely. The intention is that the $J$-vector $v \in K_i$ that is 
indexed 
by $(j, m)$, $j \in J_i$, $m < |J|$, represents an $i$-homogeneous vector 
$v \in \kernel(M)$ with the additional property that the projection of $v$ 
to 
$J_i$ is the $m$-th echelon vector of the projection of $\kernel(M)$ to 
$J_i$ 
(if it exists, otherwise we agree to let $v = 0$).
Note that in this way we actually include too many vectors in $K_i$. 
Indeed, it would be sufficient to consider all such vectors indexed by $j 
\times |J_i|$ for a single $j \in J_i$. However, since we cannot choose a 
particular $j \in J_i$ we just add all of these vectors for any $j 
\in J_i$. This 
does not cause any problems, since we do not aim at defining a basis for 
$\kernel(M)$, but just at defining a generating set.

It remains to see how we can define such a vector $v\colon J \to \field F$ 
in 
$\FPC$ given parameters $(j, m) \in J_i \times |J|$.
To this end we make use of Theorem~\ref{thm:fpc-singsol} and the formulas 
$\Typk{\ell}{1}$ and $\Typk{\ell \cdot 2}{1}[j]$ (both with $\mcO(\ell)$ 
many variables only).
Let us start with the homogeneous linear equation system $M \cdot x = 0$ 
which defines $\kernel(M)$. Given the parameters $(j,m)$, we now add 
extra constraints for the variables 
$x = (x_j)_{j \in J}$ as follows:
\begin{itemize}
 \item for $j' \in \bigcup_{i' < i} J_{i'}$ we set $x_{j'} = 0$,
 \item for $J_i = j_1 <_j \cdots <_j j_s$, we set $x_j = 0$ for $j \in \{ 
j_1, \dots, j_{m-1} \}$ and $x_{j_m} = 1$.
\end{itemize}
It is clear that the solution space of this linear equation system 
consists precisely of the $i$-homogeneous vectors $v \colon J \to \field 
F$ 
in 
$\kernel(M)$ whose projections to $J_i$ are $m$-th echelon vectors (with 
respect to the order $<_j$ defined by $\Typk{\ell \cdot 2}{1}[j]$ on 
$J_i$).
Moreover, this system can easily be defined in $\mfA$ using an 
$\FPC$-formula which uses $\Typk{\ell \cdot 2}{1}[j]$ and $\Typk{\ell}{1}$ 
as subformulas and parameters $(j,m)$. 
We can now make use of Theorem~\ref{thm:fpc-singsol} to define a solution 
$v\colon J \to \field F$ of this system (if a solution exists) in $\FPC$ 
using 
again $\mcO(\ell)$ many variables only. This yields the desired vector in 
$K_i$, that is indexed by $(j,m)$, and it concludes our proof.
\end{proof}

By putting Theorem~\ref{thm:fpc-singsol} and Theorem~\ref{thm:fpc-ker} 
together we arrive at our desired result, namely that $\FPC$ is able to express 
solution spaces 
of \lcocy linear equation systems $M \cdot x = b$ where $M\colon I 
\times J \to \field F$ and $b \colon I \to \field F$ using $\mcO(\ell)$ 
many variables only.
Unfortunately, there is still a small problem: 
according to Theorem~\ref{thm:fpc-ker}, the index set for the solution 
space that we get is $(J \times |J|)$. However, in general, $|J|$ can be 
much larger than $|I|$, and we would like to get \emph{small} generating 
sets for expressing $\PC$-refutations when we think of our procedure from
Figure~\ref{fullPCAlgo}. In fact, when we express $k$-dimensional 
\PC-refutations in $\FPC$, then for the 
linear equation systems that we need to solve there, we only have a global 
polynomial bound on the size of the index $I$ (the set of $k$-dimensional 
monomials), but not on 
the size of the index $J$ (which indexes the generating set for $\PCx 
k(\mcP)$ that we have computed up to a certain 
stage, cf.\ Figure~\ref{fullPCAlgo}).
Fortunately, we can use the same strategy that we used in order to prove 
Theorem~\ref{thm:fpc-ker} in order to convert a (potentially large) 
generating set for a given linear space into a small one within $\FPC$.

\begin{thm}
\label{thm:fpc-smallgen}
 For every $\ell \geq 1$ there exists an $\FPC$-formula 
$\phi$ with $\mcO(\ell)$ variables such that for every structure $\mfA$ 
which contains an \lcocy matrix $M\colon I \times J \to \field F$ over a 
finite field $\field F$, the formula $\phi$ defines in $\mfA$ a
matrix $\phi^\mfA\colon I \times (I \times |I|) \to \field F$ such that 
$\im(\phi^\mfA) = \im(M)$.
\end{thm}
\begin{proof}
 The proof is analogous to our proof of Theorem~\ref{thm:fpc-ker}.
\end{proof}

\begin{cor}
\label{cor:fpc:finitefield:solvles}
  For every $\ell \geq 1$ there exist $\FPC$-formulas 
with $\mcO(\ell)$ variables such that for every structure $\mfA$ 
which contains an \lcocy linear equation system $M \cdot x = b$ for 
$M \colon I \times J \to \field F$ and $b\colon I \to \field F$ over a 
finite field $\field F$, the formulas either define a 
matrix $N\colon J \times K \to \field F$ and a $J$-vector $v\colon J \to 
\field F$ such that $\im(N) + v$ is the solution space of $M \cdot x = b$ 
where
$K \in \{ I \times |I|, J \times |J| \}$ and $|K| = \min(|I|^2, |J|^2)$,
or, in case that the solution space is empty, they define $v = \emptyset$.
\end{cor}

\subsection{Cocyclic PC-Refutations over Finite Fields in FPC}
\label{subsec:cocycPC}
We can finally come to our main result of this section. We show that 
$\FPC$ 
can express $k$-dimensional \PC-refutations over finite fields $\field F$ 
of characteristic $q$ if the inputs are polynomial equation systems 
that are interpreted in a class of $(\ell,p)$-cyclic structures, where $q 
\neq p$, 
using only $\mcO(\ell)$ variables.
As the prototype example, this situation occurs whenever we interpret 
polynomial equation systems over a field $\field F$ of characteristic 
$\characteristic(\field F) = q$ in (disjoint unions) of CFI-structures 
$\mfA \in \CFIclass{\mcF}{p}$ over $\field F_p$.
For the proof, recall that for an $\FPC$-interpretation~$\mcI$, we denote 
by 
$\Norm(\mcI)$ its normal form according to Theorem~\ref{thm:nfinter}.

\begin{thm}
\label{thm:main:PCfinitefields}
Let $Q \subsetneq \Primes$ be a (non-trivial) set of primes.
Let  $\mcI(\tup x)$ be an $\FPC$-interpretation which 
maps 
$\tau$-structures to polynomial equation systems over finite 
fields~$\field F$ of characteristic $q \in Q$.
Then for every $\ell \geq 1$, $k \geq 2$, and $p \in \Primes, p \not\in 
Q$, there exists an $\FPC$-formula $\phi$ with $\mcO(k\cdot \ell)$ many 
variables such that for every $\cyclp$ $\tau$-structure $\mfA$ and $\tup a 
\in \dom(\mfA,\tup x)$ we have that $\mfA \models \phi(\tup a)$ if, and 
only if, the polynomial equation system $\mcI(\mfA, \tup a)$ has a 
$\PC$-refutation (over the respective finite field $\field F$) of degree 
at most $k$.
\end{thm}
\begin{proof}
Let $d \geq 1$ denote the dimension of $\mcI(\tup x)$ and $r \geq 0$ the 
number of parameters, $r = |\tup x|$.
We use Theorem~\ref{thm:nfinter} to transform $\mcI(\tup x)$ into normal 
form $\Norm(\mcI)(\tup x)$. 
Then, if $\mfA$ is an $\cyclp$ $\tau$-structure and $\tup a \in \dom(\mfA, 
\tup x)$, then we know that $\Norm(\mcI)(\mfA, \tup a)$ is $(\ell \cdot (d 
+r))$-homogeneous and that the automorphism group of $\Norm(\mcI)(\mfA, 
\tup a)$ is an Abelian $p$-group.
We conclude, using Theorem~\ref{thm:homab:lpcyc}, that 
$\Norm(\mcI)(\mfA,\tup a)$ is $(\ell \cdot (d+r),p)$-cyclic. Note that $d, 
r$ are 
constants which only depend on the fixed interpretation $\mcI$.

Now, assume that we want to express in $\FPC$, given $\Norm(\mcI)(\mfA, 
\tup a)$, whether the contained polynomial equation system 
$\mcP$ over the finite field $\mbF$ of characteristic $q \neq p$ 
has a $k$-dimensional \PC-refutation.
In order to do this, we want to express the procedure from 
Figure~\ref{fullPCAlgo} in $\FPC$. Recall that the main (and only) 
difficulty 
is to (iteratively) define solution spaces of linear equation system over 
$\field F$ in $\FPC$.
However, since $\mcP$ is part of an $(\mcO(\ell),p)$-cyclic structure, all 
linear equation systems that we have to solve are $\mcO(\ell)$-cocyclic 
systems.
Since we can define the index sets for these systems using $\mcO(k)$ many 
variables in $\FPC$ (because we basically have to index all degree-$k$ 
multilinear monomials) it follows from 
Corollary~\ref{cor:fpc:finitefield:solvles} that solution sets can 
be defined in $\FPC$ using at most $\mcO(k \cdot \ell)$ many variables.
We can now translate the resulting formulas back via $\Norm(\mcI)$ which 
adds another constant factor to the number required variables that only 
depends on $\mcI$. This concludes our proof.
\end{proof}

As we said, in particular, we can apply this result for polynomial 
equation systems 
interpreted in CFI-structures. This 
will allow us to  derive lower bounds for the polynomial calculus 
just by using finite-model-theoretic arguments in 
Section~\ref{sec:lowerboundsPC}.

\section{Applications in Proof Complexity}
\label{sec:lowerboundsPC}
Our model-theoretic characterisations  of (bounded-width) resolution 
and the polynomial calculus via $\EFP$- and $\FPC$-definability allow us 
to 
uniformly (re-)prove many lower bounds on the complexity of proofs (size 
and/or width/degree) for families of propositional formulas using 
arguments from finite model theory.
The basic idea is very simple. We saw that the amount of certain logical 
resources that are required to express refutations (that is the number of 
variables) matches the 
complexity of refutations (width of clauses or degree of polynomials) up 
to linear factors.
It follows that if we exhibit families of propositional 
formulas $\Phi_n, \Psi_n$ that cannot be distinguished in $\EFP$ (or in
$\FPC$ or in $\Cinf$) using $\mcO(k)$ variables, then also the corresponding 
propositional proof systems cannot distinguish between these 
formulas using refutations of width $k$ or degree $k$, respectively. In 
particular, if one of the formulas in our family 
$(\Phi_n, \Psi_n)$, say $\Phi_n$, is satisfiable, then there cannot be a 
refutation for the indistinguishable formula $\Psi_n$ (of a certain 
complexity).

In the conference version of this article~\cite{GrPaPa17} we discussed 
these applications with respect to the resolution proof system. 
However, given that we extended our definability results for the 
polynomial calculus in this article, we can basically derive the same 
lower 
bounds directly for the full polynomial calculus over the rationals and 
over finite fields (with the Pigeonhole principle being the only 
exception). Clearly, this makes the lower bounds more interesting 
and, for conciseness, we therefore restrict our attention to the 
polynomial calculus here.

\subsection{Lower Bounds on Degree and Size of Refutations}
\label{sec:applications:low}
In this section we establish our main tool for proving lower bounds for 
the polynomial calculus.
Recall the notion of $(\ell, p)$-cyclic structures from 
Section~\ref{subsec:cyclic}. 
 
\begin{thm}
\label{thm:appl:pclow}
Let $\field F$ be a finite field or the field of rationals. Moreover, let 
$(\mcP_n)$ and $(\mcQ_n)$ be two families of 
polynomial equation systems over $\mbF$ and let 
$\mcI$ be an $\FPC$-interpretation that maps $\tau$-structures to 
polynomial equation systems over~$\field F$.
In addition, let $\ell \geq 1$ and let $p \in \Primes$ be such that 
$\characteristic(\field F) \neq p$ and let $(\mfA_n)$ and $(\mfB_n)$ be 
two families of $(\ell,p)$-cyclic $\tau$-structures such that for all $n 
\geq 1$:
\begin{itemize}
 \item $\mcI(\mfA_n) = \mcP_n$ and $\mcI(\mfB_n) = \mcQ_n$,
 \item $\mcP_n$ is satisfiable and $\mcQ_n$ is not satisfiable,
 \item $\mfA_n \equiv^{\Omega(n)} \mfB_n$.
\end{itemize}
Then the following holds:
\begin{enumerate}
 \item Let $\PCdegree(n)$ denote the minimal degree required to 
refute the system $\mcQ_n$ using the polynomial calculus over $\field F$. 
Then $\PCdegree(n) \in \Omega(n)$.
\end{enumerate}
Moreover, as a consequence of this, the following holds:
\begin{enumerate}
\stepcounter{enumi}
\item Let $\PCsize(n)$ denote the 
size of a minimal PC-refutation 
for $\mcQ_n$ over $\field F$. If the systems $\mcQ_n$, for $n \geq 1$, 
only contain $\mcO(n)$ many variables, then $\PCsize(n)$ is bounded from below 
by $2^{\Omega(n)}$.
\end{enumerate}
 \end{thm}
\begin{proof}
 For any degree $k \geq 1$, we know that by Theorem~\ref{thm:PCinCinf} (if 
$\field = \mbQ$) or Theorem~\ref{thm:main:PCfinitefields} (if $\field F$ 
is finite and of characteristic $q \neq p$) there exists a $\Cinf$-formula 
$\phi_k$ (in the case of finite fields, there even exists an $\FPC$-formula $\phi_k$, but this makes no difference for the argument) with $\mcO(k)$ many variables (note that $\ell$ is fixed) which 
expresses whether the polynomial equation systems $(\mcP_n)$, $(\mcQ_n)$ 
have a $\PC$-refutation over $\field F$ of degree at most $k$.
By translating these formulas back via the fixed $\FPC$-interpretation 
$\mcI$ (where we use $\Norm(\mcI)$ in case of finite fields) we obtain 
$\Cinf$-formulas $\psi_k$ with $\mcO(k)$ many variables such that $\mfA_n 
\models \psi_k$ if, and only if, $\mcP_n$ has a degree $k$ PC-refutation 
over $\field F$, and likewise for $\mfB_n$ and $\mcQ_n$.
However, since $\mcP_n$ is satisfiable it has no such refutation for any 
degree $k \geq 1$. Since $\mfA_n \equiv^{\Omega(n)} \mfB_n$, it thus 
follows that for $k \in \Omega(n)$, also $\mcQ_n$ has no such degree 
$k$-refutation. This proves our first claim.
The second claim follows from the size-degree trade-off for the 
polynomial calculus, see~\cite[Corollary 6.3]{ImPuSg99}.
\end{proof}

\subsection{Lower Bounds for the Graph Isomorphism Problem}
\label{subsec:lowbounds:gi}
We now discuss the prototype example for the lower bound technique on 
$\PC$-refutations (Theorem~\ref{thm:appl:pclow}). Specifically, we show 
that 
the graph 
isomorphism problem does not allow small PC-refutations neither over  
$\mbQ$ nor over finite fields. 
This result has already been established by Berkholz and Grohe 
in~\cite{BerkholzGro15,BerkholzGro17} by using known lower bounds for the 
polynomial calculus. Here, we present an alternative proof of (a 
generalisation of) their result using only arguments from finite model 
theory.

Given two graphs $G=(V,E)$ and $H=(W,F)$ it is easy to express the graph 
isomorphism problem for $G$ and $H$ as a polynomial equation system 
$\IsoForm(G,H)$ over any field $\mbF$ as follows.
We use variables $X[v \mapsto w]$, for $v \in V$ and $w \in W$, to 
indicate whether $v$ is mapped to $w$ by an isomorphism (that we are going 
to guess as a solution). We include the Boolean constraints $X^2 - 
X = 0$ as usual, i.e.\ $X[v \mapsto w] \in \{0 , 1\}$ for every solution.
Then we just have to express that every vertex $v \in V$ is mapped to 
precisely one $w \in W$: $\sum_w X[v \mapsto w] = 1$, and, dually, that 
every $w \in W$ has precisely one preimage $v \in V$: $\sum_v X[v \mapsto 
w] = 1$.
Finally we want that edges are preserved. We can achieve this by including 
for each $v_1, v_2 \in V$ and $ w_1, w_2 \in W$ such that $(v_1, v_2) \in 
E$ if, and only if, $(w_1,w_2) \not\in F$ the equation $X[v_1 \mapsto 
w_1]\cdot X[v_2 \mapsto w_2] = 0$.
It is this (fixed) encoding that Berkholz and Grohe considered 
in~\cite{BerkholzGro15,BerkholzGro17} in order to prove their 
lower bounds.
Interestingly, we can easily lift their result to a more general 
setting, namely we can 
allow arbitrary encodings of the graph isomorphism problem that are 
definable in $\FPC$ or even in $\Cinf$ and still obtain the same lower bounds.

\begin{thm}
\label{thm:low:gi}
Let $\field F$ be the field of rationals or a finite field. 
 Let $\mcI$ be an $\FPC$-interpretation that maps pairs of graphs $(G, H)$ 
to polynomial equation systems over $\field F$ such that $\mcI(G,H)$ is 
solvable if, and only if, $G$ and $H$ are isomorphic.
Then there exists a sequence $(G_n, H_n)$ of pairs of non-isomorphic 
graphs $G_n, H_n$ with bounded degree and of size $\mcO(n)$ such that 
$\PC$-refutations for the systems $\mcI(G_n, H_n)$ over $\field F$ require 
degree 
$\Omega(n)$.
\end{thm}
\begin{proof}
 Choose $p \in \{ 2, 3 \}$ such that $\characteristic(\field F) \neq p$.
 We consider the class of CFI-structures $\CFIclass{\mcF}{p}$ over $\field 
F_p$. Recall that $\mcF = \{ F_n : n \geq 1\}$ is a family of $3$-regular, 
connected expander graphs where $F_n$ has $\mcO(n)$ many vertices.
 In Lemma~\ref{lem:encodeCFI:asgraph} we observed that we can encode 
such CFI-structures as undirected graphs via $\FPC$-interpretations $\mcJ$ 
(with a corresponding inverse interpretation $\mcJ^{-1}$).
 Moreover, for a CFI-structure $\mfA=\CFIgraph{F_n}{p}{\lambda} \in 
\CFIclass{\mcF}{p}$, the graph $\mcJ(\mfA)$ encoding $\mfA$ has degree 
$\mcO(p^2)$ and contains $\mcO(p^2 \cdot n)$ vertices. Since $p \in \{ 2, 
3\}$, it follows that the graphs $\mcJ(\mfA)$ have bounded degree and 
contain $\mcO(n)$ vertices only.

For $n \geq 1$ we fix two non-isomorphic CFI-structures $\mfA_n, \mfB_n 
\in \CFIclass{\mcF}{p}$ with underlying graph $F_n$ over $\field F_p$.
We let $G_n = \mcJ(\mfA_n)$ and $H_n = \mcJ(\mfB_n)$. We claim that 
the resulting sequence $(G_n, H_n)$ satisfies the above claim.
To show this we use Theorem~\ref{thm:appl:pclow}. 
Let $\mcP_n = \mcI(G_n, G_n)$ and $\mcQ_n = \mcI(G_n, H_n)$.
Then we observe that $\mcP_n$ and $\mcQ_n$ can be interpreted in the 
structures $(\mfA_n, \mfA_n)$ and $(\mfA_n, \mfB_n)$ via $(\mcI \circ 
\mcJ)$ (where, formally, we need to slightly modify $\mcJ$ to encode 
ordered pairs of CFI-structures as ordered pairs of graphs).
By the CFI-Theorem~\ref{fpc-cfi} we know that $(\mfA_n, \mfA_n) 
\equiv^{\Omega(n)} (\mfA_n, \mfB_n)$.
By Theorem~\ref{thm:cyclp:orderedpair} and Theorem~\ref{thm:homab:lpcyc}, 
we know that the structures $(\mfA_n, \mfA_n)$ and $(\mfA_n, \mfB_n)$ are 
$\cyclp$ for some fixed $\ell \geq 1$.
Moreover, by our assumption on $\mcI$, the systems $\mcP_n$ are 
satisfiable and the systems $\mcQ_n$ are not satisfiable.
Thus, all preconditions of Theorem~\ref{thm:appl:pclow} are met, and the 
lower bound follows.
\end{proof}

Although Theorem~\ref{thm:low:gi} gives us the desired linear lower bound 
on the degree of PC-refutations for the graph isomorphism problem, we can 
not readily infer the exponential size lower bound from 
Theorem~\ref{thm:appl:pclow}. 
The reason is that the polynomial equation systems which encode the graph 
isomorphism problem might contain more than a linear number of 
variables. In fact, the number of variables in the system 
$\IsoForm(G,H)$ that we defined above contains a quadratic number of 
variables. This means that the size-degree trade-off results for the $\PC$ 
cannot be applied.

However we can fix this as follows. In our proof we used CFI-graphs 
and these are graphs of \emph{bounded colour class size}.
Formally, a \emph{graph with colour class size $k\geq 1$} is a structure 
$G=(V,E, \preceq)$ where $(V,E)$ is a graph and where $\preceq$ is a 
linear preorder on $V$ such that every class of $\preceq$-incomparable 
vertices, that is every \emph{colour class}, is of size at most $k$.
In other words, one can think of the vertices of the graph $G$ to be 
coloured while we only allow that at most $k$ vertices get the same 
colour.
We write $V = V_0 \preceq \cdots \preceq V_{n-1}$ to denote that $V$ is 
linearly ordered by $\preceq$ into $n$ colour classes $V_i$ in the 
indicated way. We have that $|V_i| \leq k$ for every $i < n$.

For CFI-graphs (that is graphs $\mcJ(\mfA)$ for $\mfA \in 
\CFIclass{\mcF}{p}$ and where $\mcJ$ is the graph encoding of 
CFI-structures from Lemma~\ref{lem:encodeCFI:asgraph}) the colour classes 
are basically given as the edge classes of the underlying graph plus the 
additional classes of inner nodes which encode the CFI-constraints, see 
our discussion preceding Lemma~\ref{lem:encodeCFI:asgraph}. The size of 
these classes is at most $\mcO(p^2)$. Since in our proof we can restrict 
to CFI-graphs over the field $\field F_p$ with $p \in \{ 2, 3 \}$, these 
edge classes are indeed of constant size.
Hence, it follows from our proof above that we can require the family 
of graphs $(G_n, H_n)$ in Theorem~\ref{thm:low:gi} to consist of graphs 
of bounded colour class size.

Now, restricted to graphs of bounded colour class size, our encoding 
$\IsoForm(G,H)$ for the graph isomorphism problem that we introduced above 
can naturally be simplified resulting in a polynomial equation system 
that uses linearly many variables only. To see this, we consider pairs of 
graphs $G=(V,E,\preceq_V)$ and $H=(W,F,\preceq_W)$ of 
colour class size $k\geq 1$ with the same number of colour classes, that is
\begin{align*}
 V &= V_0 \preceq_V  V_1 \preceq_V \cdots \preceq_V V_{n-1} \\
 W &= W_0 \preceq_W W_1 \preceq_W \cdots \preceq_W  W_{n-1}.
\end{align*}
Then each isomorphism is restricted to map vertices in the $i$-th 
colour class $V_i$ in $G$ to the $i$-th colour class $W_i$ in $H$. That 
means that in our system $\IsoForm(G,H)$ we only need to 
include variables $X[v \mapsto w]$ for all $v \in V_i, w \in W_i$, $i < 
n$.
Since the colour classes are of constant size, this means that the 
resulting system only contains a linear number of variables.
Hence, we obtain the following strengthening of
Theorem~\ref{thm:appl:pclow} for this setting.

\begin{thm}
\label{thm:low:gi:bcc}
Let $\field F$ be the field of rationals or a finite field. 
 Let $\mcI$ be an $\FPC$-interpretation that maps pairs of  graphs $(G, 
H)$ of \emph{bounded colour class size} to polynomial equation systems 
over $\field F$ such that $\mcI(G,H)$ is 
solvable if, and only if, $G$ and $H$ are isomorphic, and, moreover 
$\mcI(G,H)$ contains a linear number of variables only (linear with 
respect to the number of vertices of $G$ and $H$).
Then there exists a sequence $(G_n, H_n)$ of pairs of non-isomorphic 
graphs $G_n, H_n$ with bounded degree, of size $\mcO(n)$, and of bounded 
colour class size such that 
$\PC$-refutations for the systems $\mcI(G_n, H_n)$ over $\field F$ require 
degree 
$\Omega(n)$ and size $2^{\Omega(n)}$.
\end{thm}

\subsection{Monomial-PC versus (Full-)PC over the Field of Rationals}
\label{sec:gi}
As mentioned above, in~\cite{BerkholzGro15} Grohe and Berkholz studied 
the power 
of the polynomial 
calculus with respect to the graph isomorphism problem. 
One of their main results is that the monomial-PC over $\mbQ$ has 
precisely the same expressive power as the well-known Weisfeiler-Leman 
graph isomorphism test which, in turn, has the same expressive power as 
counting logic (with respect to isomorphism testing).
However, they left open the question of whether the full polynomial 
calculus is more expressive than its restricted variant the monomial-PC 
over $\mbQ$ with respect to the graph isomorphism problem.

\begin{thmC}[\cite{BerkholzGro15}]
For all $k \geq 2$ and graphs $G,H$ we have that
\[ G \not\equiv^k H \text{ if, and only if, } G \not\equiv^{\monPCx k} H\] 
\end{thmC}

In the above theorem, $G \not\equiv^{\monPCx k} H$ means that the 
monomial-PC (over $\mbQ$) can refute the system $\IsoForm(G,H)$ using 
degree at most $k$.
Obviously, this also implies that if $G \not\equiv^k H$, then $G 
\not\equiv^{\PCx k} H$, that is $\IsoForm(G,H)$ can be refuted in the 
full-PC with degree at most $k$. However, it remained open whether the 
converse holds as well (in particular, it remained open if the converse 
holds if we allow to increase the dimension for the Weisfeiler-Leman 
algorithm by a constant factor).

\begin{quC}[\cite{BerkholzGro15}]
\label{ques:monPCvsPC}
 Is there a function $f\colon \mbN \to \mbN$ such that for all $k \geq 2$ 
we have
\[ G \not\equiv^{\PCx k} H \,\, \Longrightarrow \,\, G \not\equiv^{f(k)} 
H?\]
\end{quC}

It immediately follows from Theorem~\ref{thm:PCinCinf} that the answer is 
affirmative and that we can choose $f$ to be linear.
\begin{thm}
\label{thm:monPCvsPC:GI}
 There is a linear function $f\colon \mbN \to \mbN$ such that for all $k 
\geq 2$ we have
\[ G \not\equiv^{\PCx k} H \,\, \Longrightarrow \,\, G \not\equiv^{f(k)} 
H. \]
\end{thm}
\begin{proof}
 Let $\mcI$ be an $\FO$-interpretation which interprets the 
$\IsoForm(G,H)$-formulas as polynomial systems over $\mbQ$ in pairs of 
graphs $(G,H)$. 
Let $r \geq 1$ be the number of variables in $\mcI$ and let $c\geq 1$ be a 
constant such that the number of variables in the $\Cinf$-formulas $\phi_k$,
that express the existence of $k$-dimensional PC-proofs according to 
Theorem~\ref{thm:PCinCinf}, is bounded by $c \cdot k$. 

We claim that $G \equiv^{r \cdot c \cdot k} H \,\, \Longrightarrow 
\,\, G \equiv^{\PCx k} H$. 
So let us assume that $G \equiv^{r \cdot c \cdot k} H$.
First of all it holds that $G \equiv^{r \cdot c \cdot k} H$ if, and only 
if, $(G,G) \equiv^{r\cdot c \cdot k} (G,H)$.
By the closure of $\Cinf$ under $\FO$-interpretations, it then follows that
$\IsoForm(G,G) \equiv^{c\cdot k} \IsoForm(G,H)$. Since $\IsoForm(G,G)$ is 
clearly satisfiable and since $\phi_k$ cannot distinguish between 
$\IsoForm(G,G)$ and $\IsoForm(G,H)$ it follows that there does not exist a
degree-$k$ $\PC$-refutation of $\IsoForm(G,H)$. Hence $G \equiv^{\PCx 
k} 
H$ as claimed.
\end{proof}

This shows that $\PC_k$ over $\bbQ$ as a graph distinguishing procedure is not substantially stronger than the $k$-dimensional Weisfeiler Leman test, and therefore, with respect to the graph isomorphism problem, $\monPC_{\Oo(k)}$ and $\PC_{\Oo(k)}$ are equally expressive. Generally speaking, though, $\PC_k$ and $\monPC_k$ differ in so far as  $\monPC_k$ proofs over $\bbQ$ can be found in polynomial time with the Gröbner basis algorithm (larger coefficients than in the input are never required), whereas for $\PC_k$, this is not the case. In light of Hakoniemi's exponential bit-complexity lower bound for $\PC_2$ \cite{Hakoniemi21}, it is in fact plausible that $\PC_k$ (over $\bbQ$) is simply not a polynomial-time proof system, and therefore in the general case strictly stronger than $\monPC_k$.

\subsection{Constraint Satisfaction Problems}
In this section we derive a dichotomy result for constraint 
satisfaction problems (CSPs) with respect to refutations in the polynomial 
calculus and the (weaker) resolution proof system.
Intuitively, what we are going to show is that each CSP either allows 
simple proofs of inconsistency, namely such proofs that can be derived in 
bounded-width resolution, or it requires proofs of very high complexity, 
that is of linear degree and exponential size, even in the much stronger 
polynomial calculus proof system. 

Let us recall the definition of CSPs.
We present the formulation as a homomorphism problem.
Let $\mfT$ be a fixed relational $\tau$-structure (the \emph{template}).
Then the \emph{constraint satisfaction problem} associated with $\mfT$ is 
the class $\Hom(\mfT)$ consisting of all $\tau$-structures $\mfA$ for 
which there exists an homomorphism $h\colon \mfA \to \mfT$.
Many combinatorial problems can be posed as CSPs. On the other hand, the 
class of all CSPs is limited in a certain sense: a famous conjecture 
by Feder and Vardi~\cite{FederV98}, which was recently confirmed 
independently by 
Bulatov~\cite{Bulatov17} and Zhuk~\cite{Zhuk17},
says that for each template $\mfT$ the problem 
$\Hom(\mfT)$ is either decidable in polynomial time (\PTIME) or complete 
for non-deterministic polynomial time (\NP-complete).
We will make use of a similar definability dichotomy for $\FPC$ soon. 

But before we do this, let us describe a simple algorithm to 
(approximately) solve constraint satisfaction problems.
This algorithm is known as the \emph{$k$-consistency test}
and it can be phrased as follows.
Fix a template $\mfT$ and consider an input structure $\mfA$.
Let us denote by $\Part^k(\mfA, \mfT)$ the set of all partial 
homomorphisms $p$ from $\mfA$ to $\mfT$ whose domain $\dom(p)$ is of size 
at most $k$ (we include the empty homomorphism $\emptyset$).
The idea is to iteratively compute restrictions $T_i \subseteq 
\Part^k(\mfA, \mfT)$ of $\Part^k(\mfA, \mfT)$ with respect to the 
following closure properties. We set $T_0 = \Part^k(\mfA, \mfT)$.
For $i \geq 1$ we set
\begin{align*}
 T_i &= \{ p \in T_{i-1} : \text{for all} \dom(p) \subseteq S 
\subseteq A \text{, } |S| \leq k \text{ th. ex. } q \in T_{i-1} \text{ 
s.th. } p \subseteq q, \dom(q) = S, \\
&\qquad\text{ and for all } q \subseteq p \text{ we have } q \in T_{i-1}\} 
\end{align*}
We output the final set $T_\infty$. In other words we iteratively 
eliminate all partial homomorphisms which cannot be extended to partial 
homomorphisms of size at most $k$ with respect to all possible 
(consistent) domains, and such partial homomorphisms for which we 
eliminated a restriction in the iteration before.
The first observation is that if there exists a homomorphism $h\colon \mfA 
\to \mfT$, then $T_\infty \neq \emptyset$, because the set of all 
restrictions of $h$ to partial homomorphisms in $\Part^k(\mfA, \mfT)$ will 
be contained in each $T_i$.
Hence, if $T_\infty = \emptyset$, then we can correctly conclude that 
$\mfA 
\not\in \Hom(\mfT)$.
If, on the other hand, $T_\infty \neq \emptyset$, then in the general case 
we must output ``we don't know''.
However, in many cases, depending on the template $\mfT$, this naive 
algorithm will work correctly on all inputs, which means that we have 
an efficient and simple way to decide the problem $\Hom(\mfT)$.

Before we come to this, let us observe that it is very easy to express the 
$k$-consistency test using bounded-width resolution.
For every $p \in \Part^k(\mfA, \mfT)$ consider a Boolean variable $X_p$ 
with the intended meaning that $X_p$ is true if $p \in T_\infty$.
According to the $k$-consistency test, we consider the following set of 
clauses:
\begin{itemize}
 \item For every $p \in \Part^k(\mfA, \mfT)$, every $\dom(p) \subseteq S 
\subseteq A$ with $|S| \leq k$:
\[ X_p \rightarrow \bigvee_{\begin{array}{c} \footnotesize
                             p \subseteq q \in \Part^k(\mfA, 
\mfT),  \\\dom(q) = S                            \end{array}
} X_q.\]
 \item For every $p \in \Part^k(\mfA, \mfT)$ and every $q \subseteq p$:
 \[ X_p \rightarrow X_q.\]
\end{itemize}
We obtain a (dual-)Horn-formula which always has the trivial model where 
we set every variable to false. Moreover, every non-trivial model is a 
witness that $T_\infty \neq \emptyset$. Since in this case $\emptyset \in 
T_\infty$, this means that if we add the single clause $X_{\emptyset}$ to 
the above formula, then we obtain a (dual-)Horn-formula which is not 
satisfiable if, and only if, $T_\infty = \emptyset$.
Moreover, note that since $\mfT$ is a fixed template, the above 
formula is of constant width. We conclude that $T_\infty = \emptyset$ if, 
and only if, bounded-width resolution can refute the above formula.
It is obvious that this formula is also interpretable in $\mfA$ using a 
first-order interpretation.

\begin{thm} For every $k \geq 1$, the $k$-consistency test can be 
expressed in  $\FO(\Res_{\mcO(k)})$.
\end{thm}

Now, what happens if we have a template $\mfT$ for which the 
$k$-consistency test is incomplete for any fixed value of $k \geq 1$?
In this case the (descriptive) complexity of the problem $\Hom(\mfT)$ 
is much higher. In fact, it follows from~\cite{AtseriasBulDaw09} 
and~\cite{BaKo14} that in this case, the problem cannot be defined in 
$\FPC$. This  ``definability dichotomy''  was first explicitly noted, and 
refined, by Dawar and Wang in~\cite{DaWa15} and in~\cite{DawarW17}:
\begin{thmC}[\cite{BaKo14,AtseriasBulDaw09,DaWa15,DawarW17}]
 For every template $\mfT$ one of the following is true.
 \begin{enumerate}
  \item Either there is a $k \geq 1$ such that the $k$-consistency 
test correctly decides $\Hom(\mfT)$, or
  \item there exists a (non-trivial) finite Abelian group $G$ such that 
the problem of deciding the solvability of linear equation systems over 
$G$ with at most three variables per equation, $\text{3LIN}(G)$, reduces 
to $\Hom(\mfT)$ via an $\FPC$-interpretation of linear size (that is the 
interpretation only increases the sizes of structures by a constant 
factor).
 \end{enumerate}
\end{thmC}

Moreover, it is known that the problem of deciding whether two 
CFI-structures $\mfA, \mfB \in \CFIclass{\mcF}{p}$ over the same 
underlying graphs are isomorphic reduces to $\text{3LIN}(\field F_p)$ via 
an $\FPC$-reduction of linear size. Hence, by applying 
Theorem~\ref{thm:appl:pclow} (and using the same arguments as in the 
Section~\ref{subsec:lowbounds:gi}) we get the following dichotomy for 
the proof systems resolution and polynomial calculus.
\begin{thm} For every template $\mfT$ one of the following holds.
\begin{enumerate}
 \item Either $\Hom(\mfT)$ can be decided using bounded-width resolution, 
or
 \item there exists a finite set of primes $P \subseteq 
\Primes$ such that for every linear-size $\FPC$-definable encoding of 
$\Hom(\mfT)$ as a system of polynomial equations $\mcP(\mfT)$ over a field 
$\field F$, which is either $\mbQ$ or a finite field with 
$\characteristic(\field F) \not\in P$, refutations of $\mcP(\mfT)$ in the 
polynomial calculus over $\field F$ require degree $\Omega(n)$ and size 
$2^{\Omega(n)}$ (where $n$ refers to the size of the input structures 
$\mfA$).
\end{enumerate}
\end{thm}

For the case of $\mbQ$, this dichotomy result has been 
established in~\cite{AtOc17} via a different proof strategy. Let us remark 
that, as a result of our approach, we can 
formulate our dichotomy result with respect to every $\FPC$-definable 
encoding (of linear size if we want to maintain exponential size lower 
bounds). Also, to the best of our knowledge, this dichotomy was not known  
for the case of the polynomial calculus over finite fields.

\section{Discussion: The Power of the Polynomial Calculus and Beyond}
\label{sec:discussion}
The resolution proof system and the polynomial calculus are two 
important and well-studied propositional proof systems.
In this article we characterised their power from the viewpoint of 
finite model theory. We proved that bounded-width resolution ($\kRes$, 
$k\geq 3$) is complete for existential fixed-point logic 
($\EFP$), that Horn-Resolution ($\HRes$) is complete for least 
fixed-point logic ($\LFP$), and that the bounded-degree monomial-PC 
($\monPCx{k}$) and the degree-$k$ polynomial calculus over $\bbQ$ with bit complexity $n^b$ ($\PC_{k,b}$ for $k \geq 2, b \geq 1$) over $\mbQ$ are complete for fixed-point 
logic with counting ($\FPC$) under (numerical) first-order reductions. Moreover, we showed that the degree-$k$ PC over $\bbQ$ without any restriction on the coefficients ($\PC_k$) can be expressed in $\Cinf$ with $\Oo(k)$ many variables. It remains open if $\PC_k$ can also be simulated in the weaker logic $\FPC$, or more generally, in $\ptime$. However, our result that $\FPC \equiv \FOnum(\PC_{k,b})$, and the fact that the proof system $\PC_{k,b}$ is strictly weaker than $\PC_k$, suggests that $\PC_k$ is really more powerful than $\FPC$.\\
Interestingly, our Theorem \ref{thm:monPCvsPC:GI} implies that for deciding the graph isomorphism problem in the polynomial calculus, using large coefficients in the refutations does not lead to additional power compared to the $k$-dimensional Weisfeiler Leman isomorphism test, which can be implemented in $\FPC$. This raises the question what precisely are the problems for which large coefficients in refutations actually take the power of the proof system beyond that of $\FPC$ and $\PC_{k,b}$.\\

Our method that takes definability as the measure for 
expressive power yields a much finer classification compared 
to the one that we get by using standard complexity-theoretic 
notions.
Indeed, it is well-known that already $\kResx 3$ is $\PTIME$-complete, 
which means that \emph{all} (fragments of) proof systems that we 
considered here are 
equivalent from the viewpoint of (algorithmic) complexity theory.
In contrast, as we saw, we obtain a more interesting landscape 
if we measure their descriptive complexity instead.

On the other hand, compared to the view of proof complexity, our 
analysis is much coarser. For instance, in our framework there is no 
explicit difference between width-$3$ and width-$k$ resolution for any $k 
\geq 4$, while, from the viewpoint of proof complexity, clearly these 
systems have different power.
The reason for this mismatch is that we allow more powerful logical
reductions (which are still weak from the 
viewpoint of finite model theory).
We believe that this more general perspective, though not as 
precise, makes it easier to pin down fundamental differences between, and 
weaknesses of, the different (fragments) of proof systems.
For instance, our results show that the resolution proof system cannot 
refute the Pigeonhole Principle for any \FO-definable encoding, 
see~\cite{GrPaPa17}, while the polynomial calculus over $\mbQ$ allows 
simple refutations (with respect to a natural encoding). Our results 
explain this 
``counting dichotomy'' very clearly: resolution corresponds to $\EFP$, a 
logic which lacks counting, and the polynomial calculus over $\mbQ$ to 
$\FPC$/$\Cinf$, logics 
which explicitly include a counting ability.
Moreover, our results highlight that the polynomial calculus has a 
severe weakness: it is not able to go beyond $\FPC$ (with respect to its 
bounded-degree and bounded bit-complexity $\PTIME$-stratification). 
Since it is known that $\FPC$ fails to express all 
$\PTIME$-properties, this implies that there are certain 
$\PTIME$-properties which do not have small 
refutations in the polynomial calculus. The prototype example is solving 
linear equation systems over finite fields.
We can exploit this connection between $\FPC$ and the polynomial calculus 
over $\mbQ$ even further to derive yet another characterisation. Indeed, 
we can show that linear programming is complete for $\FPC$ under 
(numerical) first-order reductions. This means that the power of the 
polynomial calculus over $\mbQ$ corresponds \emph{precisely} to the power 
of linear programming under numerical $\FO$-reductions. This connects the 
polynomial calculus with a very natural and significant algorithmic 
problem 
in the setting of finite model theory.

\medskip
Another interesting outcome of our work are the new finite-model theoretic 
proofs for lower bounds on the complexity of refutations in the polynomial 
calculus.
We saw that, using a uniform finite-model theoretic approach, one can show 
that many families of propositional formulas require refutations of 
exponential size.
Remarkably, we could obtain these lower bounds not only for the polynomial 
calculus over $\mbQ$, but also for the polynomial calculus over 
 finite fields. 
Also, as a result of our approach, our lower bounds are very robust 
in the sense that they do not rely on any specific encoding of 
a problem as a propositional formula, but they hold with respect to any 
($\FPC$-)definable encoding of the problem. 
For the case of the polynomial calculus this implies, for example, that 
all of the aforementioned lower bounds also hold for the \emph{polynomial 
calculus with resolution} (PCR). This proof system is nothing more than 
the polynomial calculus, but we include for any variable $X$ a syntactic 
dual variable $\bar X$ together with the axiom $1 - X = \bar X$. Clearly 
these additional axioms can be defined in $\FPC$, and so, our results do 
not 
change in any way by considering the PCR instead of the standard PC.

\medskip
Let us finally take a look at some future work.
We observe that in our lower bound proofs for the 
polynomial calculus over finite fields we do not require a precise 
connection with $\FPC$-definability (in fact, as we saw, such a precise 
match between $\FPC$ and the polynomial calculus over finite fields does 
not exist).
Indeed, for proving lower bounds it was sufficient to establish
$\FPC$-definability of refutations for families of propositional formulas 
that are defined in CFI-structures. We then made use of the fact that 
the CFI-problem is hard for $\FPC$ which gave us the lower bounds on the 
proof complexity.
Even more general, we do not need to obtain $\FPC$-definability, but, 
because of the fact that the CFI-problem is hard already for
finite-variable counting logic $\Cinf$, 
 it is sufficient to show
$\Cinf$-definability (recall that $\Cinf$ is a more powerful logic 
than $\FPC$, so showing definability is easier).
We followed these lines for the case of the polynomial calculus 
over finite fields in 
Section~\ref{sec:PCfiniteFields}. For this, we strongly made use of our
key technical results which says that CFI-structures over expander graphs 
are 
$\FPC$-homogeneous.
Recall that this means that we can order orbits of $k$-tuples in 
CFI-structures using $\FPC$-formulas with a linear number of 
variables only.

In fact, we can use this homogeneity result to develop a much more general 
strategy for proving lower bounds for certain propositional proof system 
$\Prop$.
As we explain in the following, in certain situations this result 
allows us to \emph{quantify over refutations in $\Cinf$}.
More precisely, assume that $\Prop$ has a stratification $\Prop=(\Prop_k)$ 
along a 
parameter $k \geq 1$.
Moreover, assume that whenever a family of propositional formulas 
$\mcF$ 
that is defined (via a fixed $\FPC$-interpretation) in (pairs of) 
CFI-structures, has a refutation in $\Prop_k$, then 
it also has a refutation $\mfp$ such that:
\begin{itemize}
 \item $\mfp$ is symmetric, that is invariant under all automorphisms of 
the underlying CFI-structures, and
 \item $\mfp$ can be encoded as an object that is definable in the logic 
  $C^{\omega}_{\infty\omega}$ over the underlying CFI-structures with
$\mcO(k)$ variables, 
 \item given a description of $\mfp$ as above, it can be 
verified using a $\Cinf$-formula with $\mcO(k)$ many variables, that 
$\mfp$ refutes $\mcF$.
\end{itemize}
If these (vaguely formulated) conditions are satisfied, then we can 
basically apply our techniques in order to show that certain families of 
propositional formulas, namely such formulas 
which encode the CFI-isomorphism problem, cannot be refuted in $\Prop_k$ 
for any sublinear $k$.
At the moment, we work out the details and study to what extent these 
conditions can be relaxed. 

For now, let us illustrate the usefulness of this approach by 
means of a simple example. 
If we take another look at the paper by Grohe and 
Berkholz~\cite{BerkholzGro15}, then we observe that they do, in fact, not 
only derive lower bounds on the complexity of refutations for the graph 
isomorphism problem for the polynomial calculus over $\mbQ$, but also for 
a 
stronger proof system which is known as the \emph{Positivstellensatz} (or 
\emph{Sums-of-Squares Proof System}).
Let us briefly introduce this system. The setting is the same as for the 
polynomial calculus over $\mbQ$, that is our input is a set $\mcP$ 
consisting of multivariate polynomials $p \in \mbR[\mcX]$, and our aim is 
to show that the polynomials in $\mcP$ do not have a common zero. As 
before we implicitly assume that the Boolean constraints $X^2 - X = 0$ are 
contained in $\mcP$ for every variable $X \in \mcX$.

Let us fix a degree $k \geq 2$ which is even. A degree-$k$ 
Positivstellensatz refutation of a polynomial equation system $\mcP$ over 
variables $\mcX$ consists of polynomials $f_p 
\in \mbR[\mcX]$ such that 
\[ \sum_{p \in \mcP} f_p \cdot p = 1 + s, \]
where $s$ is a sum-of-squares (sos) polynomial, that is $s = \sum_{i \in 
I} q_i^2$ for some polynomials $q_i \in \mbR[\mcX]$, and such that all 
polynomials in the above equation have degree at most $k$. Since $s(a) 
\geq 0$ for every evaluation $a\colon \mcX \to \mbR$, the existence of 
such 
a refutation clearly proves that $\mcP$ is inconsistent.
Now, as in our description above, assume that we have interpreted this 
system in a (pair) of CFI-structures, and let $\Gamma$ be the 
corresponding CFI-automorphism group. Every $\pi \in \Gamma$ extends 
(uniquely) to a permutation on $\mcX$ and so it defines a unique 
automorphism of $\mbR[\mcX]$. Moreover, this automorphism of $\mbR[\mcX]$ 
stabilises
$\mcP$. It follows that if we have a refutation as above, also 
\[ \sum_{p \in \mcP} \pi(f_p) \cdot \pi(p) = 1 + \pi(s),\]
is a refutation. Here we are just saying that refutations are mapped 
to refutations if we permute the variables in such a way that the set of 
given polynomials remains stable. Clearly, this holds for any 
reasonable proof system. In particular, note that $\pi(s)$ is also a 
sum-of-squares polynomial (because $\pi$ is an automorphism of 
$\mbR[\mcX]$).

However, in the case of the Positivstellensatz we can go one important 
step further by summing up over all refutations that we obtain in this way:
\begin{align*}
\sum_{\pi \in \Gamma} \left(\sum_{p \in \mcP} \pi(f_p) \cdot \pi(p)\right) 
= |\Gamma| + \sum_{\pi \in \Gamma} \pi(s).
\end{align*}
The importance of this equation follows from the fact that the sos 
polynomial on the right-hand side is symmetric with respect to $\Gamma$.
The simple consequence is that whenever we can derive from 
$\mcP$, via a degree-$k$ combination of polynomials, a polynomial $1 + s$, 
where $s$ is an sos polynomial, then we can also derive from $\mcP$ a 
polynomial $1 + \hat s$ where $\hat s$ is a \emph{symmetric} sos 
polynomial (note that sos polynomials are closed under addition).

This already brings us very close to our proof strategy from above: we saw 
that whenever there is a degree-$k$ refutation, there is also a symmetric 
one.
Let us now complete our argument for the case of the Positivstellensatz 
more explicitly.
The most important question is how we can obtain the symmetric 
polynomial $1 + \hat{s}$ in $\Cinf$.
The key insight is that we don't have to bother too much about this, 
because $1 + \hat{s}$ is symmetric.
Clearly, we can describe $r= 1 + \hat{s}$ as a mapping $r\colon M_k \to 
\mbR$ where $M_k$ denotes the set of all monomials of degree at most $k$.
Since $r$ is symmetric, $r$ is a vector with the same entries on all 
$\Gamma$-orbits on $M_k$. We now make use of the fact that CFI-structures 
are $\FPC$-homogeneous. This allows us to order the $\Gamma$-orbits on 
$M_k$ in $\FPC$ using only $\mcO(k)$ many variables.
Using this we can see that we can describe the vector $r$ by using a 
mapping from an \emph{ordered set} to~$\mbR$. This is a quite simple object 
from the viewpoint of $\Cinf$ as it has nothing to do with the 
underlying structure. In particular, we can explicitly quantify 
over all such mappings, since we have infinite conjunctions and 
disjunctions available in $\Cinf$. 
The final step is to verify that, having guessed such a vector $r\colon M_k 
\to \mbR$ in $\Cinf$, this vector is indeed a refutation, that is $r = 1 + 
\hat{s}$ for a symmetric sos-polynomial $\hat s$, and that $r$ can be 
derived from $\mcP$ using a degree-$k$ polynomial combination.
The latter problem is about solving a linear equation system over $\mbR$ 
which can be done $\Cinf$ by what we saw in Section~\ref{sec:fpc} 
(it is not hard to see that dealing with real numbers in this 
context is easy: since we are working in $\Cinf$ and not in $\FPC$, we 
can quantify explicitly over (sets of) real numbers that we can use for 
our definitions).

The former problem can be reformulated as follows.
Let $M_{k/2}$ denote the set of monomials over $\mcX$ of degree at most 
$k/2$.
Let $S$ be the $M_{k/2} \times M_{k/2}$-matrix over $\mbR$ which is 
defined by letting $S(m,n)$ be the leading coefficient of the monomial 
$m \cdot n$, $m,n \in M_{k/2}$, in $\hat s$ that we get when we 
syntactically expand the sos polynomial $\hat s$. Then $S$ is symmetric 
and, as a consequence of 
the syntactic form of $\hat s$ ($\hat s$ is an sos polynomial), $S$ can 
be written as a sum of matrices $vv^T$ where the $v\colon M_{k/2} \to 
\mbR$ correspond to the summands in $\hat{s}$. Vice versa, assume 
that $S$ can be written in this form. Let $z$ be the $M_{k/2}$-vector 
whose entries are the monomials $m \in M_{k/2}$, i.e.\ $z(m) = m$. Then it 
is easy to see that $z^T S z$ is an sos polynomial. Hence, $\hat{s}$ is an 
sos polynomial if, and only if, the corresponding matrix $S$ can be 
written as a sum of matrices $vv^T$ for $v\colon M_{k/2} \to \mbR$.
This condition is equivalent to saying that $S$ is positive semi-definite, 
which, in turn, is equivalent to saying that $S$ has only non-negative 
eigenvalues.
It is known that the eigenvalues of matrices over $\mbQ$ are definable in 
$\FPC$, see~\cite{DawarGroHolLau09}. It is easy to adapt this definability 
result to our setting which shows that the positive semi-definiteness of 
$S$ can be certified in $\Cinf$ (using $\mcO(k)$ many variables) as well.

This proof (sketch) shows that all lower bounds for the 
polynomial calculus that we obtained in 
Section~\ref{sec:applications:low}, 
that is for graph isomorphism refutations and for the CSP dichotomy,
remain valid for the Positivstellensatz. These lower bounds have 
been known before, but it is nice to see how easily they can be derived by 
using our newly developed finite-model-theoretic tools.
Again, let us stress that what makes our arguments particularly simple 
is the $\FPC$-homogeneity of CFI-structures.
As we saw, this result allows us to quantify over refutations in 
$\Cinf$ (assuming that symmetric refutations with certain syntactic 
properties exist), so we are only left with the usually much simpler task 
of verifying such refutations in $\Cinf$.
As indicated above, this line of research is part of on ongoing project 
where we explore the power of symmetric proof systems from the viewpoint 
of finite model theory more thoroughly, so we defer the details to this 
upcoming work.\\
\\
\\
\textbf{Acknowledgements:} We would like to thank Joanna Ochremiak for drawing our attention to a mistake in the previous version of this paper, and Tuomas Hakoniemi for answering detailed questions on the issue of bit-complexity in the polynomial calculus over $\bbQ$.


\bibliographystyle{abbrv}
\bibliography{references}

\end{document}